\documentclass[11pt,11pt,letterpaper,reqno]{amsart}
\usepackage[utf8]{inputenc}
\usepackage{mathrsfs}
\usepackage{enumitem}
\usepackage{amstext}
\usepackage{amsthm}
\usepackage{amssymb}
\usepackage{graphicx}
\usepackage{geometry}
\PassOptionsToPackage{normalem}{ulem}
\usepackage{ulem}
\usepackage[bookmarks=false,
 breaklinks=false,pdfborder={0 0 1},backref=false,colorlinks=false]
 {hyperref}

\makeatletter
\numberwithin{equation}{section}
\numberwithin{figure}{section}

\usepackage{amscd}
\usepackage{latexsym}
\usepackage{amscd}
\usepackage[mathscr]{euscript}
\usepackage{mathrsfs}
\usepackage{xcolor}
\usepackage{color}

\usepackage{lscape}
\usepackage{epstopdf}
\usepackage{array}
\usepackage{tabularx}
\usepackage{longtable}
\usepackage{ltxtable}
\usepackage{color}

\usepackage{tikz}
\usetikzlibrary{decorations,arrows}
\usetikzlibrary{decorations.pathmorphing}
\usepgflibrary{decorations.pathreplacing}
\usepackage{pgf}
\usetikzlibrary{patterns}
\tikzset{
  schraffiert/.style={pattern=horizontal lines,pattern color=#1},
  schraffiert/.default=black
}
\tikzset{
    ultra thin/.style= {line width=0.1pt},
    very thin/.style=  {line width=0.2pt},
    thin/.style=       {line width=0.4pt},
    semithick/.style=  {line width=0.6pt},
    thick/.style=      {line width=0.8pt},
    very thick/.style= {line width=1.2pt},
    ultra thick/.style={line width=2.4pt}
}
 \usetikzlibrary{shapes}
\definecolor{hellgrau}{rgb}{0.93,0.93,0.93}
\definecolor{hellergrau}{rgb}{0.97,0.97,0.97}
\definecolor{hellgruen}{rgb}{0.6,1.35,0.5}
\definecolor{grau}{rgb}{0.93,0.93,0.93}
\definecolor{hellblau}{rgb}{0.8,0.8,2.0}
\definecolor{blau}{rgb}{0.3,0.5,2.0}
\definecolor{hellrot}{rgb}{2.0,0.6,0.6}
\definecolor{gruen}{rgb}{0.3,0.75,0.2}
\definecolor{rot}{rgb}{0.9,0.1,0.1}
\definecolor{orang}{rgb}{1.3,0.65,0}
\allowdisplaybreaks
\DeclareMathAlphabet{\mathpzc}{OT1}{pzc}{m}{it}
\newcommand{\Aa}{{\mathcal A}}

\newcommand{\precd}{\prec\!\!\prec}					
\newcommand{\strict}{\subseteq_{\text{strict}}} 	
\newcommand{\ind}{\mbox{ind}}						
\newcommand{\indB}{\mbox{ind}_B}						
\newcommand{\motw}{=_2}

\allowdisplaybreaks
\usepackage{xcolor}
\newcommand{\Hmm}[1]{\leavevmode{\marginpar{\tiny%
$\hbox to 0mm{\hspace*{-0.5mm}$\leftarrow$\hss}%
\vcenter{\vrule depth 0.1mm height 0.1mm width \the\marginparwidth}%
\hbox to 0mm{\hss$\rightarrow$\hspace*{-0.5mm}}$\\\relax\raggedright #1}}}
\definecolor{green}{RGB}{0, 180, 0}
\definecolor{cyan}{RGB}{0, 180, 180}
\definecolor{yellow}{RGB}{211,211,0}

\renewcommand{\subset}{\subseteq}

\makeatother

\theoremstyle{plain}
\newtheorem{thm}{\protect\theoremname}[section]
\theoremstyle{remark}
\newtheorem*{rem*}{\protect\remarkname}
\theoremstyle{plain}
\newtheorem{prop}[thm]{\protect\propositionname}
\theoremstyle{definition}
\newtheorem{defn}[thm]{\protect\definitionname}
\newtheorem{example}[thm]{\protect\examplename}
\theoremstyle{plain}
\newtheorem{lem}[thm]{\protect\lemmaname}
\newtheorem{cor}[thm]{\protect\corollaryname}
\theoremstyle{remark}
\newtheorem{rem}[thm]{\protect\remarkname}
\providecommand{\corollaryname}{Corollary}
\providecommand{\definitionname}{Definition}
\providecommand{\examplename}{Example}
\providecommand{\lemmaname}{Lemma}
\providecommand{\propositionname}{Proposition}
\providecommand{\remarkname}{Remark}
\providecommand{\theoremname}{Theorem}

\begin{document}

\global\long\def\theenumi{\alph{enumi}}%

\global\long\def\ui{\mathbf{\textrm{i}}}%

\global\long\def\ue{\mathbf{\textrm{e}}}%

\global\long\def\ud{\mathbf{\textrm{d}}}%

\global\long\def\sgn{\mathrm{sign}}%

\global\long\def\id{\mathbf{1}}%

\global\long\def\C{\mathbb{C}}%

\global\long\def\R{\mathbb{R}}%

\global\long\def\Q{\mathbb{Q}}%

\global\long\def\N{\mathbb{N}}%

\global\long\def\Z{\mathbb{Z}}%

\global\long\def\Id{\mathbf{\mathbb{I}}}%

\global\long\def\V{\mathcal{V}}%


\global\long\def\Aa{\mathcal{A}}%

\global\long\def\tree{\mathcal{T}_{\alpha}}%

\global\long\def\graph{\mathcal{G}_{\alpha}}%

\global\long\def\paths{\Gamma_{\alpha}}%

\global\long\def\emap{E_{\alpha}}%

\global\long\def\Cen{Z}%

\global\long\def\co{\textbf{c}}%

\global\long\def\cop{\textbf{\ensuremath{\widetilde{\co}}}}%

\global\long\def\cm{[\textbf{c},m]}%

\global\long\def\cmn{[\textbf{c},m,n]}%

\global\long\def\Co{\mathscr{C}}%

\global\long\def\emp{\emptyset}%

\global\long\def\tr{\textit{tr}}%

\global\long\def\Nz{\N_{0}}%

\global\long\def\Nmo{\N_{-1}}%


\global\long\def\precd{\prec_{\mathrm{str}}}%

\global\long\def\strict{\subseteq_{\mathrm{str}}}%

\global\long\def\Mult{\textrm{\ensuremath{\mathcal{M}}}}%

\global\long\def\relA{\mbox{ind}^{A}_{\textrm{rel}}}%

\global\long\def\ind{\mbox{ind}}%

\global\long\def\indA{\mbox{ind}_{A}}%

\global\long\def\indB{\mbox{ind}_{B}}%

\global\long\def\indC{\mbox{ind}_{C}}%

\global\long\def\dR{\delta_{R}}%

\global\long\def\motw{=_{2}}%

\global\long\def\spec{\mathrm{spec}}%

\global\long\def\tA{\mathit{A}}%

\global\long\def\tB{\mathit{B}}%

\global\long\def\tC{\mathit{G}}%

\global\long\def\sigc{\sigma_{\co}}%

\global\long\def\sigcm{\sigma_{[\co,m]}}%

\global\long\def\sigcmn{\sigma_{[\co,m,n]}}%

\global\long\def\alk{\alpha_{k}}%

\global\long\def\Ham{H_{\alpha,V}}%

\global\long\def\Hrat{H_{\frac{p}{q},V}}%

\global\long\def\Halk{H_{\alk,V}}%

\global\long\def\nHcV{H^{\times n}_{\co,V}}%

\global\long\def\nHc{H^{\times n}_{\co}}%

\global\long\def\nHcm{H^{\times n}_{[\co,m]}}%

\global\long\def\nHcmV{H^{\times n}_{[\co,m],V}}%

\global\long\def\Sturm{\omega_{\alpha}}%

\global\long\def\IDS{N_{\alpha,V}}%

\global\long\def\Hc{H_{\co}}%

\global\long\def\HcV{H_{\co,V}}%

\global\long\def\Hcm{H_{[\co,m]}}%

\global\long\def\HcmV{H_{[\co,m],V}}%

\global\long\def\Hcmn{H_{[\co,m,n]}}%

\global\long\def\HcmnV{H_{[\co,m,n],V}}%

\global\long\def\Hcmo{H_{[\co,m,1]}}%

\global\long\def\Hcz{H_{[\co,0]}}%

\global\long\def\Hco{H_{[\co,1]}}%

\global\long\def\thc{\theta_{\co}}%

\global\long\def\thcm{\theta_{[\co,m]}}%

\global\long\def\thcmn{\theta_{[\co,m,n]}}%

\global\long\def\pc{p_{\co}}%

\global\long\def\pcm{p_{[\co,m]}}%

\global\long\def\pcmn{p_{[\co,m,n]}}%

\global\long\def\qc{q_{\co}}%

\global\long\def\qcm{q_{[\co,m]}}%

\global\long\def\qcmn{q_{[\co,m,n]}}%

\global\long\def\tc{t_{\co}}%

\global\long\def\tcm{t_{[\co,m]}}%

\global\long\def\tcmn{t_{[\co,m,n]}}%

\global\long\def\tco{t_{[\co,1]}}%

\global\long\def\tcmo{t_{[\co,-1]}}%

\global\long\def\tcz{t_{[\co,0]}}%

\global\long\def\lc{\lambda_{\co}}%

\global\long\def\lcm{\lambda_{[\co,m]}}%

\global\long\def\lcmn{\lambda_{[\co,m,n]}}%

\global\long\def\lo{\lambda_{\mathbf{o}}}%

\global\long\def\mo{\mu_{\mathbf{o}}}%

\global\long\def\Nc{N_{\co}}%

\global\long\def\Ncm{N_{[\co,m]}}%

\global\long\def\Ncmn{N_{[\co,m,n]}}%

\global\long\def\Ic{I_{\co}}%

\global\long\def\Ico{I^{1}_{[\co,1]}}%

\global\long\def\Icmno{I^{1}_{[\co,m,n]}}%

\global\long\def\IcmnN{I^{M+1}_{[\co,m,n]}}%

\global\long\def\Icz{I_{[\co,0]}}%

\global\long\def\Jcz{J_{[\co,0]}}%

\global\long\def\Kcz{K_{[\co,0]}}%

\global\long\def\cz{[\co,0]}%

\global\long\def\sigcz{\sigma_{[\co,0]}}%

\global\long\def\Jcm{J_{[\co,m]}}%

\global\long\def\Kcm{K_{[\co,m]}}%

\global\long\def\Icm{I_{[\co,m]}}%

\global\long\def\Icmn{I_{[\co,m,n]}}%

\global\long\def\Icmo{I_{[\co,m,1]}}%

\global\long\def\Icmoo{I^{1}_{[\co,m,1]}}%

\global\long\def\Icmoi{I^{i}_{[\co,m,1]}}%

\global\long\def\Icmi{I^{i}_{[\co,m]}}%

\global\long\def\Icmni{I^{i}_{[\co,m,n]}}%

\global\long\def\Icmnj{I^{j}_{[\co,m,n]}}%

\global\long\def\Icmii{I^{i+1}_{[\co,m]}}%

\global\long\def\Icmnii{I^{i+1}_{[\co,m,n]}}%

\global\long\def\ohn{\omega_{\frac{p_{k}}{q_{k}}}}%

\global\long\def\ohnm{\omega_{\frac{p_{k-1}}{q_{k-1}}}}%

\global\long\def\ohnp{\omega_{\frac{p_{k+1}}{q_{k+1}}}}%

\global\long\def\fl#1{\emph{\ensuremath{\left\lfloor #1\right\rfloor }}}%

\global\long\def\set#1#2{\left\{  #1\thinspace:\thinspace#2\right\}  }%

\global\long\def\newmacroname{\{\}}%

\global\long\def\crit{V_{\mathrm{crit}}}%

\global\long\def\crito{V^{quasi}_{\mathrm{crit}}}%

\global\long\def\TmV{P}%

\global\long\def\Tquas{P^{quasi}}%

\title{Complete hierarchical structure of the spectral bands in the Kohmoto model}
\author{Ram Band, Siegfried Beckus, Raphael Loewy}
\address{Department of Mathematics\\
Technion - Israel Institute of Technology\\
Haifa, Israel}
\email{ramband@technion.ac.il}
\address{Institute of Mathematics\\
University of Potsdam\\
Potsdam, Germany}
\email{beckus@uni-potsdam.de}
\address{Department of Mathematics\\
 Technion - Israel Institute of Technology\\
 Haifa, Israel}
\email{loewy@technion.ac.il}
\begin{abstract}
We study the Kohmoto model, a family of discrete Schr\"odinger operators with Sturmian potentials depending on a frequency and a coupling constant.
We prove that, for all non-vanishing coupling constants, all spectral bands admit a hierarchical structure. 
This structure offers
a variety of applications, including a detailed description of the
Kohmoto butterfly and a central step towards the resolution of the
dry ten Martini problem for Sturmian Hamiltonians, which we carry
out in a subsequent work.
\end{abstract}

\maketitle

\tableofcontents{}

\section{Introduction and main results \label{sec: introduction and main results}}

For $\alpha\in[0,1]$ and $V\in\R$, consider the self-adjoint operator
$H_{\alpha,V}:\ell^{2}(\Z)\to\ell^{2}(\Z)$ defined by 
\begin{align}
(\Ham\psi)(n) & :=\psi(n+1)+\psi(n-1)+V\chi_{\left[1-\alpha,1\right)}(n\alpha\mod 1)\thinspace\psi(n),\label{eq: Hamiltonian defined}
\end{align}
where $\chi_{[1-\alpha,1)}$ is the characteristic function of the
interval $[1-\alpha,1)$, $V\in\R$ is the strength of the potential,
called the \emph{coupling constant}, and $\alpha$ is called the \emph{frequency}.
Here, $n\alpha\mod 1$ is the fractional part of $n\alpha$ in $[0,1)$.

When $\alpha\notin\Q$, this operator $\Ham$ is called a \emph{Sturmian
Hamiltonian}, since the sequence 
\[
\omega_{\alpha}\in\{0,1\}^{\Z},\qquad\omega_{\alpha}(n):=\chi_{\left[1-\alpha,1\right)}(n\alpha\mod 1),\:n\in\Z,
\]
is called a \emph{Sturmian sequence}. When $\alpha=\frac{p}{q}\in\Q$,
the sequence $\omega_{\alpha}\in\{0,1\}^{\Z}$ is $q$-periodic, that
is, $\omega_{\alpha}(n)=\omega_{\alpha}(n+q)$ for all $n\in\Z$.
The associated operator $\Hrat$ is periodic and assuming that $V\neq0$
and $p$ and $q$ are coprime, its spectrum, $\sigma(H_{\frac{p}{q},V})$,
consists of exactly $q$ closed intervals, which are called \emph{spectral
bands}; see, e.g., \cite{Tes00,Raym95,DaFi24-book_2,BaBeBiTh22}.
These periodic operators may also serve as approximations of Sturmian
Hamiltonians with irrational frequencies $\alpha\in[0,1]\backslash\Q$
\cite{Raym95,BIST89,DaFi24-book_2,BaBeLo_DTMP26}. The family of all
operators $(\Ham)_{\alpha\in[0,1]}$ as $\alpha$ ranges over $\left[0,1\right]$
is called the \emph{Kohmoto model} \cite{KohKadTan_prl83,OstKim_phys85}.
Plotting the spectra $\sigma(H_{\alpha,V})$ for various values of
$\alpha\in\Q$ gives rise to the Kohmoto butterfly; see Figure~\ref{fig: Initial Spectra}.
This structure exhibits a striking self-similar and fractal nature,
attracting interest from both mathematicians and physicists.

We start by presenting some notations which are needed to state our
main theorem. Finite continued fraction expansions are commonly used
to represent rational numbers via
\[
c_{0}+\frac{1}{c_{1}+\frac{1}{\ddots+\frac{1}{c_{k}}}}\in\Q.
\]

We extend the standard definitions and introduce the space of \emph{augmented
finite continued fraction expansions,}
\[
\Co:=\left\{ [0],~[0,0]\right\} \cup\bigcup_{k\in\N}\set{[0,0,c_{1},\ldots,c_{k}]}{c_{1},\ldots,c_{k-1}\in\N,~c_{k}\in\N_{-1}},
\]
where $\Nmo:=\N\cup\{-1,0\}$. This notation uses the convention that
the two first entries of all $\co\in\Co$, satisfy $c_{-1}=c_{0}=0$.
The connection between the finite continued fraction expansions and
the rational numbers is as follows.

The \emph{evaluation map} $\varphi:\Co\to\R\cup\left\{ \infty\right\} $
is defined for all $\co=[0,c_{0},c_{1},\ldots,c_{k}]\in\Co\backslash\left\{ [0]\right\} $
by
\begin{equation}
\varphi([0,c_{0},c_{1},\dots,c_{k}]):=\begin{cases}
\varphi([0,c_{0},c_{1},\dots,c_{k-2}]), & k\in\N\textrm{ and }c_{k}=0,\\[0.1cm]
c_{0}+\frac{1}{c_{1}+\frac{1}{\ddots+\frac{1}{c_{k}}}}, & \text{otherwise},
\end{cases}\label{eq: phi map from C to Q}
\end{equation}
and $\varphi([0]):=\infty$.

We use the $\Co$-space to show that there is a hierarchical structure
which involves both nesting and interlacing of spectral bands. In
particular, we define two types ($\tA$ and $\tB$) of spectral bands
(Definition~\ref{def: A B types}). In order not to delay the presentation
of the main theorem, we defer the introduction of spectral bands types
and the related hierarchical structure to Section~\ref{sec: Spectral Types}.
We only mention here that the definition of these types depends on
$\co\in\Co$, rather than on the rational value $\varphi(\co)$; noting
that there exist $\co,\tilde{\co}\in\Co$ with $\varphi(\co)=\varphi(\tilde{\co})$,
see a discussion in Section~\ref{sec: Spectral Types}. Our main
theorem states the following.
\begin{thm}
\label{thm: Every band is A or B} For all $V\neq0$ and $\co\in\Co$
with $\varphi(\co)\in[0,1]\cap\Q$, every spectral band in $\sigma\left(H_{\varphi(\co),V}\right)$
is either of type $A$ or $B$ and its type is independent of the
value of $V>0$ respectively $V<0$.
\end{thm}

This theorem forms a central step towards our resolution of the dry
ten Martini problem \cite{BaBeLo_DTMP26}. In the large coupling regime
$V>4$, an analogue of Theorem~\ref{thm: Every band is A or B} was
proven by Raymond \cite{Raym95} in a slightly different terminology;
see the review \cite{BaBeBiTh22}, which is adapted to the present
framework. The hierarchical structure which that result provided was
a powerful tool for estimating the fractal dimensions in the large
coupling regime; see \cite{KiKiLa03,Damanik2008,Liu2014,Damanik2015,CaoQu_arXiv23}.
Hence, the result above for all $V\neq0$ might enhance the analysis
of Kohmoto model's fractal dimensions. Moreover, we expect that this
structure provides new insights into the self-similarity of the Kohmoto
butterfly in terms of number-theoretic properties of $\alpha$, such
as its continued fraction expansion; see \cite{BecBelTho25} for progress
in this direction.

It is not trivial to extend the original result of Raymond \cite{Raym95}
from the large coupling regime to $0<V\leq4$ is more delicate. The
difficulty arises since spectral bands start to overlap and classical
approaches based on trace maps are no longer sufficient to control
the spectral band relative structure. We overcome this challenge by
\begin{itemize}
\item studying the whole $\Co$-space of finite continued fractions expansions
simultaneously, see Section~\ref{sec: Proof_AB_Type}.
\item establishing an interlacing theorem for the Floquet--Bloch matrices,
see Section~\ref{subsec: Perturbation argument}. For this sake,
a new concept of eigenvalue admissibility plays a central role.
\item using a uniform Lipschitz bound in $V$ of the spectral bands, see
Section~\ref{sec: Spectral Types}.
\end{itemize}
The paper is organized as follows. Section~\ref{sec: Spectral Types}
introduces the formal definitions of spectral bands of type $\tA$
and $\tB$. Section~\ref{sec: Proof_AB_Type} contains the proof
of the main theorem via induction over the $\Co$-space, subject to
two key ingredients: that the backward type implies the forward type,
and the induction base. The necessary spectral descriptions via Floquet-Bloch
matrices and transfer matrix traces are developed in Section~\ref{sec: Basic spectral analysis},
where we also show a symmetry of the spectrum -- reducing the study
to $V>0$ -- and state an interlacing theorem, whose proof is deferred
to the Appendix~\ref{App: Perturbation argument}. The assumptions
of this theorem are then verified in Section~\ref{sec: Admissibility. index relations},
where we also develop the necessary tools and the partial results
which support the proof of the main theorem. Sections~\ref{subsec:Backward-implies-forward}
and \ref{sec: Pf_InductionBase} are devoted to proving these two
key ingredients: Section~\ref{subsec:Backward-implies-forward} shows
that the backward type implies the forward type, while Section~\ref{sec: Pf_InductionBase}
establishes the induction base. Basic combinatorial properties of
Sturmian sequences are collected in Appendix~\ref{App: Sturmian dynamical systems},
while standard techniques for trace identities are deferred to Appendix~\ref{App: Trace Maps}.

\section{The $\protect\tA$/$\protect\tB$ type of spectral bands\label{sec: Spectral Types}}

We start by returning to the definition of the $\Co$-space and further
exploring it.

~From (\ref{eq: phi map from C to Q}) follows that for $k\geq1$,
\begin{align*}
\varphi([0,c_{0},c_{1},\dots,c_{k},-1]) & =\varphi([0,c_{0},c_{1},\dots,c_{k}-1])\quad\text{and}\\
\varphi([0,c_{0},\dots,c_{k-1},0]) & =\varphi([0,c_{0},\dots,c_{k-2}]).
\end{align*}
The second identity may be intuitively understood if one allows $c_{k}$
to take real values and then consider the limit $c_{k}\to0$. Such
non-standard continued fraction expansions, ending with $0$ or $-1$
play a special role in the theory presented here, see e.g. Definition~\ref{def: backward type}.
\begin{rem*}
Observe that $\mathrm{Image}(\varphi)\subset(\Q\cap[0,1])\cup\{-1,\infty\}$.
The values $-1$ and $\infty$ deserve a special treatment and this
is done in the forthcoming statements, see e.g. Definition~\ref{def: sigma_c}
(see also \cite[Sec.~2.1]{BaBeBiTh22}). Currently, just note that
$\varphi(\co)=-1\Leftrightarrow\co=[0,0,-1]$, and 
\[
\varphi(\co)=\infty\quad\Leftrightarrow\quad\co\in\left\{ [0],~[0,0,0],~[0,0,1,-1]\right\} .
\]
\end{rem*}
It should be noted that the evaluation map $\varphi$ is surjective
but not injective. In fact for each rational $\frac{p}{q}\in(0,1)$,
there exist exactly two $\co,\tilde{\co}\in\Co$ whose last digit
is not $0$ nor $-1$ and $\varphi(\co)=\frac{p}{q}=\varphi(\tilde{\co})$,
see \cite[Ch.~I.4]{Khinchin_book64}. This is used in Proposition~\ref{prop: duality of A-B types}
below.

The following are well-known properties of the periodic Schrödinger
operators $\Hrat$ with $p,q$ coprime, see e.g. \cite{Tes00,Raym95,DaFi24-book_2,BaBeBiTh22}.
\begin{prop}
\label{prop: Basic spectral prop periodic} Let $V\in\mathbb{R}\backslash\left\{ 0\right\} $
and $\frac{p}{q}\in[0,1]$ such that $p$ and $q$ are coprime. Then
$\Hrat$ has absolutely continuous spectrum and the spectrum $\sigma\left(H_{\frac{p}{q},V}\right)$
consists of exactly $q$ connected components, each being a closed
interval \textup{$I$}.
\end{prop}

Let $I\subseteq\R$ be a closed interval. We define its left and right
endpoints by
\[
L(I):=\inf_{x\in I}x\quad\textrm{respectively}\quad R(I):=\sup_{x\in I}x.
\]
An antisymmetry of the Floquet--Bloch matrices yields $\sigma\left(H_{\frac{p}{q},V}\right)=-\sigma\left(H_{\frac{p}{q},-V}\right)$,
as shown in Lemma~\ref{lem: Symmetry Spectrum negative coupling}
below, reducing the study to $V>0$.
\begin{defn}
\label{def: A spectral band is continuous}For $p,q$ coprime, a map
$I:V\mapsto I(V)=\left[L(I(V)),R(I(V))\right],~V>0,$ is called a
\emph{spectral band} in $\sigma\left(H_{\frac{p}{q},V}\right)$ if
there is a $0\leq j<q$, such that for all $V>0$, $I(V)$ is the
$j$-th connected component (counted from the left) of $\sigma\left(H_{\frac{p}{q},V}\right)$.
\end{defn}

\begin{rem*}
In the following, we will abuse terminology and also refer to the
evaluation of that map, i.e., $I(V)$, as a spectral band. This is
a common terminology in the literature. Whether a spectral band means
the map itself or its evaluation will be either understood from the
context or explicitly mentioned.
\end{rem*}
The spectral bands vary continuously with respect to the Hausdorff
metric $d_{H}$ on the compact subsets of $\R$ induced by the Euclidean
distance,
\[
d_{H}(X,Y):=\max\left\{ \sup_{x\in X}d(x,Y),~\sup_{y\in Y}d(y,X)\right\} .
\]

\begin{prop}
\label{prop: Lipschitz spectral edges}Let $\alpha\in[0,1]$ and $V,V'\in\R$.
Then
\[
d_{H}\big(\sigma(H_{\alpha,V}),\sigma(H_{\alpha,V'})\big)\leq|V-V'|.
\]

In particular, if $I:V\mapsto I(V)$ is a spectral band of $\sigma\left(H_{\frac{p}{q},V}\right))$,
then for all $V,V'>0$,
\[
\max\big\{|L(I(V))-L(I(V'))|,\,|R(I(V))-R(I(V'))|\big\}\leq|V-V'|.
\]
\end{prop}

\begin{proof}
The first statement follows from the operator norm estimate $\|H_{\alpha,V}-H_{\alpha,V'}\|\leq|V-V'|$.
The second follows directly from the definition of the Hausdorff metric.
\end{proof}

We continue introducing basic relations on spectral bands.
\begin{defn}
\label{def: order relations on spectral bands - as intervals} For
two closed intervals $I$ and $J$ define the following order relations.
\begin{enumerate}
\item The interval \emph{$I$ is} contained \emph{in $J$}:
\[
I\subseteq J\quad\Leftrightarrow\quad L(J)\leq L(I)\mathrm{<R(I)\leq R(J)}.
\]
\item The interval \emph{$I$ is }strictly contained \emph{in $J$}:
\[
I\strict J\quad\Leftrightarrow\quad L(J)<L(I)\mathrm{<R(I)<R(J)}.
\]
\item The interval \emph{$I$ is} to the left of \emph{$J$ (respectively
$J$} \emph{is} to the right of\emph{ $I$):}
\[
I\prec J\quad\Leftrightarrow\quad L(I)<L(J)\mathrm{~and~R(I)<R(J).}
\]
\item The interval \emph{$I$ is} strictly to the left of \emph{$J$ (respectively
$J$} \emph{is} strictly to the right of\emph{ $I$):}
\[
I\precd J\quad\Leftrightarrow\quad R(I)<L(J).
\]
\end{enumerate}
\end{defn}

For two spectral bands $I:V\mapsto\left[L(I(V)),R(I(V))\right]$ and $J:V\mapsto\left[L(J(V)),R(J(V))\right]$, we extend these strict (i.e., irreflexive) order relations if they hold for all $V>0$.
\begin{enumerate}
\item The spectral band \emph{$I$ is }strictly contained \emph{in $J$}:
\[
I\strict J\quad\Leftrightarrow\quad\forall V>0:\quad I(V)\strict J(V)
\]
\item The spectral band \emph{$I$ is} to the left of \emph{$J$ (respectively
$J$} \emph{is} to the right of\emph{ $I$):}
\[
I\prec J\quad\Leftrightarrow\quad\forall V>0:\quad I(V)\prec J(V).
\]
\end{enumerate}
Note that it is possible that $I$ is to the left of $J$ even if
$I(V)\cap J(V)\neq\emptyset$ for some value of $V$.

We now have all tools to define the types of spectral bands. As discussed
in Section~\ref{sec: introduction and main results} the types depend
on the finite continued fraction expansion and so the following terminology
will be used.
\begin{defn}
\label{def: sigma_c}For all $V\in\R$, and $\co\in\Co$ define 
\[
\sigc(V):=\begin{cases}
\sigma\left(H_{\varphi(\co),V}\right), & \varphi(\co)\in[0,1],\\
\sigma\left(H_{1,-V}\right)=[-2-V,2-V], & \varphi(\co)=-1,\\
\R & \varphi(\co)=\infty.
\end{cases}
\]
\end{defn}

We will see in Section~\ref{subsec: spectra via transfer matrices}
that these sets are naturally defined by traces associated with $\co\in\Co$.
Towards the next definition we introduce the following notation:
\[
[\co,m]:=[0,0,c_{1},\ldots,c_{k},m],\qquad m\in\N_{-1},
\]
which is employed only when $[\co,m]\in\Co$. This notation will be
used in several statements and proofs below.
\begin{defn}
\label{def: backward type}Let $V\in\R$ and $\co\in\Co$ be such
that $\varphi(\co)\not\in\left\{ -1,\infty\right\} $ and $[\co,0],[\co,-1]\in\Co$.
A spectral band $I(V)$ of $\sigma_{\co}(V)$ is called
\begin{itemize}
\item \begin{flushleft}
\emph{backward type $\tA$} \\
if there exists a spectral band $J(V)$ in $\sigma_{[\co,0]}(V)$
such that $\mbox{\ensuremath{I(V)\strict J(V)}}$.
\par\end{flushleft}
\item \begin{flushleft}
\emph{weak backward type $\tA$} \\
if there exists a spectral band $J(V)$ in $\sigma_{[\co,0]}(V)$
such that $\mbox{\ensuremath{I(V)\subseteq J(V)}}$.
\par\end{flushleft}
\item \begin{flushleft}
\emph{backward type $\tB$} \\
if there exists a spectral band $J(V)$ in $\sigma_{[\co,-1]}(V)$
such that $\mbox{\ensuremath{I(V)\strict J(V)}}$.
\par\end{flushleft}
\item \begin{flushleft}
\emph{weak backward type $\tB$} \\
if there exists a spectral band $J(V)$ in $\sigma_{[\co,-1]}(V)$
such that $\mbox{\ensuremath{I(V)\subseteq J(V)}}$.
\par\end{flushleft}

\end{itemize}
\end{defn}

We note that, by definition, whether a spectral band is of (weak)
backward type $\tA$ or $\tB$ depends on V, since bands are treated
as intervals here. In Theorem~\ref{thm: Every band is A or B}, we
state that each spectral band in $\sigc(V)$ has a unique type (either
$\tA$ or $\tB$), independent of $V>0$ (and likewise for $V<0$).
For $V=0$, all spectra coincide with $[-2,2]$, so this case is excluded.
The notation $\tA$ and $\tB$ is adopted from \cite{KiKiLa03}, where
it appeared for the specific case $\alpha=\frac{\sqrt{5}-1}{2}$.
For visual reasons we do not use the $II$, $III$ notation as in
\cite{Raym95}, see also a discussion in \cite[Sec.~4.2]{BaBeBiTh22}.

By definition $\sigma_{\cop}(V)=\sigc(V)$ for $\co,\widetilde{\co}\in\Co$
with $\varphi(\co)=\varphi(\widetilde{\co})$. Nevertheless, we emphasize
that the backward type of a spectral band $I(V)$ of $\sigc(V)$ depends
on $\co\in\Co$ and not on its evaluation $\varphi(\co)\in[0,1]$.
This is demonstrated in the next proposition, which shows why it is
advantageous to use the space $\Co$ and not just rational numbers.
\begin{prop}
\label{prop: duality of A-B types}Let $V\in\R$ and $\co\in\Co$
with $\varphi(\co)\not\in\left\{ -1,\infty\right\} $ and $[\co,m]\in\Co$
for all $m\in\N_{-1}$. For $m\geq2$, we have $\sigma_{[\co,m]}(V)=\sigma_{[\co,m-1,1]}(V)$.
Moreover, if $I(V)$ is a spectral band in $\sigma_{[\co,m]}(V)=\sigma_{[\co,m-1,1]}(V)$,
then both of the following hold
\begin{itemize}
\item \begin{flushleft}
$I(V)$ is of (weak) backward type $A$ in $\sigma_{[\co,m]}(V)$
if and only if\\
 $I(V)$ is of (weak) backward type $B$ in $\sigma_{[\co,m-1,1]}(V)$.\\
~
\par\end{flushleft}
\item \begin{flushleft}
$I(V)$ is of (weak) backward type $B$ in $\sigma_{[\co,m]}(V)$
if and only if \\
$I(V)$ is of (weak) backward type $A$ in $\sigma_{[\co,m-1,1]}(V)$.
\par\end{flushleft}
\end{itemize}
\end{prop}

\begin{proof}
If $m\geq2$, then $\varphi([\co,m])=\varphi([\co,m-1,1])$ follows
by the definition of $\varphi$ (this is actually a well-known duality
for finite continued fraction expansions \cite[Ch.~I.4]{Khinchin_book64}).
Now, $\sigma_{[\co,m]}(V)=\sigma_{[\co,m-1,1]}(V)$ follows. We suppress
the $V$ dependence in the following.

Let $I$ be a spectral band in $\sigma_{[\co,m]}=\sigma_{[\co,m-1,1]}$.
By definition, $I$ is of backward type $A$ in $\sigma_{[\co,m]}$
if and only if it is strictly contained in a spectral band of $\sigma_{[\co,m,0]}=\sigma_{\co}$
(where we used $\varphi([\co,m,0])=\varphi(\co)$). On the other hand,
$I$ is of backward type $B$ in $\sigma_{[\co,m-1,1]}$ if and only
if it is strictly contained in a spectral band of $\sigma_{[\co,m-1,1,-1]}=\sigma_{\co}$
(where we used $\varphi([\co,m-1,1,-1])=\varphi([\co,m-1,0])=\varphi(\co)$).
This proves the first equivalence and the second one follows similarly.
\end{proof}

\begin{figure}
\includegraphics[scale=0.75]{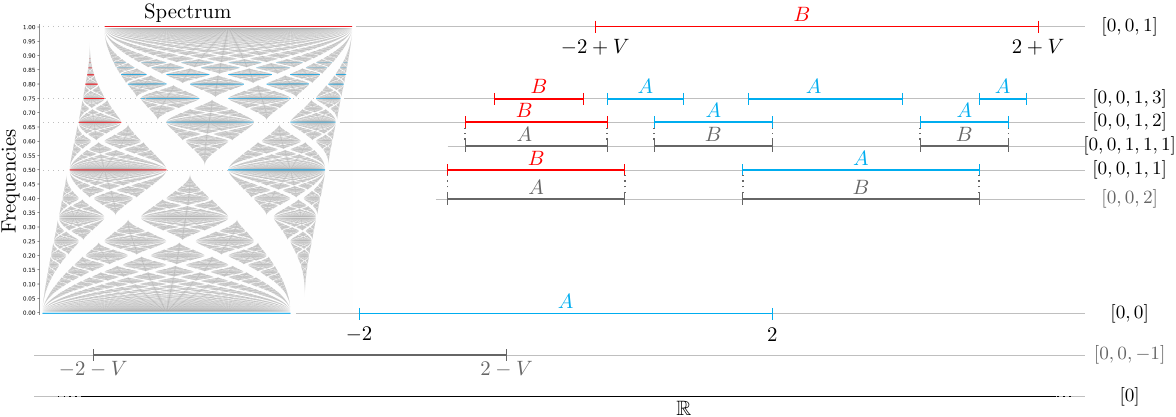} \caption{A plot of various spectra $\sigma_{\protect\co}(V)$ for $\protect\co\in\protect\Co$.
The spectral bands are colored according to their backward types ($A$
in blue and $B$ in red). The embedding of these spectral bands within
the Kohmoto butterfly is highlighted.}
\label{fig: Initial Spectra}
\end{figure}

\begin{example}
\label{exa:A-B types} A short computation shows 
\[
\sigma_{[0,0]}(V)=[-2,2]=:I_{[0,0]}(V)\quad\text{and}\quad\sigma_{[0,0,1]}(V)=[-2+V,2+V]=:I_{[0,0,1]}(V)
\]
Moreover, $I_{[0,0]}(V)$ is of backward type $\tA$ but not of weak
backward type $\tB$, while $I_{[0,0,1]}(V)$ is of backward type
$\tB$ but not of weak backward type $\tA$, for all $V>0$.

Since $0=\varphi([0,0])$ and $1=\varphi([0,0,1])$ have unique finite
continued fraction expansions, their spectral bands have a unique
type. In contrast, each rational $\frac{p}{q}\in(0,1)$ admits exactly
two representations, so Proposition~\ref{prop: duality of A-B types}
applies. For example, $\frac{1}{2}=\varphi([0,0,2])=\varphi([0,0,1,1])$
and $\frac{2}{3}=\varphi([0,0,1,2])=\varphi([0,0,1,1,1])$, see Figure~\ref{fig: Initial Spectra}.

In addition, we observe in Figure~\ref{fig: Initial Spectra} that
$I_{[0,0]}$ contains exactly one band from each of $\sigma_{[0,0,1,1]}$,
$\sigma_{[0,0,1,2]}$ and $\sigma_{[0,0,1,3]}$. These bands form
a nested sequence in $\sigma_{[0,0,1,n]}$ for $n\in\N$ and are therefore
of backward type $\tB$. This spectral band nesting is part of a hierarchical
structure of the spectral bands which is proved in the current paper.
We formalize this next (the spectral band nesting mentioned above
appears as Property~\ref{enu: B1 property}). 
\end{example}

\begin{defn}
\label{def: forward type} Let $V\in\R$. Let $\co\in\Co$ and $m\in\N$
be such that $\varphi(\co)\not\in\left\{ -1,\infty\right\} $ and
$[\co,m]\in\Co$. A spectral band $I_{\co}(V)$ of $\sigma_{\co}(V)$
is called of \emph{$m$-forward type $\tA$ }with $M=m-1$ (respectively\emph{
$m$-forward type $\tB$} with $M=m$) if the following holds.
\begin{enumerate}[align=left, leftmargin={*}, label=(\Alph*)]
\item  \label{enu: A property} There exist $M$ spectral bands in $\sigma_{[\co,m]}(V)$
(denoted $I^{1}_{[\co,m]}(V),\ldots,I^{M}_{[\co,m]}(V)$) such that\linebreak{}
~
\begin{enumerate}[label=(A\arabic*)]
\item  \label{enu: A1 property} $I^{i}_{[\co,m]}(V)\strict I_{\co}(V)$
for all $1\leq i\leq M$.\\
In particular, these bands are of backward type $A$.\\
~
\item \label{enu: A2 property} $I^{i}_{[\co,m]}(V)$ is not of weak backward
type $B$ for all $1\leq i\leq M$.\\
~
\end{enumerate}
\item \label{enu: B property} For each $n\in\N$, there exist $M+1$ spectral
bands in $\sigma_{[\co,m,n]}(V)$\\
(denoted $I^{1}_{[\co,m,n]}(V),\ldots,I^{M+1}_{[\co,m,n]}(V)$) such
that\\
~
\begin{enumerate}[label=(B\arabic*)]
\item  \label{enu: B1 property}$I^{j}_{[\co,m,n]}(V)\strict I^{j}_{[\co,m,n-1]}(V)$
for all $1\leq j\leq M+1$, where $I^{j}_{[\co,m,0]}(V):=I_{\co}(V)$.\\
In particular, these bands are of backward type $B$.\\
~
\item \label{enu: B2 property} $I^{j}_{[\co,m,n]}$ is not of weak backward
type $A$ for all $1\leq j\leq M+1$.
\end{enumerate}
\end{enumerate}
~
\begin{enumerate}[align=left, leftmargin={*}, label=(I)]
\item \label{enu: I property} For each $n\in\N$, we have 
\[
I^{1}_{[\co,m,n]}\prec I^{1}_{[\co,m]}\prec I^{2}_{[\co,m,n]}\prec I^{2}_{[\co,m]}\ldots\prec I^{M}_{[\co,m]}\prec I^{M+1}_{[\co,m,n]}.
\]
\end{enumerate}
\end{defn}

\begin{rem*}
The nested spectral band structure given by Property~\ref{enu: B1 property}
will be called the \emph{tower property.}
\end{rem*}
Finally, the notions above are combined to define the type $A$ and
$B$ spectral bands.
\begin{defn}
\label{def: A B types} Let $V\in\R$. Let $\co\in\Co$ and $m\in\N$
such that $\varphi(\co)\not\in\left\{ -1,\infty\right\} $ and $[\co,m]\in\Co$.
A spectral band $I_{\co}(V)$ of $\sigma_{\co}(V)$ is called
\begin{itemize}
\item \emph{$m$-type $A$} if $I_{\co}(V)$ is of backward type $A$ and
of $m$-forward type $A$.
\item \emph{$m$-type $B$} if $I_{\co}(V)$ is of backward type $B$ and
of $m$-forward type $B$.
\item \emph{type $A$} if $I_{\co}(V)$ is of $m$-type $A$ for all $m\in\N$.
\item \emph{type $B$} if $I_{\co}(V)$ is of $m$-type $B$ for all $m\in\N$.
\end{itemize}
\end{defn}

The starting point of this work is the following substantial result
of Raymond which appeared more than three decades ago.
\begin{thm}
\cite{Raym95}\label{thm: V>4 Type A =000026 B} For all $V>4$ and
$\co\in\Co$ with $\varphi(\co)\not\in\left\{ -1,\infty\right\} $,
every spectral band in $\sigma_{\co}(V)$ is either of type $A$ or
$B$ and its type is independent of the value of $V>4$.

Moreover, for a spectral band $\Ic(V)$ and $m,n\in\N$, the spectral
bands $\Icmi(V)$ and $\Icmnj(V)$ introduced in the forward property
\ref{enu: A property} and \ref{enu: B property} are unique for $V>4$,
i.e. $\Ic(V)$ does not contain any other spectral band of $\sigcm(V)$
respectively $\sigcmn(V)$.
\end{thm}

We took liberty with phrasing Theorem~\ref{thm: V>4 Type A =000026 B}
differently than it originally appeared in \cite{Raym95} (it actually
did not appear there as a single theorem). In particular, the notation
used in \cite{Raym95} is different than ours; we had to adapt the
notation for the sake of our proofs. We have done such an adaptation
already in the review \cite[Thm.~4.22]{BaBeBiTh22}, as a preliminary
step towards the current paper.
\begin{defn}
\label{def: Notion Icmi. Icmnj and M}Let $m,n\in\N$, $\co\in\Co$
be such that $\varphi(\co)\not\in\left\{ -1,\infty\right\} $ and
$[\co,m]\in\Co$. For a spectral band $\Ic$ in $\sigc$ define the
associated unique value 
\[
M:=\begin{cases}
	m-1,\qquad & I_{\co}(V)\text{ is of backward type \ensuremath{A} for all }V>4,\\
	m,\qquad & I_{\co}(V)\text{ is of backward type \ensuremath{B} for all }V>4,
\end{cases}
\]
and the unique spectral bands $\left\{ I^{i}_{[\co,m]}\right\} ^{M}_{i=1}$
of $\sigcm$ and the unique spectral bands $\left\{ I^{j}_{[\co,m,n]}\right\} ^{M+1}_{j=1}$
of $\sigcmn$ satisfy \ref{enu: A property}, \ref{enu: B property}
and \ref{enu: I property} for all $V>4$.

Note that $M$ actually depends both on the backward type of $\Ic$
and on $m$, but we omit this dependence from the notation.
\end{defn}

The existence and uniqueness of the spectral bands $\left\{ I^{i}_{[\co,m]}\right\} ^{M}_{i=1}$
and $\left\{ I^{j}_{[\co,m,n]}\right\} ^{M+1}_{j=1}$ are justified
by Theorem~\ref{thm: V>4 Type A =000026 B}. Due to Corollary~\ref{prop: Lipschitz spectral edges},
we may consider the continuous maps $V\mapsto\Icmi(V)$ and $V\mapsto\Icmnj(V)$
on $V\in(0,\infty)$.

A word of caution is needed regarding the notation in Definition~\ref{def: Notion Icmi. Icmnj and M}.
In order to know to which spectral band the notation $\Icmi$ refers
to within $\sigcm$, one needs to know which spectral band $\Ic$
was designated. For different choices of spectral bands $\Ic$ within
$\sigc$, the spectral bands $\Icmi$ and $\Icmnj$ will also differ.
This should not lead to confusion, since in the beginning of each
proof or discussion, the spectral band $\Ic$ will be explicitly indicated.
\begin{defn}
\label{def:(mV)-type property}Let $m\in\N$, $\co\in\Co$ be such
that $\varphi(\co)\not\in\left\{ -1,\infty\right\} $ and $[\co,m]\in\Co$.
For $V>0$, we say that a spectral band $\Ic$ in $\sigc$ satisfies
$\TmV(m,V)$ (the \emph{$(m,V)$-property}) if both of the following
hold:
\end{defn}

\begin{enumerate}
\item either
\begin{itemize}
\item for all $V'\geq V$, $\Ic(V')$ is of backward type $\tA$ with $M=m-1$,
\end{itemize}
or
\begin{itemize}
\item for all $V'\geq V$, $\Ic(V')$ is of backward type $\tB$ with $M=m$,
\end{itemize}
\item for all $V'\geq V$ and for all $n\in\N$, the unique spectral bands
$\left\{ I^{i}_{[\co,m]}(V')\right\} ^{M}_{i=1}$ of $\sigcm(V')$
and the unique spectral bands $\left\{ I^{j}_{[\co,m,n]}(V')\right\} ^{M+1}_{j=1}$
of $\sigcmn(V')$, as in Definition~\ref{def: Notion Icmi. Icmnj and M},
satisfy \ref{enu: A property}, \ref{enu: B property} and \ref{enu: I property}.
\end{enumerate}

\section{Proof of the main theorem and further consequences\label{sec: Proof_AB_Type}}

In this section we prove Theorem~\ref{thm: Every band is A or B}
by induction over the space $\Co$ of finite continued fractions.
For all $\co\in\Co$ and $m\in\N$ such that $\varphi(\co)\not\in\left\{ -1,\infty\right\} $
and $[\co,m]\in\Co$, define
\[
\crit([\co,m]):=\inf\set{V\geq0}{\begin{array}{c}
\text{each spectral band \ensuremath{\Ic} in \ensuremath{\sigma_{\co}} \text{satisfies \ensuremath{\TmV(m,V)}}}\end{array}}.
\]

Hence, to prove Theorem~\ref{thm: Every band is A or B} we should
then show that $\crit([\co,m])=0$ for all relevant $\co\in\Co$ and
$m\in\N$. Theorem~\ref{thm: Every band is A or B} is stated for
both $V>0$ and $V<0$, but it is possible to restrict to the case
$V>0$ thanks to an antisymmetry of the Floquet-Bloch matrices, stated
in Lemma~\ref{lem: Symmetry Spectrum negative coupling}. The proof
of Theorem~\ref{thm: Every band is A or B} is carried by induction
over the space $\Co$, and it relies on two key ingredients: Proposition~\ref{prop: Backward implies forward}
and the induction base. Both ingredients require some detailed analysis
and their proofs are postponed to Sections~\ref{subsec:Backward-implies-forward}
and \ref{sec: Pf_InductionBase}, respectively. In the current section
we prove Theorem~\ref{thm: Every band is A or B} assuming Proposition~\ref{prop: Backward implies forward}
and the induction base.
\begin{prop}
[Backward implies forward type] \label{prop: Backward implies forward}
Let $\co\in\Co$ and $\varphi(\co)\in(0,1)$. If each spectral band
$I_{\co}(V)$ in $\sigma_{\co}(V)$ is either of backward type $A$
for all $V>0$ or of backward type $B$ for all $V>0$, then $\crit([\co,m])=0$
for all $m\in\N$.
\end{prop}

The induction step in the proof of of Theorem~\ref{thm: Every band is A or B}
is divided into two: increasing the number of digits of $\co\in\Co$
for which the statement is valid (call it a horizontal step), and
showing the validity for all values of the last digit of $\co$ (vertical
step). These steps are stated and proven in the next two lemmata with
the aid of Proposition~\ref{prop: Backward implies forward} and
afterwards we are able to prove Theorem~\ref{thm: Every band is A or B}.
Throughout this section we use the notational conventions of Definition~\ref{def: Notion Icmi. Icmnj and M}.
\begin{lem}
[horizontal induction step] \label{lem: horizontal induction} Let
$m\in\N$ and $\co\in\Co$ be such that $\varphi(\co)\not\in\left\{ -1,\infty\right\} $
and $[\co,m]\in\Co$. If $\crit([\co,m])=0$ and $\crit([\co,m,1])=0$,
then $\crit([\co,m,1,n])=0$ for all $n\in\N$.
\end{lem}

\begin{proof}
Let $\co':=[\co,m,1]$. We have to show that $\crit([\co',n])=0$
for all $n\in\N$. Since $m\in\N$, we conclude $\varphi(\co')\in(0,1)$.
Thus, Proposition~\ref{prop: Backward implies forward} implies that
it suffices to prove that each spectral band in $\sigma_{\co'}(V)$
is either of backward type $A$ for all $V>0$ or of backward type
$B$ for all $V>0$. Let $I_{\co'}$ be a spectral band in $\sigma_{\co'}$.
By Theorem~\ref{thm: V>4 Type A =000026 B}, we already have that
$I_{\co'}(V)$ is either of backward type $A$ for all $V>4$ or of
backward type $B$ for all $V>4$. We treat each of these two cases
separately.

\underline{Case 1:} (For all $V>4$, $I_{\co'}(V)$ is of backward
type $A$). In this case, using $\sigma_{[\co',0]}(V)=\sigma_{[\co,m]}(V)$
(as $c':=[\co,m,1]$ and so $\varphi([\co',0])=\varphi([\co,m])$)
we conclude that $I_{\co'}(V)$ is strictly included in a spectral
band of $\sigma_{[\co,m]}(V)$ for all $V>4$. By Theorem~\ref{thm: V>4 Type A =000026 B},
there is a unique spectral band $I_{[\co,m]}(V)$ such that $I_{\co'}(V)=I^{i}_{[\co,m,1]}(V)\strict I_{[\co,m]}(V)$,
for all $V>4$. Since $\crit([\co,m,1])=0$, we conclude $I_{\co'}(V)=I^{i}_{[\co,m,1]}(V)\strict I_{[\co,m]}(V)$
for all $V>0$ implying that $I_{\co'}(V)$ is of backward type $A$
for all $V>0$.

\underline{Case 2:} (For all $V>4$, $I_{\co'}(V)$ is of backward
type $B$). In this case, using $\sigma_{[\co',-1]}(V)=\sigma_{[\co,m,0]}(V)=\sigma_{\co}(V)$
(as $c':=[\co,m,1]$ and so $\varphi([\co',-1])=\varphi([\co,m,0])$)
we conclude that $I_{\co'}(V)$ is strictly included in a spectral
band of $\sigma_{\co}(V)$ for all $V>4$. By Theorem~\ref{thm: V>4 Type A =000026 B},
there is a unique spectral band $I_{\co}(V)$ such that $I_{\co'}(V)=I^{j}_{[\co,m,1]}(V)\strict I_{\co}(V)$,
for all $V>4$. Since $\crit([\co,m])=0$, we conclude $I_{\co'}(V)=I^{j}_{[\co,m,1]}(V)\strict I_{\co}(V)$
for all $V>0$ implying that $I_{\co'}(V)$ is of backward type $B$
for all $V>0$.
\end{proof}

\begin{lem}
[vertical induction step] \label{lem:vertical-induction-step} Let
$\co\in\Co$ be such that $\varphi(\co)\not\in\left\{ -1,\infty\right\} $
and $[\co,m]\in\Co$ for all $m\in\N$. If $\crit([\co,m])=0$ for
all $m\in\N$ and $\crit([\co,1,n])=0$ for all $n\in\N$, then $\crit([\co,m,n])=0$
for all $m,n\in\N$.
\end{lem}

\begin{proof}
We denote by $T(m)$ the statement that $\crit([\co,m,n])=0$ for
all $n\in\N$. We use induction over $m\in\N$ to prove that $T(m)$
holds for all $m\in\N$ and the lemma follows. The induction base
$T(1)$ is true by the assumption in the lemma.

Suppose $T(m)$ holds. Denote $\co':=[\co,m+1]$. We need to show
$\crit([\co,m+1,n])=\crit([\co',n])=0$ for all $n\in\N$. Since $m\in\N$,
we have $m+1\geq2$ and so $\varphi(\co')\in(0,1)$. Thus, Proposition~\ref{prop: Backward implies forward}
implies that it suffices to prove that each spectral band in $\sigma_{\co'}(V)$
is either of backward type $A$ for all $V>0$ or of backward type
$B$ for all $V>0$. Let $I_{\co'}$ be a spectral band in $\sigma_{\co'}$.
By Theorem~\ref{thm: V>4 Type A =000026 B}, $I_{\co'}(V)$ is either
of backward type $A$ for all $V>4$ or of backward type $B$ for
all $V>4$. We treat each of these two cases separately.

\underline{Case 1:} (For all $V>4$, $I_{\co'}(V)$ is of backward
type $A$). In this case, using $\sigma_{[\co',0]}(V)=\sigma_{\co}(V)$
(as $c':=[\co,m+1]$ and so $\varphi([\co',0])=\varphi([\co])$) we
conclude that $I_{\co'}(V)$ is strictly included in a spectral band
of $\sigc(V)$ for all $V>4$. By Theorem~\ref{thm: V>4 Type A =000026 B}
there is a unique spectral band $\Ic(V)$ such that $I_{\co'}(V)=I^{i}_{[\co,m+1]}(V)\strict I_{\co}(V)$,
for all $V>4$. Since $\crit([\co,m+1])=0$, we conclude $I_{\co'}(V)=I^{i}_{[\co,m+1]}(V)\strict I_{\co}(V)$
for all $V>0$ implying that $I_{\co'}(V)$ is of backward type $A$
for all $V>0$.

\underline{Case 2:} (For all $V>4$, $I_{\co'}(V)$ is of backward
type $B$). In this case, using $\sigma_{[\co',-1]}(V)=\sigma_{[\co,m]}(V)$
(as $c':=[\co,m+1]$ and so $\varphi([\co',-1])=\varphi([\co,m])$)
we conclude that $I_{\co'}(V)$ is strictly included in a spectral
band of $\sigcm(V)$ for all $V>4$. Recalling that $\sigma_{\co'}(V)=\sigma_{[\co,m,1]}(V)$
(again by $\varphi(\co')=\varphi([\co,m,1])$) and applying Theorem~\ref{thm: V>4 Type A =000026 B},
we conclude that there is a unique spectral band $I_{[\co,m]}(V)$
such that $I_{\co'}(V)=I^{i}_{[\co,m,1]}(V)\strict I_{[\co,m]}(V)$,
for all $V>4$. Since by the induction hypothesis $\crit([\co,m,1])=0$,
we conclude $I_{\co'}(V)=I^{i}_{[\co,m,1]}(V)\strict I_{[\co,m]}(V)$
for all $V>0$ implying that $I_{\co'}(V)$ is of backward type $B$
for all $V>0$.
\end{proof}

We are ready to prove Theorem~\ref{thm: Every band is A or B}, using
Lemmata~\ref{lem: horizontal induction} and \ref{lem:vertical-induction-step}
for the induction step and where the ingredients needed for the proof
of the induction base are postponed to Section~\ref{sec: Pf_InductionBase}.
\begin{proof}[Proof of Theorem~\ref{thm: Every band is A or B}]
 Thanks to the anti-symmetric property, $\sigc(V)=-\sigc(-V)$, proven
in Lemma~\ref{lem: Symmetry Spectrum negative coupling}, it suffices
to consider the case $V>0$. We should therefore prove that $\crit([\co,m])=0$
for all $\co\in\Co$ with $\varphi(\co)\in[0,1]$ and $[\co,m]\in\Co$,
where $m\in\N$.

For $l\in\N$, we denote by $T(l)$ the statement that 
\begin{equation}
\crit([0,0,c_{1},\ldots,c_{l}])=0\qquad\text{and}\qquad\crit([0,0,c_{1},\ldots,c_{l+1}])=0,\label{eq: Every band A-B: Proof - IB}
\end{equation}
for all $[0,0,c_{1},\ldots,c_{l},c_{l+1}]\in\Co$ with $c_{l+1}\in\N$.

We start from the induction base, $T(1)$. By Lemma~\ref{lem: I_0.0 properties}
and Lemma~\ref{Lem-=00005B0.0.1=00005D_m-forward_B},
\[
\crit([0,0,c_{1}])=0\qquad\text{and}\qquad\crit([0,0,1,c_{2}])=0,
\]
hold for all $c_{1},c_{2}\in\N$. Then Lemma~\ref{lem:vertical-induction-step}
(with $\co=[0,0]$, $m=c_{1}$, $n=c_{2}$) gives that $\crit([0,0,c_{1},c_{2}])=0$
for all $c_{1},c_{2}\in\N$, proving the induction base.

Now, suppose $T(l)$ holds for $l\in\N$. In order to prove $T(l+1)$,
it suffices to show that $\crit([0,0,c_{1},\ldots,c_{l+1},c_{l+2}])=0$
for all $c_{l+2}\in\N$. Apply Lemma~\ref{lem: horizontal induction}
(for $\co=[0,0,c_{1},\ldots,c_{l-1}]$ and $m=c_{l}$) to the induction
hypothesis (\ref{eq: Every band A-B: Proof - IB}) with $c_{l+1}=1$
to conclude 
\[
\crit([0,0,c_{1},\ldots,c_{l},1,c_{l+2}])=0,
\]
for all $c_{l+2}\in\N$. Using this and the induction hypothesis (\ref{eq: Every band A-B: Proof - IB}),
we apply Lemma~\ref{lem:vertical-induction-step} (for $\co=[0,0,c_{1},\ldots,c_{l}]$,
$m=c_{l+1}$, $n=c_{l+2}$) and get that 
\[
\crit([0,0,c_{1},\ldots,c_{l},c_{l+1},c_{l+2}])=0,
\]
for all $c_{l+1},c_{l+2}\in\N$. Hence, $T(l+1)$ holds.
\end{proof}

Having proven Theorem~\ref{thm: Every band is A or B} we devote
the rest of this section for drawing some interesting conclusions
from this classification of all spectral bands to types $\tA$ and
$\tB$. The next proposition presents an equivalent formulation of
this dichotomy.
\begin{prop}
\label{prop: A-B-dichotomy extended also for backward} Let $\co=[0,c_{0},c_{1},c_{2},\ldots,c_{k}]\in\Co$
with $\varphi(\co)\not\in\left\{ -1,\infty\right\} $ and $V\neq0$.
For a spectral band $\Ic$ in $\sigc$, we have the following equivalences
\[
\Ic(V)\textrm{ is of type }A\quad\Leftrightarrow\quad\begin{array}{c}
\Ic(V)\textrm{ is of}\\
\textrm{backward type }A
\end{array}\quad\Leftrightarrow\quad\Ic(V)\strict\sigma_{[0,c_{0},c_{1},\ldots,c_{k-1}]}(V)
\]
and
\[
\Ic(V)\textrm{ is of type }B\quad\Leftrightarrow\quad\begin{array}{c}
\Ic(V)\textrm{ is of}\\
\textrm{backward type }B
\end{array}\quad\Leftrightarrow\quad\begin{array}{c}
\Ic(V)\not\subseteq\sigma_{[0,c_{0},c_{1},\ldots,c_{k-1}]}(V)\textrm{ and }\\
\Ic(V)\strict\sigma_{[0,c_{0},c_{1},\ldots,c_{k-2}]}(V).
\end{array}
\]
\end{prop}

\begin{proof}
By Lemma~\ref{lem: Symmetry Spectrum negative coupling}, it suffices
to consider $V>0$. The cases $\varphi(\co)\in\{0,1\}$ follow from
Lemma~\ref{lem: I_0.0 properties} and Lemma~\ref{Lem-=00005B0.0.1=00005D_m-forward_B}.

Let $\varphi(\co)\in(0,1)$. By Theorem~\ref{thm: Every band is A or B},
each band is of type $\tA$ or $\tB$. The left-most equivalences
in the statement, namely the equivalence between type $\tA\backslash\tB$
and backward type $\tA\backslash\tB$ follows from Proposition~\ref{prop: Backward implies forward}.

By Definition~\ref{def: backward type}, $\Ic(V)$ is of backward
type $\tA$ if and only if $\Ic(V)\strict\sigma_{[0,c_{0},\ldots,c_{k-1}]}(V)$,
proving the right-most equivalence in the first line.

For type $\tB$, if $\Ic(V)\not\subseteq\sigma_{[0,c_{0},\ldots,c_{k-1}]}(V)$,
then $\Ic(V)$ is not of backward type $\tA$ and hence of type B
by Theorem~\ref{thm: Every band is A or B}. Conversely, type B implies
this non-inclusion. It remains to show $\Ic(V)\strict\sigma_{[0,c_{0},\ldots,c_{k-2}]}(V)$
if $\Ic(V)$ is of type $B$.

By definition, $\Ic(V)\strict J_{1}(V)$ for some band in $\sigma_{[0,c_{0},\ldots,c_{k}-1]}(V)$.
If $c_{k}=1$, we are finished since then $\sigma_{[0,0,c_{1},\ldots,c_{k}-1]}(V)=\sigma_{[0,0,c_{1},\ldots,c_{k-2}]}(V)$.
Otherwise, $J_{1}(V)$ cannot be of backward type $\tA$. Indeed,
if it were, then $J_{1}(V)\strict\sigma_{[0,c_{0},\ldots,c_{k}-1,0]}(V)=\sigma_{[0,c_{0},\ldots,c_{k-1}]}(V)$,
which contradicts $\Ic(V)\not\subseteq\sigma_{[0,c_{0},\ldots,c_{k-1}]}(V)$.

Iterating this argument yields a chain of spectral bands $J_{n}(V)$
in $\sigma_{[0,0,c_{1},\ldots,c_{k}-n]}(V)$ for $n\in\left\{ 1,\ldots,c_{k}-1\right\} $
of backward type $\tB$ such that 
\[
\Ic(V)\strict J_{1}(V)\strict J_{2}(V)\strict\ldots\strict J_{c_{k}-1}(V).
\]
In particular, $\Ic(V)\strict J_{c_{k}-1}(V)\strict\sigma_{[0,c_{0},c_{1},\ldots,c_{k-1},1,-1]}(V)=\sigma_{[0,c_{0},\ldots,c_{k-2}]}(V)$.
\end{proof}

\begin{rem*}
A spectral band $\Ic(V)$ of type $A$ may also satisfy $\Ic(V)\strict\sigma_{[0,c_{0},c_{1},\ldots,c_{k-2}]}(V)$
for some values of $V$. Such examples of spectral bands of type $\tA$
can be found in the spectrum $\sigma_{[0,0,1,2]}$. This explains
why the classification of type $\tB$ includes also the condition
$\Ic(V)\not\subseteq\sigma_{[0,c_{0},c_{1},\ldots,c_{k-1}]}(V)$.

The rightmost characterization of spectral bands in Proposition~\ref{prop: A-B-dichotomy extended also for backward}
is of particular interest when considering a sequence of rational
approximations of an irrational frequency $\alpha\in[0,1]\backslash\Q$.
This is used in \cite{BaBeLo_DTMP26} for the solution of the dry
ten Martini problem for Sturmian Hamiltonians (see also \cite{BaBeLo23-MFO}
for a brief description). We bring next two additional corollaries
which provide interesting information on the spectral band types from
the perspective of rational approximations.
\end{rem*}
\begin{cor}
\label{cor: number of A-B bands}Let $V\neq0$ and $\co:=[0,c_{0},c_{1},\ldots,c_{k}]\in\Co$
for $k\in\N_{0}$ be such that $\varphi(\co)\in[0,1]$ and $c_{k}\in\N$
if $k\geq1$. Define the rational numbers $\alpha_{j}:=\varphi([0,c_{0},c_{1},\ldots,c_{j}])=\frac{p_{j}}{q_{j}}\in[0,1]\cap\Q$
for $0\leq j\leq k$ with coprime $p_{j},q_{j}$ and $q_{-1}=0$,
$q_{0}=1$. Then the spectrum $\sigma_{k}$ contains $q_{k}-q_{k-1}$
spectral bands of type $A$, and $q_{k-1}$ spectral bands of type
$B$.
\end{cor}

\begin{proof}
Due to $\sigc(V)=-\sigc(-V)$, proven in Lemma~\ref{lem: Symmetry Spectrum negative coupling},
it suffices to consider the case $V>0$.

For $k\in\N_{0}$ and $\co:=[0,c_{0},c_{1},\ldots,c_{k}]\in\Co$,
define\,
\begin{align*}
\mathcal{N}^{(A)}_{k} & :=\textrm{number of spectral bands of type \ensuremath{A} in \ensuremath{\sigma_{\co}},}\\
\mathcal{N}^{(B)}_{k} & :=\textrm{number of spectral bands of type \ensuremath{B} in \ensuremath{\ensuremath{\sigma_{\co}}}. }
\end{align*}
 By standard properties of continued fractions \cite[Thm. 1]{Khinchin_book64},
we have 
\begin{equation}
q_{-1}=0,~q_{0}=1,~q_{j}=c_{j}q_{j-1}+q_{j-2},\qquad1\leq j\leq k.\label{eq: Proof vertices are spectral bands - recursion q_k}
\end{equation}
We will inductively over $k\in\N_{0}$ prove 
\begin{equation}
\mathcal{\mathcal{N}}^{(A)}_{k}\geq q_{k}-q_{k-1}\quad\textrm{and}\quad\mathcal{N}^{(B)}_{k}\geq q_{k-1}.\label{eq: Proof vertices are spectral bands - Induction Hyp. Spectral bands}
\end{equation}
For the induction base, we show the estimate for $k\in\left\{ 0,1\right\} $:
First let $\co=[0,0]$ with $k=0$. Then $\sigc$ consist of exactly
$q_{0}-q_{-1}=1$ spectral band $I=[-2,2]$ of type $\tA$ and $q_{-1}=0$
spectral bands of type $\tB$, see Example~\ref{exa:A-B types}.
Now let $k=1$ and consider $\co=[0,0,m]$ for some $m\in\N$. We
have $q_{1}=m$ and $q_{0}=1$. Then Lemma~\ref{lem: I_0.0 properties}~(\ref{enu: lem-I_0.0 properties - 3})
asserts that $\sigma_{[0,0,m]}$ consists of exactly $m-1=q_{1}-q_{0}$
spectral bands of backward type $\tA$ and $q_{0}=1$ spectral bands
of backward type $\tB$. Thus, these bands are of type $\tA$ respectively
$\tB$ by Proposition~\ref{prop: A-B-dichotomy extended also for backward}.
This finishes the induction base. Note that we actually proved equality.

For the induction step suppose (\ref{eq: Proof vertices are spectral bands - Induction Hyp. Spectral bands})
holds for $k$ and $k-1$. Let $\co=[0,c_{0},c_{1},\ldots,c_{k+1}]\in\Co$
with $c_{k+1}\in\N$. By forward type properties of a spectral band
of type $\tA$ and $\tB$ in $\sigma_{[0,c_{0},c_{1},\ldots,c_{k}]}$
and $\sigma_{[0,c_{0},c_{1},\ldots,c_{k-1}]}$, we conclude with the
induction hypothesis and (\ref{eq: Proof vertices are spectral bands - recursion q_k})
that
\begin{align*}
\mathcal{N}^{(A)}_{k+1} & \geq\left(c_{k+1}-1\right)\cdot\mathcal{N}^{(A)}_{k}+c_{k+1}\cdot\mathcal{N}^{(B)}_{k}\geq c_{k+1}q_{k}+q_{k-1}-q_{k}=q_{k+1}-q_{k},\\
\mathcal{\mathcal{N}}^{(B)}_{k+1} & \geq c_{k}\cdot\mathcal{\mathcal{N}}^{(A)}_{k-1}+\left(c_{k}+1\right)\cdot\mathcal{N}^{(B)}_{k-1}\geq c_{k}q_{k-1}+q_{k-2}=q_{k}.
\end{align*}
This proves (\ref{eq: Proof vertices are spectral bands - recursion q_k})
for all $\co:=[0,c_{0},c_{1},\ldots,c_{k}]\in\Co$ and $k\in\N_{0}$.
These inequalities are actually equalities since $\mathcal{\mathcal{N}}^{(A)}_{k}+\mathcal{\mathcal{N}}^{(B)}_{k}$
is bounded from above by the total number of spectral bands in $\sigma_{\co}$,
which equals to $q_{k}$ by Proposition~\ref{prop: Basic spectral prop periodic}.
\end{proof}

\begin{cor}
Let $V\neq0$ and $\co:=[0,c_{0},c_{1},\ldots,c_{k}]\in\Co$ for $k\in\N_{0}$
be such that $\varphi(\co)\in[0,1]$ and $c_{k}\in\N$ if $k\geq1$.
For a spectral band $\Ic(V)$ and $m,n\in\N$, the spectral bands
$\Icmi(V)$ and $\Icmnj(V)$ introduced in the forward property \ref{enu: A property}
and \ref{enu: B property} are unique for $V\neq0$, i.e. $\Ic(V)$
does not contain any other spectral band of $\sigcm(V)$ respectively
$\sigcmn(V)$.
\end{cor}

\begin{proof}
The uniqueness of the spectral bands follows immediately from our
counting argument in Corollary~\ref{cor: number of A-B bands} and
the forward property of each spectral band.
\end{proof}

\section{Two perspectives describing the spectra\label{sec: Basic spectral analysis}}

We discuss two descriptions of the spectra $\sigc(V)$: one via Floquet--Bloch
matrices and one using transfer matrices. Floquet--Bloch theory reduces
the spectral analysis to families of Hermitian matrices, providing
a framework for which we develop a suitable interlacing theorem. We
combine both descriptions to control spectral band edges.

\subsection{Floquet-Bloch matrices and an interlacing theorem\label{subsec: Perturbation argument}}

Given an $n\times n$ matrix $H$ and $\theta\in[0,2\pi]$, define
the $n\times n$ matrix
\[
H(\theta):=H+e^{-i\theta}\mathbb{I}_{1,n}+e^{i\theta}\Id_{n,1},
\]
where $\Id_{i,j}$ denotes the $n\times n$ matrix that has only zeros
except at the $(i,j)$-th entry where it is equal to one.

For $\alpha\in[0,1]$, we use in the following the notation $\omega_{\alpha}(n):=\chi_{\left[1-\alpha,1\right)}(n\alpha\mod 1)$
for the potential. Let $V\in\R$, $\co\in\Co$ be such that $\left\{ -1,\infty\right\} \not\ni\varphi(\co)=\frac{p}{q}$
with $p,q$ coprime. Recall the self-adjoint operator $H_{\varphi(\co),V}:\ell^{2}(\Z)\to\ell^{2}(\Z)$
introduced in Equation (\ref{eq: Hamiltonian defined}). The spectral
analysis of $H_{\varphi(\co),V}$ is done via the following hermitian
$q\times q$ matrix
\[
H_{\co,V}:=H_{\varphi(\co),V}|_{[0,q-1]}=\begin{pmatrix}V\omega_{\varphi(\co)}(0) & 1 & 0 & \ldots &  & 0\\
1 & V\omega_{\varphi(\co)}(1) & 1 & \ldots\\
0 & 1 & \ddots\\
\vdots & \ddots &  & \ddots &  & 0\\
0 &  &  &  &  & 1\\
0 & 0 & \cdots & 0 & 1 & V\omega_{\varphi(\co)}(q-1)
\end{pmatrix}.
\]

Note the ambiguity in the notation between the operator $H_{\varphi(\co),V}$
on $\ell^{2}(\Z)$ and the $q\times q$ matrix $H_{\co,V}$.

The Floquet-Bloch matrices $H_{\co,V}(\theta)$ determine the spectrum
of $H_{\varphi(\co),V}$ (see e.g. \cite{Hochstadt-LAA_1975} and
\cite[Thm.~2.7]{DaFi24-book_2}) by 
\begin{equation}
\sigc(V)=\sigma\left(H_{\varphi(\co),V}\right)=\bigcup_{\theta\in[0,\pi]}\sigma\left(H_{\co,V}(\theta)\right).\label{eq: Spectrum union thetas}
\end{equation}

\begin{lem}
\label{lem: Symmetry Spectrum negative coupling}For all $V\in\R$
and $\co\in\Co$ with $\varphi(\co)\in\left[0,1\right]$, we have
$\sigma_{\co}(V)=-\sigma_{\co}(-V).$
\end{lem}

\begin{proof}
Using the unitary $q\times q$ diagonal matrix, $D:=\mathrm{diag}\left\{ 1,-1,1,-1,\ldots\right\} $,
we obtain
\begin{itemize}
\item $D^{-1}H_{\co,-V}(\theta)\thinspace D=-H_{\co,V}(\theta)$ if $q$
is even,
\item $D^{-1}H_{\co,-V}(\theta)\thinspace D=-H_{\co,V}(\theta+\pi)$ if
$q$ is odd.
\end{itemize}
By this unitary equivalence and (\ref{eq: Spectrum union thetas}),
this yields $\sigc(V)=-\sigma_{\co}(-V)$.
\end{proof}

We have already mentioned (Proposition~\ref{prop: Basic spectral prop periodic})
that $\sigc(V)$ consists of exactly $q$ intervals (spectral bands).
By standard arguments, the endpoints of these intervals are given
by the eigenvalues of $H_{\co,V}(0)$ and $H_{\co,V}(\pi)$. Hence,
the values $\theta\in\{0,\pi\}$ play a significance role in (\ref{eq: Spectrum union thetas}).

The spectral decomposition (\ref{eq: Spectrum union thetas}) may
be also written in terms of the following $nq\times nq$-matrix

\[
\nHcV:=H_{\varphi(\co),V}|_{[0,nq-1]}=\begin{pmatrix}H_{\co,V} & \Id_{q,1} & 0 & \ldots &  & 0\\
\Id_{^{1,q}} & H_{\co,V} & \Id_{q,1} & \ldots\\
0 & \Id_{^{1,q}} & \ddots\\
\vdots & \ddots &  & \ddots &  & 0\\
0 &  &  &  &  & \Id_{q,1}\\
0 & 0 & \cdots & 0 & \Id_{^{1,q}} & H_{\co,V}
\end{pmatrix}.
\]

The diagonal of $\nHcV$ consists of $n$ repetitions of the diagonal
of $H_{\co,V}$, corresponding to the minimal period of the potential
sequence. Hence,
\begin{equation}
\sigc(V)=\sigma\left(H_{\varphi(\co),V}\right)=\bigcup_{\theta\in[0,\pi]}\sigma\left(\nHcV(\theta)\right).\label{eq: Spectrum union thetas - n times potential}
\end{equation}

The eigenvalues of $\nHcV(0)$ and $\nHcV(\pi)$ determine the band
edges, but also occur in the interior within these intervals (a detailed
description appears in the proof of Lemma~\ref{lem: Floquet-Bloch matrices - n times fundamental domain}).

Although $\nHcV$ may seem redundant purely for determining the spectrum
$\sigc(V)$, it becomes essential in the sequel. The matrices $H_{\co,V},H^{\times n}_{[\co,m],V}$,
and $H_{[\co,m,n],V}$ describe the spectra $\sigc$, $\sigcm$ and
$\sigcmn$, respectively. By Lemma~\ref{lem: periods of three approximants},
the diagonal of $H_{[\co,m,n],V}$ is a concatenation of those of
$H_{\co,V}$ and $\nHcm$, so that
\[
H_{[\co,m,n],V}=\nHcmV\oplus H_{\co,V}\quad\text{or}\quad H_{[\co,m,n],V}=H_{\co,V}\oplus\nHcmV,
\]

depending on the parity of the length of $\co$. This structure yields
an eigenvalue interlacing theorem. Writing the eigenvalues of a Hermitian
$q\times q$ matrix H as
\begin{equation}
\lambda_{0}(H)\leq\lambda_{1}(H)\leq\ldots\leq\lambda_{q-1}(H),\label{eq: Eigenvalues matrix - increasing order}
\end{equation}

we obtain the following.
\begin{thm}
[Interlacing theorem]\label{thm: perturbation thm for our matrices}
Let $V>0$. Let $m,n\in\N$ and $\co\in\Co$ be such that $\varphi(\co)\not\in\left\{ -1,\infty\right\} $
and $[\co,m]\in\Co$. Let $\thc,\thcm,\thcmn\in\{0,\pi\}$ and denote
\[
Y=H_{[\co,m,n],V}\left(\thcmn\right)\quad\textrm{and}\quad X=\nHcmV\left(\thcm\right)\oplus H_{\co,V}(\thc).
\]
 If $\thc+\thcm+\thcmn\in\{0,2\pi\}$, then 
\[
\lambda_{j-1}(Y)\leq\lambda_{j}(X)\leq\lambda_{j+1}(Y).
\]
Furthermore, if $\lambda_{j}(X)$ is a simple eigenvalue of $X$,
then both inequalities are strict.
\end{thm}

Theorem~\ref{thm: perturbation thm for our matrices} is proven in
the Appendix~\ref{App: Perturbation argument}. Note that even though
the eigenvalues depend on the parameter $V>0$, the inequalities of
the eigenvalues hold independently of the value $V>0$ attains\footnote{Note that simplicity may depend on $V>0$.}.
The condition in the previous theorem naturally leads to the following
useful definition of admissibility.
\begin{defn}
[Admissibility]\label{def: admissibility} Let $m,n\in\N$ and $\co\in\Co$
be such that $\varphi(\co)\not\in\left\{ -1,\infty\right\} $ and
$[\co,m]\in\Co$.
\begin{enumerate}
\item The values $\thc,\thcm,\thcmn\in\{0,\pi\}$ are called \emph{admissible}
if $\thc+\thcm+\thcmn\in\{0,2\pi\}$.
\item ~For each $\tilde{\co}\in\left\{ \co,\cm,\cmn\right\} $, let $I_{\tilde{\co}}:V\mapsto I_{\tilde{\co}}(V),V>0,$
be a spectral band in $\sigma_{\tilde{\co}}$, and let $\lambda_{\tilde{\co}}:V\mapsto\lambda_{\tilde{\co}}(V),V>0,$
satisfy either $\lambda_{\tilde{\co}}(V)=L\left(I_{\tilde{\co}}(V)\right)$
for all $V>0$ or $\lambda_{\tilde{\co}}(V)=R\left(I_{\tilde{\co}}(V)\right)$
for all $V>0$. Then we call $\lc,\lcm,\lcmn$ admissible, if there
exist admissible $\thc,\thcm,\thcmn\in\{0,\pi\}$ such that for all
$V>0,$ 
\[
\lc(V)\in\sigma\left(\HcV(\thc)\right),~\lcm(V)\in\sigma\left(\nHcmV(\thcm)\right),~\lcmn(V)\in\sigma\left(\HcmnV(\thcmn)\right).
\]
\end{enumerate}
\end{defn}

\begin{rem}
\label{rem:Admissible=00003DEvenNumberPis}We emphasize here that
$\thc,\thcm,\thcmn$ are admissible if the triple has an even number
of $\pi$'s. In particular, 
\[
\thc,\thcm,\thcmn\;\textrm{are not admissible}\quad\Leftrightarrow\quad\thc,\pi-\thcm,\thcmn\;\textrm{are admissible.}
\]

We further note that the maps $\lambda_{\tilde{\co}}:V\mapsto\lambda_{\tilde{\co}}(V),V>0$,
appearing in the definition are Lipschitz continuous by Proposition~\ref{prop: Lipschitz spectral edges}.
Moreover, Lemma~\ref{lem: Floquet-Bloch matrices} implies that for
each such map -- left or right end point of a fixed spectral band
$I_{\tilde{\co}}$ -- there exists a unique $\theta_{\tilde{\co}}\in\left\{ 0,\pi\right\} $
such that $\lambda_{\tilde{\co}}(V)\in\sigma\left(H_{\tilde{\co,V}}(\theta_{\tilde{\co}})\right)$
for all $V>0$. In particular, admissibility of $\lc,\lcm,\lcmn$
is independent of $V>0$.
\end{rem}

\subsection{Transfer matrices and their traces \label{subsec: spectra via transfer matrices}}

We present the well-known formalism for transfer matrices, see e.g.
\cite[Ch.~2]{DaFi22-book_1}. For $V\in\mathbb{R}$, define 
\[
M_{[0]}(E,V):=\begin{pmatrix}1 & -V\\
0 & 1
\end{pmatrix},\qquad M_{[0,0]}(E,V):=\begin{pmatrix}E & -1\\
1 & 0
\end{pmatrix}
\]
and recursively define the transfer matrices for $\co=[0,0,c_{1},\ldots,c_{k}]\in\Co$
(where $k\in\N$) by
\[
M_{\co}(E,V):=M_{[0,0,c_{1}\ldots,c_{k-2}]}(E,V)M_{[0,0,c_{1},\ldots,c_{k-1}]}(E,V)^{c_{k}}.
\]
Denote the traces of the transfer matrices by
\begin{equation}
\tc(E,V):=\tr(M_{\co}(E,V)).\label{eq: traces}
\end{equation}
Our description only slightly deviates from the conventional one,
by referring to all the elements of $\Co$ (within the literature
above we take a route which is the closest to \cite{Raym95}).

Denote by $\chi_{\HcV(\theta)}$ the characteristic polynomial of
the Floquet-Bloch matrix $\HcV(\theta)$. Then we have (see e.g.,
\cite[Eq. (23)]{Hochstadt-LAA_1975}, \cite[Thm. 5.4.1,(iii)]{Simon2011}
or \cite[Lem.~II.2]{BaBeBiTh22}) that
\begin{equation}
\chi_{\HcV(\theta)}(E)=\tc(E,V)-2\cos(\theta),\qquad\theta\in\left[0,2\pi\right].\label{eq: Chatacteristic polynomial and traces}
\end{equation}
This leads to the following well-known result, see e.g. \cite[Thm.~7.2.7]{DaFi24-book_2}
and \cite[Sec.~5.4]{Simon2011}.
\begin{lem}
\label{lem: ConnecSpectrTrace} For all $\co,\widetilde{\co}\in\Co$
with $\varphi(\widetilde{\co})=\varphi(\co)$, we have $t_{\widetilde{\co}}(E,V)=\tc(E,V)$
for all $E,V\in\R$. Furthermore,
\[
\sigc(V)=\set{E\in\R}{\left|\tc(E,V)\right|\leq2},\qquad\co\in\Co,\,V\in\R.
\]
\end{lem}

\begin{proof}
This is an immediate consequence of (\ref{eq: Spectrum union thetas})
and (\ref{eq: Chatacteristic polynomial and traces}). 
\end{proof}

\begin{lem}
\label{lem: Traces and spectral edges} Let $V\in\mathbb{R}\setminus\{0\}$,
$\co\in\Co$ with $\varphi(\co)\not\in\left\{ -1,\infty\right\} $.
Then the following statements hold.
\begin{enumerate}
\item \label{enu: prop-Traces and spectral edges - 1}For $E\in\R$, we
have $|t_{\co}(E,V)|=2$, if and only if $E\in\{L(\Ic(V)),R(\Ic(V))\}$\textup{
for some spectral band} $\Ic(V)$ in $\sigc(V)$.
\item \label{enu: prop-Traces and spectral edges - 2}If a spectral band
$\Ic$ in $\sigc$ is
\begin{itemize}
\item of backward type $A$, then $|t_{[\co,0]}(E,V)|\leq2$ for all $E\in\Ic$.
The estimate is strict if $\varphi(\co)\in(0,1)$.
\item of backward type $B$, then $|t_{[\co,-1]}(E,V)|\leq2$ for all $E\in\Ic$.
The estimate is strict if $\varphi(\co)\in(0,1)$.
\end{itemize}
\item \label{enu: prop-Traces and spectral edges - 3}For $m\geq0$, we
have $t_{[\co,m+1]}=t_{\co}t_{[\co,m]}-t_{[\co,m-1]}$.
\end{enumerate}
\end{lem}

\begin{proof}
Let $\varphi(\co)=\frac{p}{q}$ be such that $p,q$ are coprime.\\
(\ref{enu: prop-Traces and spectral edges - 1}) This is an immediate
consequence of \cite[Thm.~5.4.2]{Simon2011} and that $\sigc(V)$
consists of exactly $q$ spectral bands, see Proposition~\ref{prop: Basic spectral prop periodic}.
\\
(\ref{enu: prop-Traces and spectral edges - 2}) This follows from
Definition~\ref{def: backward type}, Lemma~\ref{lem: ConnecSpectrTrace}
and (\ref{enu: prop-Traces and spectral edges - 1}).\\
(\ref{enu: prop-Traces and spectral edges - 3}) This well-known identity
is proven in \cite{Raym95}. The reader is also referred to Appendix~\ref{Sec-TraceMaps}
for related results and more references, see also \cite[Lem.~3.8]{BaBeBiTh22}.
\end{proof}

\begin{rem*}
The first statement (\ref{enu: prop-Traces and spectral edges - 1})
of the lemma says that the traces attain the values $\pm2$ exactly
at the spectral band edges. This does not hold for $\co=[0]$ where
$\tc(E,V)=2$ and $\sigc(V)=\R$.
\end{rem*}
The next statement is based on standard techniques of transfer matrix
traces and its proof is included in the Appendix~\ref{Sec-TraceMaps}.
\begin{lem}
\label{lem: Trace estimates =00005Bc.m.n=00005D} Let $V\in\R$, $m\in\N$,
$\co\in\Co$ be such that $\varphi(\co)\not\in\left\{ -1,\infty\right\} $
and $[\co,m]\in\Co$. Let $I(V)$ be a spectral band in $\sigma_{\co}(V)$
of backward type $A$ or backward type $B$. Then for $E\in\{L(I(V)),R(I(V))\}$
and $n\in\mathbb{N}$, the following holds.
\begin{enumerate}
\item \label{enu: prop-Trace estimates =00005Bc.m.n=00005D - 1} $|t_{[\co,m]}(E,V)|\geq2\hspace{3cm}\quad\Rightarrow\quad|t_{[\co,m,n]}(E,V)|\geq2$.
\item \label{enu: prop-Trace estimates =00005Bc.m.n=00005D - 2} $|t_{[\co,m]}(E,V)|>2\hspace{3cm}\quad\Rightarrow\quad|t_{[\co,m,n]}(E,V)|>2$.
\item \label{enu: prop-Trace estimates =00005Bc.m.n=00005D - 3} $\varphi(\co)\in(0,1)\text{ and }|t_{[\co,m]}(E,V)|\geq2\quad\Rightarrow\quad|t_{[\co,m,n]}(E,V)|>2$.
\end{enumerate}
\end{lem}

The concept of admissible eigenvalues (Definition~\ref{def: admissibility})
can also be characterized in terms of the traces of these eigenvalues,
which is a central tool towards the solution of the dry ten Martini
problem \cite{BaBeLo_DTMP26}.
\begin{prop}
\label{prop: Chatacterization Admissible via traces}Let $m\in\N$
and $\co\in\Co$ be such that $\varphi(\co)\not\in\left\{ -1,\infty\right\} $
and $[\co,m]\in\Co$. For each $\tilde{\co}\in\left\{ \co,[\co,m],[\co,m,1]\right\} $,
let $I_{\tilde{\co}}:V\mapsto I_{\tilde{\co}}(V),V>0,$ be a spectral
band in $\sigma_{\tilde{\co}}$ and $\lambda_{\tilde{\co}}\in\left\{ L\left(I_{\tilde{\co}}\right),R\left(I_{\tilde{\co}}\right)\right\} $.
Then $\lc,\lcm,\lambda_{[\co,m,1]}$ are admissible if and only if
\[
\sgn\left(\tc\left(\lc(V)\right)\cdot\tcm\left(\lcm(V)\right)\cdot t_{[\co,m,1]}\left(\lambda_{[\co,m,1]}(V)\right)\right)=+1\quad\text{for all }V>0.
\]
\end{prop}

\begin{proof}
Let $\tilde{\co}\in\left\{ \co,[\co,m],[\co,m,1]\right\} $. By Lemma~\ref{lem: Traces and spectral edges}~(\ref{enu: prop-Traces and spectral edges - 1})
and (\ref{eq: Chatacteristic polynomial and traces}), $\lambda_{\tilde{\co}}\in\left\{ L\left(I_{\tilde{\co}}\right),R\left(I_{\tilde{\co}}\right)\right\} $
if and only if there exists $\theta_{\tilde{\co}}\in\left\{ 0,\pi\right\} $
such that $\lambda_{\tilde{\co}}(V)\in\sigma\left(H_{\tilde{\co},V}(\theta_{\tilde{\co}})\right)$
for all $V>0$. Moreover, (\ref{eq: Chatacteristic polynomial and traces})
gives $t_{\tilde{\co}}\left(\lambda_{\tilde{\co}}(V)\right)=2\cos(\theta_{\tilde{\co}})$
for all $V\in\R$. Hence, this value equals $2$ if $\theta_{\tilde{\co}}=0$
and $-2$ if $\theta_{\tilde{\co}}=\pi$.

By Definition~\ref{def: admissibility}, the triple $\lc,\lcm,\lambda_{[\co,m,1]}$
is admissible if and only if the triple $\theta_{\co}$, $\theta_{[\co,m]}$,
$\theta_{[\co,m,1]}$ contains an even number of $\pi$'s. Equivalently,
the triple of trace values $\tc\left(\lc(V)\right)$, $\tcm\left(\lcm(V)\right)$,
$t_{[\co,m,1]}\left(\lambda_{[\co,m,1]}(V)\right)$ contains an even
number of entries equal to $-2$ for all $V>0$, namely their product
has a positive sign.
\end{proof}

\section{Tools towards proving forward type\label{sec: Admissibility. index relations}}

This section is devoted to various technical tools needed for the
induction base of the proof of Theorem~\ref{thm: Every band is A or B}
(Section~\ref{sec: Pf_InductionBase}) and for proving Proposition~\ref{prop: Backward implies forward}
- that backward type implies forward type (Section~\ref{subsec:Backward-implies-forward}).

Let us provide a short overview of this section. In Subsection~\ref{subsec: Counting-indices-and-eigenvalues},
we introduce an eigenvalue counting function, which later plays a
crucial role in application of the interlacing theorem (Theorem~\ref{thm: perturbation thm for our matrices}).
Since eigenvalue admissibility is a necessary condition in the interlacing
theorem, we give a useful characterization of it in Subsection~\ref{subsec: Characterization Admissible}.
With this at hand, in Subsection~\ref{subsec: counting-function and eigenvalue estimates}
we provide Lemma~\ref{lem: Basics for eigenavlue inequalities} which
is a manifestation of the interlacing theorem (Theorem~\ref{thm: perturbation thm for our matrices}).
In effect, it is this lemma which is going to be directly applied,
rather than Theorem~\ref{thm: perturbation thm for our matrices}.
In Subsection~\ref{subsec: Index-identities}, we develop index relations
which are needed whenever we apply Lemma~\ref{lem: Basics for eigenavlue inequalities}.
Then, Subsection~\ref{subsec: Forward properties} applies the various
index relations, eigenvalue estimates and trace estimates to prove
that the spectral bands $\Icmi$ and $\Icmnj$ maintain certain properties
from Definition~\ref{def: forward type} when $V$ decreases to zero.

Throughout this section we use the notational conventions of Definition~\ref{def: Notion Icmi. Icmnj and M}
without pointing them out all the time.

\subsection{Counting spectral bands and eigenvalues\label{subsec: Counting-indices-and-eigenvalues}}

In this subsection we consider two types of counting functions: for
the spectral bands in $\sigc$ and for the eigenvalues of the matrices
$\HcV(\theta)$, $\nHcV(\theta)$ and relate both types of functions.

First, we recall that $\sigc(V)$ consists of exactly $q$ intervals
for $\varphi(\co)=\frac{p}{q}$ (see Proposition~\ref{prop: Basic spectral prop periodic}
and Lemma~\ref{lem: ConnecSpectrTrace}) and that we consider each
spectral band as a Lipschitz continuous map, $V\mapsto\Ic(V)$, for
$V>0$ (Definition~\ref{def: A spectral band is continuous} and
Proposition~\ref{prop: Lipschitz spectral edges}). This justifies
the following.
\begin{defn}
\label{def: index of spectral band}[Index of a spectral band] Let
$\Ic$ be a spectral band of $\sigc$. The \emph{index} of $\Ic$
(in $\sigc$) is defined by 
\[
\ind(\Ic):=\left|\set{I\textrm{ is a spectral band of }\sigc}{I\prec\Ic}\right|.
\]
\end{defn}

\begin{rem}
\label{rem: index spectral band}Note that the index counting starts
from zero, namely $0\leq\ind(\Ic)\leq q-1$ where $\varphi(\co)=\frac{p}{q}$
with $p,q$ coprime. Moreover, we emphasize that $\ind(\Ic)$ is independent
of $V>0$, allowing us to assume $V>4$ in some instances and use
Theorem~\ref{thm: V>4 Type A =000026 B}.
\end{rem}

In order to apply the interlacing theorem (Theorem~\ref{thm: perturbation thm for our matrices}),
we need to count eigenvalues. Let $\left\{ \lambda_{i}(H)\right\} ^{n-1}_{i=0}$
be the eigenvalues (increasingly arranged and counted with multiplicity)
of an $n\times n$ matrix $H$, as in (\ref{eq: Eigenvalues matrix - increasing order}).
\begin{defn}
\label{def: Counting function}[Counting function] For an $n\times n$
hermitian matrix $H$, the \emph{eigenvalue counting function} is
defined by
\[
N(\lambda;H):=\left|\set{0\leq i\leq n-1}{\lambda_{i}(H)<\lambda}\right|.
\]
\end{defn}

\begin{rem*}
Note that $N(\lambda;~H)$ may attain the value zero and also that
$N(\lambda_{i}(H);~H)=i$ for each $0\leq i\leq n-1$ where $\lambda_{i}(H)$
is simple.
\end{rem*}
We will be in particular interested in evaluating the counting function
for an eigenvalue which is also an edge of a certain spectral band.
The index of that spectral band is then related to the counting of
its edge point, as follows.
\begin{lem}
\label{lem: Floquet-Bloch matrices} Let $V>0$, $\co\in\Co$ and
$\left\{ -1,\infty\right\} \not\ni\varphi(\co)=\frac{p}{q}$ with
$p,q$ coprime. Let $\Ic$ be a spectral band of $\sigc$ and $\theta\in\{0,\pi\}$.
\begin{enumerate}
\item \label{enu: lem-Floquet-Bloch matrices - 1} We have
\[
\ind(\Ic)-q\equiv\frac{1}{\pi}\theta\mod 2\quad\Leftrightarrow\quad L(\Ic(V))\in\sigma\left(\HcV(\theta)\right)
\]
and
\[
\ind(\Ic)+1-q\equiv\frac{1}{\pi}\theta\mod 2\quad\Leftrightarrow\quad R(\Ic(V))\in\sigma\left(\HcV(\theta)\right).
\]
\item \label{enu: lem-Floquet-Bloch matrices - 2} If $L(\Ic(V))\in\sigma\left(\HcV(\theta)\right)$,
then 
\[
\ind(\Ic)=N\left(L(\Ic(V));\HcV(\theta)\right).
\]
\item \label{enu: lem-Floquet-Bloch matrices - 3} If $R(\Ic(V))\in\sigma\left(\HcV(\theta)\right)$,
then 
\[
\ind(\Ic)=N\left(R(\Ic(V));\HcV(\theta)\right).
\]
\end{enumerate}
\end{lem}

\begin{proof}
This follows from the next lemma and the fact that $H_{\co,V}(\theta)=\nHcV(\theta)$
if $n=1$.
\end{proof}

Lemma~\ref{lem: Floquet-Bloch matrices} can be generalized as follows.
\begin{lem}
\label{lem: Floquet-Bloch matrices - n times fundamental domain}
Let $V>0$, $\co\in\Co$ and $\left\{ -1,\infty\right\} \not\ni\varphi(\co)=\frac{p}{q}$
with $p,q$ coprime. Let $\Ic$ be a spectral band of $\sigc$ and
$\theta\in\{0,\pi\}$. Then the following holds for $n\in\N$.
\begin{enumerate}
\item \label{enu: lem-Floquet-Bloch matrices - n times fundamental domain - 1}
If $n\in\N$ is even, then $L(\Ic(V))\in\sigma\left(\nHcV(0)\right)$
and $R(\Ic(V))\in\sigma\left(\nHcV(0)\right)$.
\item \label{enu: lem-Floquet-Bloch matrices - n times fundamental domain - 2}
If $n\in\N$ is odd, then
\[
\ind(\Ic)-q\equiv\frac{1}{\pi}\theta\mod 2\quad\Leftrightarrow\quad L(\Ic(V))\in\sigma\left(\nHcV(\theta)\right)
\]
and
\[
\ind(\Ic)+1-q\equiv\frac{1}{\pi}\theta\mod 2\quad\Leftrightarrow\quad R(\Ic(V))\in\sigma\left(\nHcV(\theta)\right).
\]
\item \label{enu: lem-Floquet-Bloch matrices - n times fundamental domain - 3}
If $L(\Ic(V))\in\sigma\left(\nHcV(\theta)\right)$, then 
\[
n\cdot\ind(\Ic)=N\left(L(\Ic(V));\nHcV(\theta)\right)
\]
and there exists $\lambda\in\sigma\left(\nHcV(\pi-\theta)\right)$
such that 
\[
L(\Ic(V))<\lambda\leq R(\Ic(V))\quad\textrm{and}\quad N\left(\lambda;\nHcV(\pi-\theta)\right)=N\left(L(\Ic(V));\nHcV(\theta)\right).
\]
\item \label{enu: lem-Floquet-Bloch matrices - n times fundamental domain - 4}
If $R(\Ic(V))\in\sigma\left(\nHcV(\theta)\right)$, then 
\[
n\cdot\left(\ind(\Ic)+1\right)-1=N\left(R(\Ic(V));\nHcV(\theta)\right)
\]
and for $n\geq2$, there exists $\lambda\in\sigma\left(\nHcV(\pi-\theta)\right)$
such that 
\[
L(\Ic(V))\leq\lambda<R(\Ic(V))\quad\textrm{and}\quad N\left(\lambda;\nHcV(\pi-\theta)\right)=N\left(R(\Ic(V));\nHcV(\theta)\right)-1.
\]
\item \label{enu: lem-Floquet-Bloch matrices - n times fundamental domain - 5}We
have $\left|\left\{ \lambda\in\sigma\left(\nHcV(\theta)\right)\cap\Ic(V)\right\} \right|=n$. If $\lambda\in\sigma\left(\nHcV(\theta)\right)\cap\left\{ L(\Ic(V)),R(\Ic(V))\right\} $,
then $\lambda$ is a simple eigenvalue of $\nHcV(\theta)$.
\end{enumerate}
\end{lem}

\begin{proof}
Recall from (\ref{eq: Spectrum union thetas - n times potential})
that the spectrum $\sigc(V)$ is given as the union of the eigenvalues
of $\nHcV(\theta)$ over all $\theta\in[0,\pi]$. Denote by $\lambda^{(\theta)}_{j}:=\lambda_{j}\left(\nHcV(\theta)\right)$
for $0\leq j\leq nq-1$ the eigenvalues of $\nHcV(\theta)$ in increasing
order counting multiplicities, see (\ref{eq: Eigenvalues matrix - increasing order}).
These eigenvalues for $\theta\in\{0,\pi\}$ are arranged as follows,
\begin{equation}
\ldots\leq\lambda^{(\pi)}_{nq-4}\leq\lambda^{(\pi)}_{nq-3}<\lambda^{(0)}_{nq-3}\leq\lambda^{(0)}_{nq-2}<\lambda^{(\pi)}_{nq-2}\leq\lambda^{(\pi)}_{nq-1}<\lambda^{(0)}_{nq-1},\label{eq: lem-Floquet-Bloch matrices - n times fundamental domain - eigenvalues}
\end{equation}
noting that the strict inequalities above appear whenever we compare
eigenvalues with different $\theta$ values (see e.g. \cite[Eq. (25)]{Hochstadt-LAA_1975}).
We use these eigenvalues to recursively define the following intervals
\[
\ldots,~J_{l}:=[\lambda^{(\theta_{l})}_{l},\lambda^{(\pi-\theta_{l})}_{l}],~\ldots,~J_{nq-2}:=[\lambda^{(0)}_{nq-2},\lambda^{(\pi)}_{nq-2}],~J_{nq-1}:=[\lambda^{(\pi)}_{nq-1},\lambda^{(0)}_{nq-1}],
\]
for appropriately chosen $\theta_{l}\in\{0,\pi\}$. We note that these
intervals are ordered, i.e. $J_{l}\prec J_{l+1}$ for all $0\leq l\leq nq-2$.

We now make a connection between these intervals, and the spectral
bands $\Ic$ of $\sigc$. By Proposition~\ref{prop: Basic spectral prop periodic}
and Lemma~\ref{lem: ConnecSpectrTrace}, $\sigc(V)$ consists of
exactly $q$ disjoint intervals - called spectral bands. For each
such spectral band $\Ic$ of $\sigc$, set $j=\ind(\Ic)$ and $I_{j}:=\Ic(V)$
for the given $V>0$.

We show a few \uline{auxiliary claims}, and then use them to prove
the statements in the lemma.
\begin{itemize}
\item[(1)]  For all $0\leq l\leq nq-1$, the endpoints $L(J_{l})$ and $R(J_{l})$
correspond to eigenvalues with different $\theta$ values. Moreover,
$R(J_{l})$ and $L(J_{l+1})$ correspond to the same value of $\theta\in\{0,\pi\}$
for all $0\leq l\leq nq-2$.
\item[(2)]  The equalities 
\[
\sigc(V)=\bigcup^{q-1}_{j=0}I_{j}=\bigcup^{nq-1}_{l=0}J_{l}\text{\ensuremath{\quad\textrm{and}\quad I_{q-1-j}=\bigcup^{n-1}_{l=0}J_{nq-1-nj-l}}}~\textrm{ for all }0\leq j\leq q-1
\]
hold.
\item[(3)]  For $\theta\in\{0,\pi\}$, each $I_{j}$ contains exactly $n$ eigenvalues
of $\sigma\left(\nHcV(\theta)\right).$
\item[(4)]  We have $R(I_{q-1})=\lambda^{(0)}_{nq-1}$.
\end{itemize}
Claim (1) is immediate from the definition of the intervals $J_{l}$.
The first equality in (2) follows from (\ref{eq: Spectrum union thetas})
and (\ref{eq: Spectrum union thetas - n times potential}). Thus,
each $I_{j}$ is the union of some of the consecutive intervals $J_{l}$.
By \cite[Theorem 1]{Hochstadt-LAA_1984} each $n$ consecutive $J_{l}$
bands touch (so that their union is a single connected component)
and this inductively implies the second equality in (2). This also
implies (3). To deduce (4) we combine $I_{q-1}=\bigcup^{n-1}_{l=0}J_{nq-1-l}$
(which follows from (2)) with $J_{nq-1}:=[\lambda^{(\pi)}_{nq-1},\lambda^{(0)}_{nq-1}]$.

We now use the claims above to prove the different statements of the
lemma.

(\ref{enu: lem-Floquet-Bloch matrices - n times fundamental domain - 1}):
The claims (1) and (2) for even $n\in\N$ imply that the left and
right spectral edges of $\Ic(V)$ correspond to the same value $\theta\in\{0,\pi\}$.
Combining this with claim (4) implies that all spectral edges of $\Ic$
correspond to the value $\theta=0$.

(\ref{enu: lem-Floquet-Bloch matrices - n times fundamental domain - 2}):
The claims (1) and (2) for odd $n\in\N$ imply that the left and right
spectral edge of $\Ic(V)$ correspond to a different value of $\theta\in\{0,\pi\}$.
Hence, the value of $\theta\in\{0,\text{\ensuremath{\pi\}}}$ which
corresponds to $L(I_{j})$ alternates with $j$ (and it also alternates
for $R(I_{j})$). Combining this with claim (4) yields the statement
in (\ref{enu: lem-Floquet-Bloch matrices - n times fundamental domain - 2}).

(\ref{enu: lem-Floquet-Bloch matrices - n times fundamental domain - 3})
and (\ref{enu: lem-Floquet-Bloch matrices - n times fundamental domain - 4}):
 The first equality in (\ref{enu: lem-Floquet-Bloch matrices - n times fundamental domain - 3})
and (\ref{enu: lem-Floquet-Bloch matrices - n times fundamental domain - 4})
follows from claim (3). Note that for (\ref{enu: lem-Floquet-Bloch matrices - n times fundamental domain - 4})
the quantity $N\left(R(\Ic(V));\nHcV(\theta)\right)$ counts $n-1$
eigenvalues in the spectral band $\Ic$(V) and $n$ eigenvalues for
each spectral band $I(V)\prec\Ic(V)$ (which are $\ind(\Ic)$ many).

\begin{figure}
\includegraphics[scale=0.87]{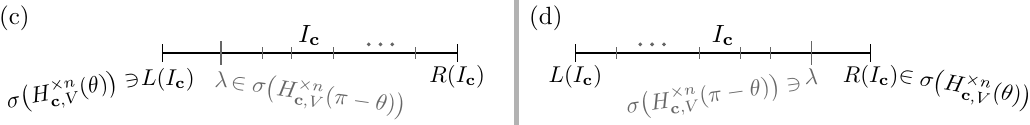}
\caption{A sketch for the proof of (c) and (d) in Lemma~\ref{lem: Floquet-Bloch matrices - n times fundamental domain}.}
\label{fig: Floquet-Bloch matrices - n times fundamental domain}
\end{figure}

We turn to prove the second claim in (\ref{enu: lem-Floquet-Bloch matrices - n times fundamental domain - 3}).
It follows from claim (2) that there exists $\theta\in\left\{ 0,\pi\right\} $
such that $L(\Ic(V))\in\sigma\left(\nHcV(\theta)\right).$ Using the
notation for the eigenvalues of $\nHcV(\theta)$ introduced in the
beginning of the proof, we can write $\lambda^{(\theta)}_{l}:=L(\Ic(V))$,
for some $0\leq l\leq nq-1$. Now, we define 
\[
\lambda:=\lambda^{(\pi-\theta)}_{l}=\min\set{\tilde{\lambda}}{\tilde{\lambda}\in\sigma\left(\nHcV(\pi-\theta)\right)\quad\textrm{and}\quad\tilde{\lambda}>\lambda^{(\theta)}_{l}:=L(\Ic(V))}
\]
 (as sketched in Figure~\ref{fig: Floquet-Bloch matrices - n times fundamental domain},(c))
and show that this is the desired $\lambda\in\sigma\left(\nHcV(\pi-\theta)\right)$
in the statement of (\ref{enu: lem-Floquet-Bloch matrices - n times fundamental domain - 3}).
By the construction in the beginning of the proof we get that $L(\Ic(V))<\lambda$
and $J_{l}=\left[L(\Ic(V)),\lambda\right]$. Furthermore, $J_{l}$
is the left-most sub-interval within $\Ic(V)$, as in the decomposition
of claim (2). Hence, $L(\Ic(V))<\lambda\leq R(\Ic(V))$, as stated
in (\ref{enu: lem-Floquet-Bloch matrices - n times fundamental domain - 3}).
To complete the proof of (\ref{enu: lem-Floquet-Bloch matrices - n times fundamental domain - 3})
we just note that $N\left(L(\Ic(V));\nHcV(\theta)\right)=l$, just
by the choice of $0\leq l\leq nq-1$ and similarly $N\left(\lambda;\nHcV(\pi-\theta)\right)=l$.
Hence, $N\left(\lambda;\nHcV(\pi-\theta)\right)=N\left(L(\Ic(V));\nHcV(\theta)\right)$.

It is left to prove the second claim in (\ref{enu: lem-Floquet-Bloch matrices - n times fundamental domain - 4}).
This follows similarly as in (\ref{enu: lem-Floquet-Bloch matrices - n times fundamental domain - 3}).
First, there exists $\theta\in\left\{ 0,\pi\right\} $ such that $R(\Ic(V))\in\sigma\left(\nHcV(\theta)\right)$;
we write $\lambda^{(\theta)}_{l}:=R(\Ic(V))$, for some $0\leq l\leq nq-1$;
we define
\[
\lambda:=\lambda^{(\pi-\theta)}_{l}=\max\set{\tilde{\lambda}}{\tilde{\lambda}\in\sigma\left(\nHcV(\pi-\theta)\right)\quad\textrm{and}\quad\tilde{\lambda}<\lambda^{(\theta)}_{l}=R(\Ic(V))}
\]
 (as sketched in Figure~\ref{fig: Floquet-Bloch matrices - n times fundamental domain}).
Then $L(\Ic(V))\leq\lambda<R(\Ic(V))$ holds. If $n\geq2$, then $\lambda$
is in the interior of $\Ic(V)$ and so the eigenvalue $\lambda\in\sigma\left(\nHcV(\pi-\theta)\right)$
has multiplicity two by (\ref{eq: lem-Floquet-Bloch matrices - n times fundamental domain - eigenvalues})
and claim (2). Thus, $N\left(\lambda;\nHcV(\theta)\right)$ counts
$n-2$ eigenvalues in $\Ic(V)$ while $N\left(R(\Ic(V));\nHcV(\theta)\right)$
counts $n-1$ eigenvalues in $\Ic(V)$. Hence, $N\left(\lambda;\nHcV(\pi-\theta)\right)=N\left(R(\Ic(V));\nHcV(\theta)\right)-1$
follows proving (\ref{enu: lem-Floquet-Bloch matrices - n times fundamental domain - 4}).

(\ref{enu: lem-Floquet-Bloch matrices - n times fundamental domain - 5})
This is an immediate consequence of claim (2) and (\ref{eq: lem-Floquet-Bloch matrices - n times fundamental domain - eigenvalues}).
\end{proof}

\subsection{A characterization of admissibility\label{subsec: Characterization Admissible}}

We recall the definition of admissibility (Definition~\ref{def: admissibility})
for a triple of eigenvalues. We now use the lemmata of the previous
subsection in order to provide an equivalent condition for admissibility.
Since the definition of admissibility is independent of $V>0$ (as
is also mentioned within Definition~\ref{def: admissibility}), we
omit the $V$-dependence from the notation in this subsection. For
example, we write $\Ic,\lambda_{\co}$ and $\nHc$ instead of writing
$\Ic(V),\lambda_{\co}(V)$ and $\nHcV$.
\begin{lem}
\label{lem: Admissibility criterion} Let $m,n\in\N$, and $\co\in\Co$
be such that $\varphi(\co)\not\in\left\{ -1,\infty\right\} $ and
$[\co,m]\in\Co$. For each $\cop\in\left\{ \co,~[\co,m],~[\co,m,n]\right\} $,
let $I_{\cop}$ be a spectral band in $\sigma_{\cop}$ and $\lambda_{\cop}\in\left\{ L(I_{\cop}),R(I_{\cop})\right\} $
and denote 
\[
\dR(\lambda_{\cop}):=\begin{cases}
0,\qquad & \lambda_{\cop}=L\big(I_{\cop}\big),\\
1,\qquad & \lambda_{\cop}=R\big(I_{\cop}\big).
\end{cases}
\]
Then $\lc,\lcm,\lcmn$ are admissible if and only if 
\[
\ind\big(I_{\co}\big)+n\cdot\ind\big(I_{[\co,m]}\big)+\ind\big(I_{[\co,m,n]}\big)\equiv\dR(\lc)+n\cdot\dR(\lcm)+\dR(\lcmn)\mod 2.
\]
\end{lem}

\begin{proof}
Let $\cop\in\Co$ be such that $\left\{ -1,\infty\right\} \not\ni\varphi(\cop)=\frac{p_{\cop}}{q_{\cop}}$
with $p_{\cop},q_{\cop}$ coprime. Let $I_{\cop}$ be a spectral band
of $\sigma_{\cop}$ and $\lambda_{\cop}$ an edge (left or right)
of $I_{\cop}$. In particular, by Lemma~\ref{lem: Floquet-Bloch matrices}
(\ref{enu: lem-Floquet-Bloch matrices - 1}) $\lambda_{\cop}$ is
an eigenvalue in $H_{\cop}(\theta_{\cop})$ for some $\theta_{\cop}\in\left\{ 0,\pi\right\} $
and 
\begin{equation}
\ind(I_{\cop})+\dR(\lambda_{\cop})-q_{\cop}\equiv\frac{1}{\pi}\theta_{\cop}\mod 2.\label{eq: 0 or pi for Ic endpoint}
\end{equation}

We will apply (\ref{eq: 0 or pi for Ic endpoint}) in the following
for both $\cop=\co$ and $\cop=[\co,m,n]$. However, recall from the
admissibility definition (Definition~\ref{def: admissibility}) that
we need to consider $\lcm$ as an eigenvalue of the matrix $\nHcm(\thcm)$
(rather than the matrix $\Hcm(\thcm)$). Therefore, we need to develop
an alternative identity to (\ref{eq: 0 or pi for Ic endpoint}). This
is done with the aid of Lemma~\ref{lem: Floquet-Bloch matrices - n times fundamental domain}
(\ref{enu: lem-Floquet-Bloch matrices - n times fundamental domain - 1})
and (\ref{enu: lem-Floquet-Bloch matrices - n times fundamental domain - 2})
from which we conclude that 
\begin{equation}
n\cdot\left(\ind(\Icm)+\dR(\lcm)-\qcm\right)\equiv\frac{1}{\pi}\thcm\mod 2,\label{eq: 0 or pi for Ic.m endpoint}
\end{equation}
for both even and odd values of $n\in\N$.

To conclude the proof we sum Equation (\ref{eq: 0 or pi for Ic endpoint})
for $\cop=\co$ and for $\cop=\cmn$ and we add to it Equation (\ref{eq: 0 or pi for Ic.m endpoint}).
This yields 
\begin{align*}
\left(\ind(I_{\co})+\dR(\lambda_{\co})\right)+ & \left(\ind(\Icmn)+\dR(\lcmn)\right)\\
+n \left(\ind(\Icm)+\dR(\lcm)\right) & -\left(\qc+n\ \qcm+\qcmn\right)\equiv\frac{1}{\pi}\left(\thc+\thcm+\thcmn\right)\mod 2.
\end{align*}

By definition, admissibility of $\lc,\lcm,\lcmn$ is equivalent to
admissibility of the values $\thc,\thcm,\thcmn\in\{0,\pi\}$, which
is equivalent to $\frac{1}{\pi}\left(\thc+\thcm+\thcmn\right)\equiv0\mod 2$.
To end the proof, we just substitute above the equality $\mbox{\ensuremath{\qc+n\cdot\qcm=\qcmn}}$,
which is standard in the theory of finite continued fraction expansions
(see Lemma~\ref{lem: periods of three approximants}, (\ref{enu: lem-periods of three approximants-2})).
\end{proof}

\subsection{Eigenvalue inequalities resulting from interlacing theorem\label{subsec: counting-function and eigenvalue estimates}}

Combining the interlacing theorem (Theorem \ref{thm: perturbation thm for our matrices})
with Lemma \ref{lem: Floquet-Bloch matrices - n times fundamental domain}
gives the following useful lemma, which is applied many times in the
following subsections.
\begin{lem}
\label{lem: Basics for eigenavlue inequalities} Let $V>0$, $m,n\in\N$,
$\co\in\Co$ be such that $\varphi(\co)\not\in\left\{ -1,\infty\right\} $
and $[\co,m]\in\Co$. Let $\thc,\thcm,\thcmn\in\left\{ 0,\pi\right\} $
and
\[
\lo\in\sigma\left(\nHcmV(\thcm)\oplus\HcV(\thc)\right)\quad\textrm{and\ensuremath{\quad\mo}\ensuremath{\in\sigma}}\left(\HcmnV(\thcmn)\right).
\]
Define
\[
N_{\co}:=N\left(\lo;~\HcV(\thc)\right),\quad N_{[\co,m]}:=N\left(\lo;~\nHcmV(\thcm)\right)
\]
and
\[
N_{[\co,m,n]}:=N\left(\mo;~\HcmnV(\thcmn)\right).
\]
\begin{enumerate}
\item \label{enu: lem-Basics for eigenavlue inequalities - 1} Let $\Mult_{\lo}$
be the multiplicity of the eigenvalue $\lo$ of $\nHcmV(\thcm)\oplus\HcV(\thc)$.
If $\thc,\thcm,\thcmn$ are admissible, then the following implications
hold:
\begin{equation}
N_{\co}+N_{[\co,m]}<N_{[\co,m,n]}\quad\Rightarrow\quad\lo\leq\mo,\label{eq: lem-Basics for eigenavlue inequalities - 1}
\end{equation}
\begin{equation}
N_{\co}+N_{[\co,m]}+\Mult_{\lo}-1>N_{[\co,m,n]}\quad\Rightarrow\quad\lo\geq\mo.\label{eq: lem-Basics for eigenavlue inequalities - 2}
\end{equation}
If, additionally, $\lo$ is a simple eigenvalue of $\nHcmV(\thcm)\oplus\HcV(\thc)$
(i.e., $\Mult_{\lo}=1)$, then the two inequalities on the right hand
sides of (\ref{eq: lem-Basics for eigenavlue inequalities - 1}) and
(\ref{eq: lem-Basics for eigenavlue inequalities - 2}) are strict.
\item \label{enu: lem-Basics for eigenavlue inequalities - 2} If
\begin{itemize}
\item $\thc,\thcm,\thcmn$ are not admissible and
\item $\Icm$ is a spectral band in $\sigcm$ satisfying $\sigma\left(\HcV(\thc)\right)\cap\Icm(V)=\emptyset$,
\end{itemize}
then the following implications hold:
\begin{equation}
\lo=L(\Icm(V)),\quad N_{\co}+N_{[\co,m]}<N_{[\co,m,n]}\quad\Rightarrow\quad\lo<\mo\label{eq: lem-Basics for eigenavlue inequalities - 3}
\end{equation}
and for $n\geq2$,
\begin{equation}
\lo=R(\Icm(V)),\quad N_{\co}+N_{[\co,m]}-1>N_{[\co,m,n]}\quad\Rightarrow\quad\lo>\mo.\label{eq: lem-Basics for eigenavlue inequalities - 4}
\end{equation}

\end{enumerate}
\end{lem}

\begin{rem*}
We emphasize that $\lo$ and $\mo$ do depend on $V$, but the implications
of the lemma do not.
\end{rem*}
\begin{proof}
We start by noting the following rather trivial counting relation
\begin{equation}
N\left(\lo;~\nHcmV(\thcm)\oplus\HcV(\thc)\right)=N_{\co}+N_{[\co,m]}.\label{eq: lem-Basics for eigenavlue inequalities - counting relation}
\end{equation}
(\ref{enu: lem-Basics for eigenavlue inequalities - 1}) Suppose that
$\thc,\thcm,\thcmn$ are admissible. Both of the required implications
(\ref{eq: lem-Basics for eigenavlue inequalities - 1}) and (\ref{eq: lem-Basics for eigenavlue inequalities - 2})
follow from Theorem~\ref{thm: perturbation thm for our matrices},
when keeping in mind the counting relation (\ref{eq: lem-Basics for eigenavlue inequalities - counting relation}).
Explicitly, denoting the eigenvalues of $\nHcmV(\thcm)\oplus\HcV(\thc)$
by $\lambda_{0}\leq\lambda_{1}\leq\lambda_{2}\leq\ldots$ in increasing
order, we get that $\lo=\lambda_{\Nc+\Ncm}=\ldots=\lambda_{\Nc+\Ncm+M_{\lo}-1}$.
Similarly $\mo=\mu_{\Ncmn}$, if the eigenvalues of $\HcmnV(\thcmn)$
are denoted by $\mu_{0}\leq\mu_{1}\leq\mu_{2}\leq\ldots$ in increasing
order. Hence,
\begin{itemize}
\item (\ref{eq: lem-Basics for eigenavlue inequalities - 1}) follows by
applying Theorem~\ref{thm: perturbation thm for our matrices} for
$\lo=\lambda_{\Nc+\Ncm}$, $\mo=\mu_{\Ncmn}$, and
\item (\ref{eq: lem-Basics for eigenavlue inequalities - 2}) follows by
applying Theorem~\ref{thm: perturbation thm for our matrices} for
$\lo=\lambda_{\Nc+\Ncm+\Mult_{\lo}-1}$, $\mo=\mu_{\Ncmn}$.
\end{itemize}
If, additionally, $\lo$ is a simple eigenvalue of $\nHcmV(\thcm)\oplus\HcV(\thc)$,
then $\Mult_{\lo}=1$ and the relevant statement within Theorem~\ref{thm: perturbation thm for our matrices}
yields the corresponding strict inequalities.

(\ref{enu: lem-Basics for eigenavlue inequalities - 2}) Suppose that
$\thc,\thcm,\thcmn$ are not admissible and let $\Icm$ be a spectral
band in $\sigcm$ satisfying $\sigma\left(\HcV(\thc)\right)\cap\Icm(V)=\emptyset$.

In the first case (Equation~(\ref{eq: lem-Basics for eigenavlue inequalities - 3})),
we assume $\lo=L(\Icm(V))$ and $N_{\co}+N_{[\co,m]}<N_{[\co,m,n]}$.
We aim to apply Theorem~\ref{thm: perturbation thm for our matrices}
directly but $\thc,\thcm,\thcmn$ are not admissible. Thus, we change
one of these values to attain an admissible triple. More precisely,
$\thc,\pi-\thcm,\thcmn$ are admissible, see Remark~\ref{rem:Admissible=00003DEvenNumberPis}.
By Lemma~\ref{lem: Floquet-Bloch matrices - n times fundamental domain}~(\ref{enu: lem-Floquet-Bloch matrices - n times fundamental domain - 3}),
there exists a $\lambda\in\nHcmV(\pi-\thcm)$ such that
\begin{align*}
 & \lo=L(\Icm(V))<\lambda\leq R(\Icm(V))\\
\textrm{and}\quad & N\left(\lambda;\nHcmV(\pi-\thcm)\right)=N\left(\lo;\nHcmV(\thcm)\right).
\end{align*}

 Thus, $\lambda\in\Icm(V)$ and $\sigma\left(\HcV(\thc)\right)\cap\Icm(V)=\emptyset$
lead to the equation $N\left(\lambda;~\HcV(\thc)\right)=N\left(\lo;~\HcV(\thc)\right)$.
Therefore, (\ref{eq: lem-Basics for eigenavlue inequalities - counting relation})
implies
\begin{align*}
N\left(\lambda;~\nHcmV(\pi-\thcm)\oplus\HcV(\thc)\right) & =N_{\co}+N_{[\co,m]}.
\end{align*}
Since $N_{\co}+N_{[\co,m]}<N_{[\co,m,n]}$ and $\thc,\pi-\thcm,\thcmn$
are admissible, Theorem~\ref{thm: perturbation thm for our matrices}
yields $\lambda\leq\mo$. Using $\lo<\lambda$, we conclude $\lo<\mo$,
as claimed.

The second case (Equation~(\ref{eq: lem-Basics for eigenavlue inequalities - 4}))
follows similar arguments, using Lemma~\ref{lem: Floquet-Bloch matrices - n times fundamental domain}~(\ref{enu: lem-Floquet-Bloch matrices - n times fundamental domain - 4}).
\end{proof}

\subsection{Index identities of the spectral bands\label{subsec: Index-identities}}

In order to apply Lemma~\ref{lem: Basics for eigenavlue inequalities},
we need to be able to compare the spectral positions of $\lo$ and
$\mo$ ($\Nc$, $\Ncm$ and $\Ncmn$) which appear in Lemma~\ref{lem: Basics for eigenavlue inequalities}.
We have already seen (Lemma~\ref{lem: Floquet-Bloch matrices - n times fundamental domain})
that such spectral positions are connected to spectral band indices.
Hence, towards applying Lemma~\ref{lem: Basics for eigenavlue inequalities},
we develop in this subsection some connections between indices of
spectral bands (Lemma~\ref{lem: index identities for spectral bands}).
For the upcoming statements and proofs, we introduce the following
notations (see Figure~\ref{fig: Jcm and Kcm} for a sketch).
\begin{defn}
\label{def: left and right proceeded band} Let $m\in\N_{0}$ and
$\co,\cm\in\Co$ be such that $\varphi(\co)\not\in\left\{ -1,\infty\right\} $
and $\varphi([\co,m])\not\in\left\{ -1,\infty\right\} $. For a spectral
band $\Ic$ in $\sigc$, define the \emph{associated spectral bands}
$\Jcm,~\Kcm$ in $\sigcm$ to be the unique spectral bands (if they
exist) such that for $V>4$,
\begin{itemize}
\item $\Jcm(V)$ is the right-most band of $\sigcm(V)$ for which $\Jcm(V)\prec\Ic(V)$,
and
\item $\Kcm(V)$ is the left-most band of $\sigcm(V)$ for which $\Ic(V)\prec\Kcm(V)$.
\end{itemize}
\end{defn}

\begin{rem*}
Note that it might be that some of the bands $\Jcm$ and $\Kcm$ do
not exist. In such a case, this in an empty convention. Further note
that $\varphi([\co,m])\in\left\{ -1,\infty\right\} $ for $m\in\N_{0}$
can only happen if $\co=[0,0]$ and $m=0$ in which case such spectral
bands $\Jcm$ and $\Kcm$ do not exist.
\end{rem*}
The reason for including $V>4$ in the definition above is explained
in the beginning of the proof of Lemma~\ref{lem: index identities for spectral bands}.

\begin{figure}[h]
\includegraphics[scale=0.9]{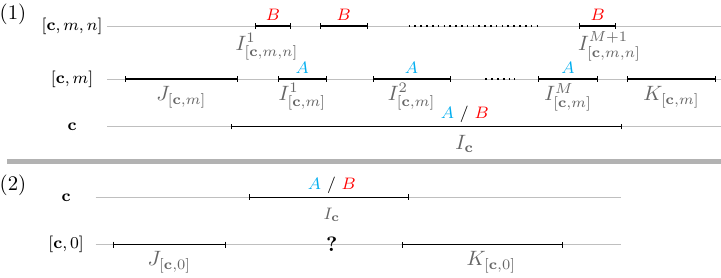} \caption{A sketch for Definition~\ref{def: left and right proceeded band}
and Lemma~\ref{lem: index identities for spectral bands}. (1) For
a fixed spectral band $\protect\Ic$ and $m\protect\geq1$, the associated
bands $\protect\Jcm$ and $\protect\Kcm$ (from Definition~\ref{def: left and right proceeded band})
and the associated bands $\{\protect\Icmi\}^{M}_{i=1}$ and $\{\protect\Icmnj\}^{M+1}_{j=1}$
(from Definition~\ref{def: Notion Icmi. Icmnj and M}) are drawn.
(2) For a fixed spectral band $\protect\Ic$ and $m=0$, the associated
bands $\protect\Jcz$ and $\protect\Kcz$ (from Definition~\ref{def: left and right proceeded band})
are drawn. If $\protect\Ic$ is of backward type $\protect\tA$ for
$V>4$, there is a spectral band between $J_{[\protect\co,0]}$ and
$K_{[\protect\co,0]}$. Otherwise, there is no spectral band between
them, namely $\protect\ind(K_{[\protect\co,0]})=\protect\ind(J_{[\protect\co,0]})+1$
(indicated by the question mark).}
\label{fig: Jcm and Kcm}
\end{figure}

\begin{lem}
\label{lem: index identities for spectral bands} Let $m,n\in\N$,
and $\co\in\Co$ be such that $\varphi(\co)\not\in\left\{ -1,\infty\right\} $
and $[\co,m]\in\Co$. Consider a spectral band $I_{\co}$ in $\sigma_{\co}$
with associated spectral bands $\{\Icmi\}^{M}_{i=1}$ and $\{\Icmnj\}^{M+1}_{j=1}$
introduced in Definition~\ref{def: Notion Icmi. Icmnj and M}. If
$M\geq1$, then for all $1\leq i\leq M$
\begin{equation}
\ind(\Icmni)=n\cdot\ind(\Icmi)+\ind(\Ic),\label{eq: Cor-unifyrecform - index relation for I_=00007Bc.m.n=00007D^=00007Bi=00007D}
\end{equation}
and
\begin{equation}
\ind(\Icmnii)=n\cdot(\ind(\Icmi)+1)+\ind(\Ic).\label{eq: Cor-unifyrecform - index relation for I_=00007Bc.m.n=00007D^=00007Bi+1=00007D}
\end{equation}
Whenever the spectral bands $\Jcm$ or $\Kcm$ associated with $\Ic$
exist, then the following hold. If $M\geq0$\textup{, }then 
\begin{align}
\ind(\Icmn^{1}) & =n\cdot(\ind(\Jcm)+1)+\ind(\Ic)\label{eq: Cor-unifyrecform - index relation for I_=00007Bc.m.n=00007D^1}
\end{align}
and
\begin{equation}
\ind\left(\Icmn^{M+1}\right)=n\cdot\ind(\Kcm)+\ind(\Ic)\label{eq: Cor-unifyrecform - index relation for I_=00007Bc.m.n=00007D^M+1}
\end{equation}
If $\Ic(V)$ is of type $B$ for $V>4$, then
\begin{equation}
\ind(\Ico)=\ind(\Ic)+\ind(\Jcz)+1=\ind(\Ic)+\ind(K_{[\co,0]}).\label{eq: Cor-unifyrecform - index relation for I_=00007Bc.1=00007D^=00007B1=00007D}
\end{equation}
\end{lem}

\begin{proof}
We start by noting that the index of a spectral band is independent
of $V>0$ (Remark~\ref{rem: index spectral band}) allowing us to
restrict to the case $V>4$ where all spectral bands are either of
type $A$ or of type \textbf{$B$ }by Theorem~\ref{thm: V>4 Type A =000026 B}.
Therefore, within this proof we allow ourselves to assume $V>4$ and
abuse notation, writing just $I$ (meaning interval and not a map)
instead of writing $I(V)$ for some $V>4$. Namely, when writing within
this proof sentences such as ``$I$ is a spectral band of type $A$
(or $B$) and belongs to $\sigc$'', we actually mean that for some
value of $V>4$, $I(V)$ is of type $A$ (or $B$) and belongs to
$\sigc(V)$.

We introduce the following extra notations for the band indices: 
\begin{align*}
\indA(\Ic) & :=\left|\set{I\textrm{ is of type \ensuremath{A} and it belongs to }\sigc}{I\prec\Ic}\right|,\\
\indB(\Ic) & :=\left|\set{I\textrm{ is of type \ensuremath{B} and it belongs to }\sigc}{I\prec\Ic}\right|.
\end{align*}
Clearly, $\ind(\Ic)=\indA(\Ic)+\indB(\Ic)$ for all $\Ic$, see Definition~\ref{def: index of spectral band}.
\begin{figure}[h]
\includegraphics[scale=0.7]{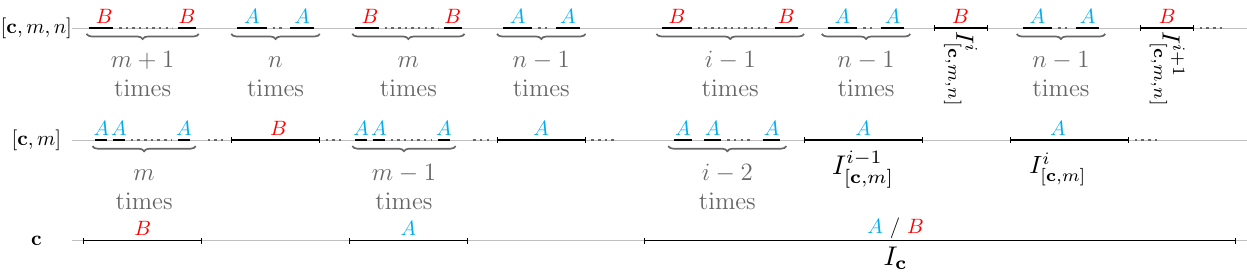} \caption{A sketch for the proof of (\ref{eq: Cor-unifyrecform - index relation for I_=00007Bc.m.n=00007D^=00007Bi=00007D})
and (\ref{eq: Cor-unifyrecform - index relation for I_=00007Bc.m.n=00007D^=00007Bi+1=00007D})
in Lemma~\ref{lem: index identities for spectral bands}.}
\label{fig: Index identities proof: 1st}
\end{figure}

We first assume that $M\geq1$. Start by examining $\Icmni$ and evaluating
$\indB(\Icmni)$ and $\indA(\Icmni)$. The spectral band $\Icmni$
is of type $B$ and belongs to $\sigcmn$. We know that $\Icmni$
is included in $\Ic$ of $\sigc$. There are additional $i-1$ spectral
bands of type $B$ in $\sigcmn$, which are to the left of $\Icmni$
and included in $\Ic$. All other spectral bands of type $B$ to the
left of $\Icmni$ come in groups of either $m$ or $m+1$ and each
such group is included in some spectral band $I$ in $\sigc$ that
is to the left of $\Ic$, see Figure~\ref{fig: Index identities proof: 1st}.
The group is of size $m$ if $I$ is of type $A$ and it is of size
$m+1$ if $I$ is of type $B$. This discussion may be summarized
in the following identity
\begin{equation}
\indB(\Icmni)=m\cdot\indA(\Ic)+(m+1)\cdot\indB(\Ic)+i-1.\label{eq: indB(I_=00007Bcmn=00007D^=00007Bi=00007D}
\end{equation}
We now evaluate $\indA(\Icmni)$. We note that all the spectral bands
of type $A$ to the left of $\Icmni$ come in groups of either $n-1$
or $n$ and each such group is included in some spectral band $I$
in $\sigcm$ that is to the left of $\Icmi$, see Figure~\ref{fig: Index identities proof: 1st}.
The group is of size $n-1$ if $I$ is of type $A$ and it is of size
$n$ if $I$ is of type $B$. This discussion may be summarized in
the following identity
\begin{equation}
\indA(\Icmni)=(n-1)\cdot\indA(\Icmi)+n\cdot\indB(\Icmi).\label{eq: indA(I_=00007Bcmn=00007D^=00007Bi=00007D}
\end{equation}
We now evaluate $\indA(\Icmi)$. We note that there are $i-1$ spectral
bands of type $A$ to the left of $\Icmi$ which are included in $\Ic$.
Every other spectral band in $\sigcm$ of type $A$ to the left of
$\Icmi$ is included in a spectral band of $\sigc$ to the left of
$\Ic$. Specifically, they come in groups of either $m-1$ or $m$
and each group is included in a spectral band $I$ in $\sigc$ to
the left of $\Ic$. The group is of size $m-1$ if $I$ is of type
$A$ and it is of size $m$ if $I$ is of type $B$, see Figure~\ref{fig: Index identities proof: 1st}.
This discussion may be summarized in the following identity
\begin{equation}
\indA(\Icmi)=(m-1)\cdot\indA(\Ic)+m\cdot\indB(\Ic)+i-1.\label{eq: indA(I_=00007Bcm=00007D^=00007Bi=00007D)}
\end{equation}
Combining the Equations~(\ref{eq: indB(I_=00007Bcmn=00007D^=00007Bi=00007D})
and (\ref{eq: indA(I_=00007Bcmn=00007D^=00007Bi=00007D}) together
with the identity $\ind(I)=\indA(I)+\indB(I)$, which holds for all
$I$, gives
\begin{align*}
\ind(\Icmni) & =\indA(\Icmni)+\indB(\Icmni)\\
 & =\left(n\cdot\ind(\Icmi)-\indA(\Icmi)\right)+\left(m\cdot\indA(\Ic)+(m+1)\cdot\indB(\Ic)+i-1\right)\\
 & =n\cdot\ind(\Icmi)+ind(\Ic),
\end{align*}
where in the last line we used (\ref{eq: indA(I_=00007Bcm=00007D^=00007Bi=00007D)}).
This proves Equation~(\ref{eq: Cor-unifyrecform - index relation for I_=00007Bc.m.n=00007D^=00007Bi=00007D}).

To prove (\ref{eq: Cor-unifyrecform - index relation for I_=00007Bc.m.n=00007D^=00007Bi+1=00007D}),
we observe that between $\Icmni$ and $\Icmnii$ there are $n-1$
spectral bands of type $A$ (the bands which are contained in $\Icmi$)
and no spectral bands of type $B$, see Figure~\ref{fig: Index identities proof: 1st}.
We therefore get
\begin{align*}
\ind(\Icmnii) & =\left(\ind(\Icmni)+1\right)+(n-1)=n\cdot(\ind(\Icmi)+1)+ind(\Ic),
\end{align*}
which proves Equation (\ref{eq: Cor-unifyrecform - index relation for I_=00007Bc.m.n=00007D^=00007Bi+1=00007D}).

\begin{figure}[h]
\includegraphics[scale=0.8]{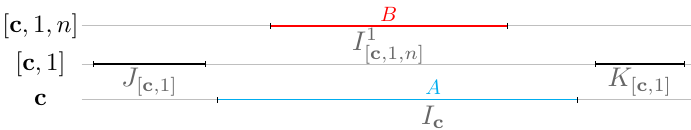} \caption{A sketch for the proof of (\ref{eq: Cor-unifyrecform - index relation for I_=00007Bc.m.n=00007D^1})
in Lemma~\ref{lem: index identities for spectral bands}. We have
$\protect\ind(\protect\Kcm)=\protect\ind(\protect\Jcm)+1$.}
\label{fig: Index identities proof: 2nd}
\end{figure}
For $M\geq1$, Equation~(\ref{eq: Cor-unifyrecform - index relation for I_=00007Bc.m.n=00007D^1})
follows from Equation~(\ref{eq: Cor-unifyrecform - index relation for I_=00007Bc.m.n=00007D^=00007Bi=00007D})
for $i=1$ and $\ind(\Icm^{1})=\ind(\Jcm)+1$. Similarly, Equation~(\ref{eq: Cor-unifyrecform - index relation for I_=00007Bc.m.n=00007D^M+1})
follows for $M\geq1$ from Equation~(\ref{eq: Cor-unifyrecform - index relation for I_=00007Bc.m.n=00007D^=00007Bi+1=00007D})
for $i=M$ and $\ind(\Kcm)=\ind(\Icm^{M})+1$.

For $M=0$, (\ref{eq: Cor-unifyrecform - index relation for I_=00007Bc.m.n=00007D^1})
and (\ref{eq: Cor-unifyrecform - index relation for I_=00007Bc.m.n=00007D^M+1})
follow similar arguments as (\ref{eq: Cor-unifyrecform - index relation for I_=00007Bc.m.n=00007D^=00007Bi=00007D})
and (\ref{eq: Cor-unifyrecform - index relation for I_=00007Bc.m.n=00007D^=00007Bi+1=00007D})
using $\ind(\Kcm)=\ind(\Jcm)+1$ if $M=0$. 

To prove (\ref{eq: Cor-unifyrecform - index relation for I_=00007Bc.1=00007D^=00007B1=00007D})
for the index of $\Ico$ we note the following. There is a bijection
between bands of type $A$ in $\sigma_{[\co,1]}$ and bands of type
$B$ in $\sigc$: a spectral band $\Ic$ in $\sigc$ of type $A$
does not contain any spectral band in $\sigma_{[\co,1]}$ but if $\Ic$
in $\sigc$ is of type $B$, then it contains exactly (using uniqueness
of these bands for $V>4$, see Theorem~\ref{thm: V>4 Type A =000026 B})
one band in $\sigma_{[\co,1]}$ of type $A$. Thus,
\[
\indA(\Ico)=\indB(\Ic)
\]
follows. We denote by $c_{k}$ the last digit in $\co$, namely, $\co:=[0,0,c_{1},\ldots,c_{k}]$.
Similar counting arguments as for (\ref{eq: indB(I_=00007Bcmn=00007D^=00007Bi=00007D})
lead to
\begin{align}
\indB(\Ico)= & c_{k}\cdot\indA(\Jcz)+(c_{k}+1)\cdot\indB(\Jcz)+c_{k}+\begin{cases}
0 & \Jcz\textrm{ is of type }A,\\
1 & \Jcz\textrm{ is of type }B,
\end{cases}\label{eq: indB(I_=00007Bcm=00007D^=00007Bi=00007D}
\end{align}
Moreover, similar counting arguments as in (\ref{eq: indA(I_=00007Bcm=00007D^=00007Bi=00007D)})
imply
\begin{align*}
\indA(\Ic) & =(c_{k}-1)\cdot\indA(\Jcz)+c_{k}\cdot\indB(\Jcz)+c_{k}-1+\begin{cases}
0 & \Jcz\textrm{ is of type }A,\\
1 & \Jcz\textrm{ is of type }B,
\end{cases}\\
 & =\indB(\Ico)-\ind(\Jcz)-1,
\end{align*}
where in the last line we used (\ref{eq: indB(I_=00007Bcm=00007D^=00007Bi=00007D}).
Hence, we arrive at 
\begin{align*}
\ind(\Ico)=\indB(\Ico)+\indA(\Ico) & =\left(\indA(\Ic)+\ind(\Jcz)+1\right)+\indB(\Ic)\\
 & =\ind(\Ic)+\ind(\Jcz)+1\\
 & =\ind(\Ic)+\ind(K_{[\co,0]})
\end{align*}
using $\ind(K_{[\co,0]})=\ind(\Jcz)+1$, which holds by definition.
Thus, (\ref{eq: Cor-unifyrecform - index relation for I_=00007Bc.1=00007D^=00007B1=00007D})
is proven.
\end{proof}

\subsection{Sufficient conditions for the forward type properties\label{subsec: Forward properties}}

In this subsection, we provide various lemmata and corollaries which
allow us to prove the various properties \ref{enu: A1 property}, \ref{enu: A2 property}, \ref{enu: B1 property}, \ref{enu: B2 property}, and \ref{enu: I property}
of the forward type as in Definition~\ref{def: forward type}. These
will be used in the next section, where we prove that backward type
implies forward type (Proposition~\ref{prop: Backward implies forward}).

We start with proving that the interlacing property \ref{enu: I property}
holds under some conditions.
\begin{lem}
\label{lem: Task (I)} Let $V_{1}>0$, $m,n\in\N$, $\co\in\Co$ be
such that $\varphi(\co)\not\in\left\{ -1,\infty\right\} $ and $[\co,m]\in\Co$.
Consider a spectral band $I_{\co}$ in $\sigma_{\co}$ with the associated
spectral bands $\{\Icmi\}^{M}_{i=1}$ and $\{\Icmnj\}^{M+1}_{j=1}$
introduced in Definition~\ref{def: Notion Icmi. Icmnj and M}. If
$1\leq i\leq M$ and $I^{i}_{[\co,m]}(V_{1})\strict I_{\co}(V_{1})$,
 then 
\[
I^{i}_{[\co,m,n]}(V_{1})\prec I^{i}_{[\co,m]}(V_{1})\prec I^{i+1}_{[\co,m,n]}(V_{1}).
\]
\end{lem}

\begin{proof}
Let $1\leq i\leq M$ and $V_{1}>0$. We need to show the following
inequalities
\begin{enumerate}
\item $L\left(\Icmni(V_{1})\right)<L\left(\Icmi(V_{1})\right)$,
\item $R\left(\Icmi(V_{1})\right)<R\left(\Icmnii(V_{1})\right)$,
\item $R\left(\Icmni(V_{1})\right)<R\left(\Icmi(V_{1})\right)$,
\item $L\left(\Icmi(V_{1})\right)<L\left(\Icmnii(V_{1})\right)$.
\end{enumerate}
We proceed proving the inequalities above one at a time via an appropriate
application of Lemma~\ref{lem: Basics for eigenavlue inequalities}.
Although the inequalities above and the assumption $I^{i}_{[\co,m]}(V_{1})\strict I_{\co}(V_{1})$
depend on the fixed $V_{1}>0$, we will abbreviate notation, for the
sake of easier reading, and omit the $V_{1}$ dependence in most parts
of this proof.

(a) We aim to apply Lemma~\ref{lem: Basics for eigenavlue inequalities}
for $\lo=L(\Icmi)$, $\mo=L(\Icmni)$. Let $\thc,\thcm,\thcmn\in\{0,\pi\}$
be such that
\[
L(\Ic)\in\sigma\left(\Hc(\thc)\right),\quad L(\Icmi)\in\text{\ensuremath{\sigma\left(\nHcm(\thcm)\right)}}\quad\textrm{and}\quad L(\Icmni)\in\ensuremath{\sigma\left(\Hcmn(\thcmn)\right)}.
\]
These spectral edges, respectively $\thc,\thcm,\thcmn$, are admissible,
as can be verified by using the index relation (\ref{eq: Cor-unifyrecform - index relation for I_=00007Bc.m.n=00007D^=00007Bi=00007D})
of Lemma~\ref{lem: index identities for spectral bands} in the characterization
of admissibility from Lemma~\ref{lem: Admissibility criterion}.
Furthermore, Lemma~\ref{lem: Floquet-Bloch matrices - n times fundamental domain}~(\ref{enu: lem-Floquet-Bloch matrices - n times fundamental domain - 3})
applied to $[\co,m]\in\Co$ implies 
\[
\Ncm:=N\left(L(\Icmi);~\nHcm(\thcm)\right)=n\cdot\ind(\Icmi).
\]
Apply Lemma~\ref{lem: Floquet-Bloch matrices}~(\ref{enu: lem-Floquet-Bloch matrices - 2})
to $[\co,m,n]\in\Co$ and the index relation (\ref{eq: Cor-unifyrecform - index relation for I_=00007Bc.m.n=00007D^=00007Bi=00007D})
of Lemma~\ref{lem: index identities for spectral bands} to conclude
\begin{align*}
\Ncmn:=N\left(L(\Icmni);~\Hcmn(\thcmn)\right) & =\ind(\Icmni)=\ind(\Ic)+n\cdot\ind(\Icmi).
\end{align*}
Since $I^{i}_{[\co,m]}(V_{1})\strict I_{\co}(V_{1})$, we infer $L(\Ic(V_{1}))<L(\Icmi(V_{1}))$
and $\sigma\left(H_{\co,V_{1}}(\thc)\right)\cap\Icmi(V_{1})=\emptyset$.
Hence, $L(\Ic)\in\sigma\left(\Hc(\thc)\right)$ and Lemma~\ref{lem: Floquet-Bloch matrices}~(\ref{enu: lem-Floquet-Bloch matrices - 2})
applied to $\co\in\Co$ lead to
\begin{align*}
\Nc:=N\left(L(\Icmi);~\Hc(\thc)\right) & =N(L(\Ic);~\Hc(\thc))+1=\ind(\Ic)+1.
\end{align*}
Summing up, we obtained $\Nc+\Ncm>\Ncmn$. Moreover, $\sigma\left(\Hc(\thc)\right)\cap\Icmi=\emptyset$
and Lemma~\ref{lem: Floquet-Bloch matrices - n times fundamental domain}~(\ref{enu: lem-Floquet-Bloch matrices - n times fundamental domain - 5})
that $\lo=L(\Icmi)$ is a simple eigenvalue of $\nHcm(\thcm)\oplus\Hc(\thc)$.
Using admissibility, Lemma~\ref{lem: Basics for eigenavlue inequalities}~(\ref{enu: lem-Basics for eigenavlue inequalities - 1})
yields the required inequality $\lo=L(\Icmi)>L(\Icmni)=\mo$.

(b) We aim to apply Lemma~\ref{lem: Basics for eigenavlue inequalities}
for $\lo=R(\Icmi)$, $\mo=R(\Icmnii)$. Let $\thc,\thcm,\thcmn\in\{0,\pi\}$
be such that
\[
R(\Ic)\in\sigma\left(\Hc(\thc)\right),~~R(\Icmi)\in\text{\ensuremath{\sigma\left(\nHcm(\thcm)\right)}}~~\textrm{and}~~R(\Icmnii)\in\ensuremath{\sigma\left(\Hcmn(\thcmn)\right)}.
\]
Then these spectral edges, respectively $\thc,\thcm,\thcmn$, are
admissible by inserting the index relation (\ref{eq: Cor-unifyrecform - index relation for I_=00007Bc.m.n=00007D^=00007Bi+1=00007D})
of Lemma~\ref{lem: index identities for spectral bands} into the
characterization of admissibility from Lemma~\ref{lem: Admissibility criterion}.
Furthermore, Lemma~\ref{lem: Floquet-Bloch matrices - n times fundamental domain}~(\ref{enu: lem-Floquet-Bloch matrices - n times fundamental domain - 4})
applied to $[\co,m]\in\Co$ implies 
\[
\Ncm:=N\left(R(\Icmi);~\nHcm(\thcm)\right)=n\cdot\left(\ind(\Icmi)+1\right)-1.
\]
Apply Lemma~\ref{lem: Floquet-Bloch matrices}~(\ref{enu: lem-Floquet-Bloch matrices - 3})
to $[\co,m,n]\in\Co$ and the index relation (\ref{eq: Cor-unifyrecform - index relation for I_=00007Bc.m.n=00007D^=00007Bi+1=00007D})
of Lemma~\ref{lem: index identities for spectral bands} to conclude
\begin{align*}
\Ncmn:=N\left(R(\Icmnii);~\Hcmn(\thcmn)\right) & =\ind(\Icmnii)\\
 & =n\cdot\left(\ind(\Icmi)+1\right)+\ind(\Ic).
\end{align*}
Since $I^{i}_{[\co,m]}(V_{1})\strict I_{\co}(V_{1})$, we infer $R(\Icmi(V_{1}))<R(\Ic(V_{1}))$
and $\sigma\left(H_{\co,V_{1}}(\thc)\right)\cap\Icmi(V_{1})=\emptyset$.
Hence, Lemma~\ref{lem: Floquet-Bloch matrices}~(c) applied to $\co\in\Co$
leads to
\begin{align*}
\Nc:=N\left(R(\Icmi);~\Hc(\thc)\right) & =N(R(\Ic);~\Hc(\thc))=\ind(\Ic).
\end{align*}
Summing up, we obtained $\Nc+\Ncm<\Ncmn$. Moreover, $\sigma\left(\Hc(\thc)\right)\cap\Icmi=\emptyset$
and Lemma~\ref{lem: Floquet-Bloch matrices - n times fundamental domain}~(\ref{enu: lem-Floquet-Bloch matrices - n times fundamental domain - 5})
that $\lo=R(\Icmi)$ is a simple eigenvalue of $\nHcm(\thcm)\oplus\Hc(\thc)$.
Using admissibility, Lemma~\ref{lem: Basics for eigenavlue inequalities}~(\ref{enu: lem-Basics for eigenavlue inequalities - 1})
yields the required inequality $\lo=R(\Icmi)<R(\Icmnii)=\mo$.

(c) We aim to apply Lemma~\ref{lem: Basics for eigenavlue inequalities}
for $\lo=R(\Icmi)$, $\mo=R(\Icmni)$. Let $\thc,\thcm,\thcmn\in\{0,\pi\}$
be such that
\[
L(\Ic)\in\sigma\left(\Hc(\thc)\right),~~R(\Icmi)\in\text{\ensuremath{\sigma\left(\nHcm(\thcm)\right)}}~~\textrm{and}~~R(\Icmni)\in\ensuremath{\sigma\left(\Hcmn(\thcmn)\right)}.
\]
Lemma~\ref{lem: Floquet-Bloch matrices - n times fundamental domain}~(\ref{enu: lem-Floquet-Bloch matrices - n times fundamental domain - 4})
applied to $[\co,m]\in\Co$ implies 
\[
\Ncm:=N\left(R(\Icmi);~\nHcm(\thcm)\right)=n\cdot\left(\ind(\Icmi)+1\right)-1.
\]
Apply Lemma~\ref{lem: Floquet-Bloch matrices}~(\ref{enu: lem-Floquet-Bloch matrices - 3})
to $[\co,m,n]\in\Co$ and the index relation (\ref{eq: Cor-unifyrecform - index relation for I_=00007Bc.m.n=00007D^=00007Bi=00007D})
of Lemma~\ref{lem: index identities for spectral bands} to conclude
\begin{align*}
\Ncmn:=N\left(R(\Icmni);~\Hcmn(\thcmn)\right) & =\ind(\Icmni)\\
 & =n\cdot\ind(I^{i}_{[c,m]})+\ind(I_{c}).
\end{align*}
Since $I^{i}_{[\co,m]}(V_{1})\strict I_{\co}(V_{1})$, we infer $L(\Ic(V_{1}))<R(\Icmi(V_{1}))<R(\Ic(V_{1}))$
and $\sigma\left(H_{\co,V_{1}}(\thc)\right)\cap\Icmi(V_{1})=\emptyset$.
Hence, $L(\Ic)\in\sigma\left(\Hc(\thc)\right)$ and Lemma~\ref{lem: Floquet-Bloch matrices}~(\ref{enu: lem-Floquet-Bloch matrices - 2})
applied to $\co\in\Co$ lead to 
\[
\Nc:=N\left(R(\Icmi);~\Hc(\thc)\right)=N\left(L(\Ic);~\Hc(\thc)\right)+1=\ind(\Ic)+1.
\]
Thus, $\Nc+\Ncm=\Ncmn+n>\Ncmn$ follows. Moreover, $\sigma\left(\Hc(\thc)\right)\cap\Icmi=\emptyset$
and Lemma~\ref{lem: Floquet-Bloch matrices - n times fundamental domain}~(\ref{enu: lem-Floquet-Bloch matrices - n times fundamental domain - 5})
that $\lo=R(\Icmi)$ is a simple eigenvalue of $\nHcm(\thcm)\oplus\Hc(\thc)$.
Observe that $\thc,\thcm,\thcmn$ are admissible, if and only if $n$
is odd by inserting the index relation (\ref{eq: Cor-unifyrecform - index relation for I_=00007Bc.m.n=00007D^=00007Bi=00007D})
of Lemma~\ref{lem: index identities for spectral bands} into the
characterization of admissibility from Lemma~\ref{lem: Admissibility criterion}.
Thus, if $n$ is odd, the previous considerations with Lemma~\ref{lem: Basics for eigenavlue inequalities}~(\ref{enu: lem-Basics for eigenavlue inequalities - 1})
yield $\lo=R(\Icmi(V))>R(\Icmni)=\mo$.

If $n$ is even, then $\thc,\thcm,\thcmn$ are not admissible. Moreover,
$\Nc+\Ncm-1>\Ncmn$ follows since $n\geq2$ if $n$ is even. Thus,
Lemma~\ref{lem: Basics for eigenavlue inequalities}~(\ref{enu: lem-Basics for eigenavlue inequalities - 2})
with $\sigma\left(\Hc(\thc)\right)\cap\Icmi=\emptyset$implies $\lo=R(\Icmi)>R(\Icmni)=\mo$.

(d) We aim to apply Lemma~\ref{lem: Basics for eigenavlue inequalities}
for $\lo=L(\Icmi)$, $\mo=L(\Icmnii)$. Let $\thc,\thcm,\thcmn\in\{0,\pi\}$
be such that
\[
R(\Ic)\in\sigma\left(\Hc(\thc)\right),\quad L(\Icmi)\in\text{\ensuremath{\sigma\left(\nHcm(\thcm)\right)}}\quad\textrm{and}\quad L(\Icmnii)\in\ensuremath{\sigma\left(\Hcmn(\thcmn)\right)}.
\]
Lemma~\ref{lem: Floquet-Bloch matrices - n times fundamental domain}~(\ref{enu: lem-Floquet-Bloch matrices - n times fundamental domain - 3})
applied to $[\co,m]\in\Co$ implies 
\[
\Ncm:=N\left(L(\Icmi);~\nHcm(\thcm)\right)=n\cdot\ind(\Icmi).
\]
Apply Lemma~\ref{lem: Floquet-Bloch matrices}~(\ref{enu: lem-Floquet-Bloch matrices - 2})
to $[\co,m,n]\in\Co$ and the index relation (\ref{eq: Cor-unifyrecform - index relation for I_=00007Bc.m.n=00007D^=00007Bi+1=00007D})
of Lemma~\ref{lem: index identities for spectral bands} to conclude
\begin{align*}
\Ncmn:=N\left(L(\Icmnii);~\Hcmn(\thcmn)\right) & =\ind(\Icmnii)=n\cdot\left(\ind(I^{i}_{[c,m]})+1\right)+\ind(I_{c}).
\end{align*}
Since $I^{i}_{[\co,m]}(V_{1})\strict I_{\co}(V_{1})$, we infer $L(\Ic(V_{1}))<L(\Icmi(V_{1}))<R(\Ic(V_{1}))$
and $\sigma\left(H_{\co,V_{1}}(\thc)\right)\cap\Icmi(V_{1})=\emptyset$.
Hence, Lemma~\ref{lem: Floquet-Bloch matrices}~(\ref{enu: lem-Floquet-Bloch matrices - 3})
applied to $\co\in\Co$ leads to 
\[
\Nc:=N\left(L(\Icmi);~\Hc(\thc)\right)=N(R(\Ic);~\Hc(\thc))=\ind(\Ic).
\]
Thus, $\Nc+\Ncm<\Ncmn$ follows. Moreover, $\sigma\left(\Hc(\thc)\right)\cap\Icmi=\emptyset$
and Lemma~\ref{lem: Floquet-Bloch matrices - n times fundamental domain}~(\ref{enu: lem-Floquet-Bloch matrices - n times fundamental domain - 5})
that $\lo=L(\Icmi)$ is a simple eigenvalue of $\nHcm(\thcm)\oplus\Hc(\thc)$.
Observe that $\thc,\thcm,\thcmn$ are admissible, if and only if $n\in\N$
is odd by inserting the index relation (\ref{eq: Cor-unifyrecform - index relation for I_=00007Bc.m.n=00007D^=00007Bi+1=00007D})
of Lemma~\ref{lem: index identities for spectral bands} into the
characterization of admissibility from Lemma~\ref{lem: Admissibility criterion}.
Thus, if $n\in\N$ is odd, the previous considerations with Lemma~\ref{lem: Basics for eigenavlue inequalities}~(\ref{enu: lem-Basics for eigenavlue inequalities - 2})
yield $\lo=L(\Icmi)<L(\Icmnii)=\mo$ as claimed.

If $n\in\N$ is even, then $\thc,\thcm,\thcmn$ are not admissible.
Thus, Lemma~\ref{lem: Basics for eigenavlue inequalities}~(\ref{enu: lem-Basics for eigenavlue inequalities - 2})
with $\sigma\left(\Hc(\thc)\right)\cap\Icmi=\emptyset$ and $\Nc+\Ncm<\Ncmn$
imply $\lo=L(\Icmi)<L(\Icmnii)=\mo$.
\end{proof}

The next lemma is tailored towards proving property \ref{enu: B2 property}.

\begin{lem}
\label{lem: Task (B2)} Let $V_{1}>0$, $m,n\in\N$ and $\co\in\Co$
be such that $\varphi(\co)\not\in\left\{ -1,\infty\right\} $ and
$[\co,m]\in\Co$. Consider a spectral band $I_{\co}$ in $\sigma_{\co}$
with the associated spectral bands $\{\Icmi\}^{M}_{i=1}$ and $\{\Icmnj\}^{M+1}_{j=1}$
introduced in Definition~\ref{def: Notion Icmi. Icmnj and M}. Let
$\Jcm$ and $\Kcm$ be the spectral bands associated with $\Ic$ as
defined in Definition~\ref{def: left and right proceeded band}.
If 
\[
\Icmno(V_{1})\strict\Ic(V_{1})\quad\textrm{and}\quad\IcmnN(V_{1})\strict\Ic(V_{1}),
\]
then 
\[
R(\Jcm(V_{1}))<R(\Icmno(V_{1}))\quad\textrm{and}\quad L(\IcmnN(V_{1}))<L(\Kcm(V_{1})).
\]
\end{lem}

\begin{rem*}
It might be that either $\Jcm$ or $\Kcm$ as defined in Definition
\ref{def: left and right proceeded band} do not exist. In such a
case, part of the statement is empty.

Combining Lemma~\ref{lem: Task (I)} and Lemma~\ref{lem: Task (B2)},
we get the following corollary which shows that properties \ref{enu: A1 property},
\ref{enu: B2 property} and \ref{enu: I property} hold under some
conditions.
\end{rem*}
\begin{cor}
\label{cor: Properties A1+B2+I}Let $V_{1}>0$, $m,n\in\N$ and $\co\in\Co$
be such that $\varphi(\co)\not\in\left\{ -1,\infty\right\} $ and
$[\co,m]\in\Co$. Consider a spectral band $I_{\co}$ in $\sigma_{\co}$
with associated spectral bands $\{\Icmi\}^{M}_{i=1}$ and $\{\Icmnj\}^{M+1}_{j=1}$
introduced in Definition~\ref{def: Notion Icmi. Icmnj and M}. If
\[
\Icm^{1}(V_{1}),~\Icm^{M}(V_{1})\strict I_{\co}(V_{1})\qquad\textrm{and}\qquad\Icmn^{1}(V_{1}),~\IcmnN(V_{1})\strict\Ic(V_{1}),
\]
then $\Ic(V_{1})$ satisfies the properties \ref{enu: A1 property},
\ref{enu: B2 property} and \ref{enu: I property}.
\end{cor}

\begin{proof}
[Proof of Corollary \ref{cor: Properties A1+B2+I}] First, we note
that the condition in the corollary is equivalent to $\Icmn^{j}(V_{1})\strict\Ic(V_{1})$
for all $1\leq j\leq M+1$ and $\Icm^{i}(V_{1})\strict I_{\co}(V_{1})$
for all $1\leq i\leq M$, since
\[
\Icmi\prec I^{i+1}_{[\co,m]}\qquad\textrm{and}\qquad\Icmnj\prec I^{j+1}_{[\co,m,n]}.
\]
The assumption that $I^{i}_{[\co,m]}(V_{1})\strict I_{\co}(V_{1})$
for all $1\leq i\leq M$ is exactly property \ref{enu: A1 property}
of $\Ic(V_{1})$. Moreover, this assumption allows us to apply Lemma~\ref{lem: Task (I)}
and obtain
\[
\Icmn^{i}(V_{1})\prec\Icm^{i}(V_{1})\prec\Icmn^{i+1}(V_{1})\qquad\textrm{for all }1\leq i\leq M.
\]
Thus, $\Ic(V_{1})$ satisfies property \ref{enu: I property}. Furthermore,
these relations imply that each of the bands $\left\{ \Icmn^{j}(V_{1})\right\} ^{M+1}_{j=1}$
is not contained in any of the bands $\left\{ \Icm^{i}(V_{1})\right\} ^{M}_{i=1}$,
which is useful towards proving property \ref{enu: B2 property}.
Recall (Definition~\ref{def: left and right proceeded band}) the
notation of the spectral bands $\Jcm$ and $\Kcm$ associated with
$\Ic$. In order to prove \ref{enu: B2 property}, it is enough to
prove that $\Icmn^{1}(V_{1})$ is not contained $\Jcm(V_{1})$ and
$\IcmnN(V_{1})$ is not contained in $\Kcm(V_{1})$. This follows
from Lemma~\ref{lem: Task (B2)}.
\end{proof}

\begin{proof}
[Proof of Lemma \ref{lem: Task (B2)}] (a) We prove that $R(\Jcm(V_{1}))<R(\Icmno(V_{1}))$.
First we note that this inequality immediately holds if $R(\Jcm(V_{1}))\leq L(\Ic(V_{1}))$,
because $\Icmno(V_{1})\strict\Ic(V_{1})$ by assumption. Therefore,
we assume from now on that $R(\Jcm(V_{1}))>L(\Ic(V_{1}))$. Although
the assumptions and the conclusions of the lemma depend on the fixed
$V_{1}>0$, we will abbreviate notation, for the sake of easier reading,
and omit the $V_{1}$ dependence in most parts of this proof. Let
$\thc,\thcm,\thcmn\in\{0,\pi\}$ be such that 
\[
R(\Ic)\in\sigma\left(\Hc(\thc)\right),\enspace R(\Jcm)\in\text{\ensuremath{\sigma\left(\nHcm(\thcm)\right)}}\enspace\textrm{and}\enspace R(\Icmno)\in\ensuremath{\sigma\left(\Hcmn(\thcmn)\right)}.
\]
Then these spectral edges, respectively $\thc,\thcm,\thcmn$, are
admissible by inserting the index relation (\ref{eq: Cor-unifyrecform - index relation for I_=00007Bc.m.n=00007D^1})
of Lemma~\ref{lem: index identities for spectral bands} into the
characterization of admissibility from Lemma~\ref{lem: Admissibility criterion}.
Furthermore, Lemma~\ref{lem: Floquet-Bloch matrices}~(\ref{enu: lem-Floquet-Bloch matrices - 3})
for $[\co,m,n]\in\Co$ and $\co\in\text{\ensuremath{\Co}}$, the index
relation (\ref{eq: Cor-unifyrecform - index relation for I_=00007Bc.m.n=00007D^1})
of Lemma~\ref{lem: index identities for spectral bands} and Lemma~\ref{lem: Floquet-Bloch matrices - n times fundamental domain}~(\ref{enu: lem-Floquet-Bloch matrices - n times fundamental domain - 4})
for the spectral band $\Jcm$ imply 
\begin{align}
N\left(R(\Icmno);~\Hcmn(\thcmn)\right) & =\ind(\Icmno)\label{eq: proof_Lem_task_B2}\\
 & =\ind(\Ic)+n\cdot(\ind(\Jcm)+1)\nonumber \\
 & =N\left(R(\Ic);~\Hc(\thc)\right)+N\left(R(\Jcm);~\nHcm(\thcm)\right)+1.\nonumber 
\end{align}
 In order to proceed, we first show that $R(\Jcm)<R(\Ic)$. Assume
by contradiction this is not the case, namely $R(\Jcm)\geq R(I_{c})$.
We aim to apply Lemma~\ref{lem: Basics for eigenavlue inequalities}
for $\lo=R(\Ic)$ and $\mo=R(\Icmno)$. With (\ref{eq: proof_Lem_task_B2})
and $R(\Jcm)\geq R(I_{c})$ at hand, we conclude
\begin{align*}
N\left(R(\Icmno);~\Hcmn(\thcmn)\right) & =N\left(R(\Ic);~\Hc(\thc)\right)+N\left(R(\Jcm);~\nHcm(\thcm)\right)+1\\
 & \geq N\left(R(\Ic);~\Hc(\thc)\right)+N\left(R(\Ic);~\nHcm(\thcm)\right)+1.
\end{align*}
Using the notation from Lemma~\ref{lem: Basics for eigenavlue inequalities},
the latter reads $\Ncmn\geq\Nc+\Ncm+1$. Thus, Lemma~\ref{lem: Basics for eigenavlue inequalities}
yields $\mo\geq\lo$ as $\thc,\thcm,\thcmn$ are admissible. On the
other hand, $\Icmno\strict\Ic$ implies $R(\Icmno)=\mo<\lo=R(\Ic)$,
a contradiction. Hence, $R(\Jcm)<R(\Ic)$ follows as claimed.

With this at hand, we continue applying once again Lemma~\ref{lem: Basics for eigenavlue inequalities},
but this time for $\lo=R(\Jcm)$ and $\mo=R(\Icmn^{1})$. Using (\ref{eq: proof_Lem_task_B2})
and $R(\Jcm)<R(\Ic)$, we infer
\begin{align*}
N\left(R(\Icmno);~\!\Hcmn(\thcmn)\right) & =N\left(R(\Ic);~\!\Hc(\thc)\right)+N\left(R(\Jcm);~\!\nHcm(\thcm)\right)+1\\
 & \geq N\left(R(\Jcm);~\!\Hc(\thc)\right)+N\left(R(\Jcm);~\!\nHcm(\thcm)\right)+1.
\end{align*}
Using the notation from Lemma~\ref{lem: Basics for eigenavlue inequalities},
the latter reads $\Ncmn\geq\Nc+\Ncm+1$. Recall that we showed in
the beginning of the proof, $L(\Ic(V_{1}))<R(\Jcm(V_{1}))$. This,
together with $R(\Jcm(V_{1}))<R(\Ic(V_{1}))$, implies that $R(\Jcm(V_{1}))$
is not an eigenvalue of $H_{\co,V_{1}}(\thc)$. Hence, $\lo=R(\Jcm(V_{1}))$
is a simple eigenvalue of $H_{\co,V_{1}}(\thc)\oplus H^{\times n}_{[\co,m],V_{1}}(\thcm)$
using Lemma~\ref{lem: Floquet-Bloch matrices - n times fundamental domain}~(\ref{enu: lem-Floquet-Bloch matrices - n times fundamental domain - 5}).
Thus, Lemma~\ref{lem: Basics for eigenavlue inequalities}~(\ref{enu: lem-Basics for eigenavlue inequalities - 1})
applied with $\Ncmn\geq\Nc+\Ncm+1$ yields that $\mo>\lo$, i.e.,
$R(\Icmno)>R(\Jcm)$, as required.

(b) We prove that $L(\IcmnN(V_{1}))<L(\Kcm(V_{1}))$. First we note
that this inequality immediately holds if $R(\Ic(V_{1}))\leq L(\Kcm(V_{1}))$,
because $\IcmnN(V_{1})\strict\Ic(V_{1})$ by assumption. Therefore,
we assume from now on that $L(\Kcm(V_{1}))<R(\Ic(V_{1}))$. In order
to simplify the notation, we will omit the dependence on $V_{1}$
in the following unless we want to emphasize its dependence. Let $\thc,\thcm,\thcmn\in\{0,\pi\}$
be such that 
\[
L(\Ic)\in\sigma\left(\Hc(\thc)\right),~~L(\Kcm)\in\text{\ensuremath{\sigma\left(\nHcm(\thcm)\right)}}~~\textrm{and}~~L(\IcmnN)\in\ensuremath{\sigma\left(\Hcmn(\thcmn)\right)}.
\]
Then these spectral edges, respectively $\thc,\thcm,\thcmn$, are
admissible by inserting the index relation (\ref{eq: Cor-unifyrecform - index relation for I_=00007Bc.m.n=00007D^M+1})
of Lemma~\ref{lem: index identities for spectral bands} into the
characterization of admissibility from Lemma~\ref{lem: Admissibility criterion}.
By Definition~\ref{def: left and right proceeded band}, we have
$\ind(\Kcm)=\ind(\Icm^{M})+1$. With this at hand, Lemma~\ref{lem: Floquet-Bloch matrices}~(\ref{enu: lem-Floquet-Bloch matrices - 2})
for $[\co,m,n]\in\Co$ and $\co\in\text{\ensuremath{\Co}}$, the index
relation (\ref{eq: Cor-unifyrecform - index relation for I_=00007Bc.m.n=00007D^M+1})
of Lemma~\ref{lem: index identities for spectral bands} and Lemma~\ref{lem: Floquet-Bloch matrices - n times fundamental domain}
(\ref{enu: lem-Floquet-Bloch matrices - n times fundamental domain - 3})
for the spectral band $\Kcm$ imply 
\begin{align}
N\left(L(\IcmnN);~\Hcmn(\thcmn)\right) & =\ind(\IcmnN)\label{eq: proof_Lem_task_B2_(b)}\\
 & =\ind(\Ic)+n\cdot\ind(\Kcm)\nonumber \\
 & =N\left(L(\Ic);~\Hc(\thc)\right)+N\left(L(\Kcm);~\nHcm(\thcm)\right).\nonumber 
\end{align}
 In order to proceed, we first show that $L(\Kcm)>L(\Ic)$. Assume
by contradiction this is not the case, namely $L(\Kcm)\leq L(I_{c})$.
We aim to apply Lemma~\ref{lem: Basics for eigenavlue inequalities}
for $\lo=L(\Ic)$ and $\mo=L(\IcmnN)$. If $L(\Kcm)=L(I_{c})$, then
$\lo$ has multiplicity $\Mult_{\lo}=2$. Thus, the previous identity
(\ref{eq: proof_Lem_task_B2_(b)}) leads to 
\[
N\left(L(\IcmnN);~\Hcmn(\thcmn)\right)<N\left(L(\Ic);~\Hc(\thc)\right)+N\left(L(\Ic);~\nHcm(\thcm)\right)+\Mult_{\lo}-1.
\]
If $L(\Kcm)<L(I_{c})$, then $N\left(L(\Kcm);~\nHcm(\thcm)\right)\leq N\left(L(\Ic);~\nHcm(\thcm)\right)-1$
follows and the multiplicity of $\lo$ satisfies $\Mult_{\lo}\geq1$.
Combing these with the Equation (\ref{eq: proof_Lem_task_B2_(b)})
and $L(\Kcm)\in\sigma\left(\nHcm(\thcm)\right)$, we conclude 
\begin{align*}
N\left(L(\IcmnN);~\!\Hcmn(\thcmn)\right) & =N\left(L(\Ic);~\!\Hc(\thc)\right)+N\left(L(\Kcm);~\!\nHcm(\thcm)\right)\\
 & <N\left(L(\Ic);~\!\Hc(\thc)\right)+N\left(L(\Ic);~\!\nHcm(\thcm)\right)+\mathbf{\Mult_{\lo}}-1.
\end{align*}
Using the notation from Lemma~\ref{lem: Basics for eigenavlue inequalities},
the latter reads $\Ncmn<\Nc+\Ncm+\Mult_{\lo}-1$ whenever $L(\Kcm)\leq L(I_{c})$.
Thus, Lemma~\ref{lem: Basics for eigenavlue inequalities}~(\ref{enu: lem-Basics for eigenavlue inequalities - 1})
yields $\mo\leq\lo$ as $\thc,\thcm,\thcmn$ are admissible. On the
other hand, $\IcmnN\strict\Ic$ implies $L(\IcmnN)=\mo>\lo=L(\Ic)$,
a contradiction. Hence, $L(\Kcm)>L(\Ic)$ follows as claimed.

With this at hand, we continue applying once again Lemma~\ref{lem: Basics for eigenavlue inequalities},
but this time for $\lo=L(\Kcm)$ and $\mo=L(\IcmnN)$. Using (\ref{eq: proof_Lem_task_B2_(b)}),
the inequality $L(\Kcm)>L(\Ic)$ and $L(I_{\co})\in\sigma(\Hc(\thc)$,
we infer
\begin{align*}
N\left(L(\IcmnN);~\Hcmn(\thcmn)\right) & =N\left(L(\Ic);~\Hc(\thc)\right)+N\left(L(\Kcm);~\nHcm(\thcm)\right)\\
 & <N\left(L(\Kcm);~\Hc(\thc)\right)+N\left(L(\Kcm);~\nHcm(\thcm)\right).
\end{align*}
Using the notation from Lemma~\ref{lem: Basics for eigenavlue inequalities},
the latter reads $\Ncmn<\Nc+\Ncm$. Recall that we showed in the beginning
of the proof, $L(\Kcm(V_{1}))<R(\Ic(V_{1}))$. This, together with
$L(\Ic(V_{1}))<L(\Kcm(V_{1}))$, implies that $L(\Kcm(V_{1}))$ is
not an eigenvalue of $H_{\co,V_{1}}(\thc)$. Hence, $\lo=L(\Kcm(V_{1}))$
is a simple eigenvalue of $H_{\co,V_{1}}(\thc)\oplus H^{\times n}_{[\co,m],V_{1}}(\thcm)$
using Lemma~\ref{lem: Floquet-Bloch matrices - n times fundamental domain}~(\ref{enu: lem-Floquet-Bloch matrices - n times fundamental domain - 5}).
Thus, Lemma~\ref{lem: Basics for eigenavlue inequalities} applied
with $\Ncmn<\Nc+\Ncm$ leads to $\mo<\lo$, i.e., $L(\IcmnN)<L(\Kcm)$.
\end{proof}

Next we show that the assumptions in the previous Corollary~\ref{cor: Properties A1+B2+I}
are satisfied whenever the spectral band is of backward type $A$
or $B$. The proofs of the previous lemmata in this subsection were
mainly based on applications of Lemma~\ref{lem: Basics for eigenavlue inequalities}.
This lemma will keep being applied in the next proofs, but we will
also need to make use of some trace identities, as appear in Lemma~\ref{lem: Traces and spectral edges}
and Lemma~\ref{lem: Trace estimates =00005Bc.m.n=00005D}.
\begin{lem}
\label{lem: Backward Implies Bands strictly included} Let $m\in\N$
and $\co\in\Co$ be such that $\varphi(\co)\in(0,1)$. Let $V_{1}>0$
and $I_{\co}$ be a spectral band in $\sigc$ with associated spectral
bands $\{\Icmi\}^{M}_{i=1}$ and $\{\Icmo^{j}\}^{M+1}_{j=1}$ introduced
in Definition~\ref{def: Notion Icmi. Icmnj and M}. If either
\end{lem}

\begin{itemize}
\item $I_{\co}(V)$ is of backward type A for all $V\geq V_{1}$ and $M:=m-1$,
or
\item $I_{\co}(V)$ is of backward type B for all $V\geq V_{1}$ and $M:=m$,
\end{itemize}
then for all $V\geq V_{1}$,
\[
I^{1}_{[\co,m,1]}(V),~I^{M+1}_{[\co,m,1]}(V)\strict I_{\co}(V)\qquad\text{and}\qquad\Icm^{1}(V),~\Icm^{M}(V)\strict I_{\co}(V).
\]

\begin{rem*}
We have to exclude the cases $\varphi(\co)\in\{0,\pm1,\infty\}$ so
that we can apply Lemma~\ref{lem: Trace estimates =00005Bc.m.n=00005D}~(\ref{enu: prop-Trace estimates =00005Bc.m.n=00005D - 3}).
\end{rem*}
\begin{proof}
The claim follows once we show that for all $V\geq V_{1}$,
\begin{align}
L(\Ic(V)) & <\min\left\{ L(\Icmo^{1}(V)),~L(\Icm^{1}(V))\right\} ,\nonumber \\
\max\left\{ R(\Icm^{M}(V)),~R(\Icmo^{M+1}(V))\right\}  & <R(\Ic(V)).\label{eq: proof Backward Implies Bands strictly included}
\end{align}
Assume by contradiction that (\ref{eq: proof Backward Implies Bands strictly included})
does not hold for some $V\geq V_{1}$. Due to Theorem~\ref{thm: V>4 Type A =000026 B},
these strict inequalities in (\ref{eq: proof Backward Implies Bands strictly included})
hold for $V>4$. Thus, the continuity of the spectral band edges in
$V>0$ (Corollary~\ref{prop: Lipschitz spectral edges}) implies
that the maximum
\[
V_{2}:=\max\set{V\geq V_{1}}{(\text{\ref{eq: proof Backward Implies Bands strictly included}})\textrm{ does not hold}}
\]
exists and $V_{2}\in\left[V_{1},4\right]$. Due to Lemma \ref{lem: Task (I)},
the strict inclusions $\Icm^{1}(V)\strict I_{\co}(V)$ and $\Icm^{M}(V)\strict I_{\co}(V)$
for $V>V_{2}$ yield
\[
L(\Icmo^{1}(V))<L(\Icm^{1}(V))\qquad\textrm{and}\qquad R(\Icm^{M}(V))<R(\Icmo^{M+1}(V))\qquad\textrm{for}\;V>V_{2}.
\]
Let $\Jcm$ and $\Kcm$ be the spectral bands associated with $\Ic$
(Definition~\ref{def: left and right proceeded band}). Since the
strict inclusions $\Icmo^{1}(V)\strict I_{\co}(V)$ and $\Icmo^{M+1}(V)\strict I_{\co}(V)$
hold for $V>V_{2}$, Lemma~\ref{lem: Task (B2)} asserts for $V>V_{2}$,
\[
R(\Jcm(V))<R(\Icmo^{1}(V))\quad\textrm{and}\quad L(\Icmo^{M+1}(V))<L(\Kcm(V)).
\]
Note that we have $R(\Icmo^{1}(V))\leq R(\Icmo^{M+1}(V))$ and $L(\Icmo^{1}(V))\leq L(\Icmo^{M+1}(V))$
for $V>0$. Hence, the continuity of the spectral band edges in $V>0$
(Corollary~\ref{prop: Lipschitz spectral edges}) leads to
\begin{align}
L(\Icmo^{1}(V_{2})) & \leq\min\left\{ L(\Icm^{1}(V_{2})),L(\Kcm(V_{2}))\right\} ,\nonumber \\
\max\left\{ R(\Icm^{M}(V_{2})),R(\Jcm(V_{2}))\right\}  & \leq R(\Icmo^{M+1}(V_{2})).\label{eq: proof Backward Implies Bands strictly included -V2}
\end{align}
Note that the spectral bands $\Jcm$ and $\Kcm$ may not exist simplifying
our considerations below. This in particular implies
\[
V_{2}=\max\set{V\geq V_{1}}{L(\Ic(V))=L(\Icmo^{1}(V))~\textrm{or}~R(\Ic(V))=R(\Icmo^{M+1}(V))}.
\]
We continue proving that this leads to a contradiction.

\underline{Case 1:} We show that $L(\Ic(V_{2}))=L(\Icmo^{1}(V_{2}))$
yields a contradiction. Set $E:=L(\Ic(V_{2}))$. Thus, $\left|t_{[\co,m,1]}(E;~V_{2})\right|=2$
follows from Lemma~\ref{lem: Traces and spectral edges} (\ref{enu: prop-Traces and spectral edges - 1}).
Since $\Ic(V_{2})$ is of backward type $A$ or $B$ (using $V_{2}\geq V_{1}$)
and $\varphi(\co)\in(0,1)$, Lemma~\ref{lem: Trace estimates =00005Bc.m.n=00005D}~(\ref{enu: prop-Trace estimates =00005Bc.m.n=00005D - 3})
yields $\left|\tcm(E;~V_{2})\right|<2$. Hence, $E$ must lie in the
interior of a spectral band in $\sigcm(V_{2})$ by Lemma~\ref{lem: Traces and spectral edges}
(\ref{enu: prop-Traces and spectral edges - 1}). Thus, Equation~(\ref{eq: proof Backward Implies Bands strictly included -V2})
and $L(\Ic(V_{2}))=L(\Icmo^{1}(V_{2}))$ lead to
\[
E=L(\Ic(V_{2}))<R(\Jcm(V_{2})).
\]
Note that if $\Jcm$ does not exist, then there is no spectra to the
left of $L(\Icm^{1}(V_{2}))$ contradicting $\left|\tcm(E;~V_{2})\right|<2$
and (\ref{eq: proof Backward Implies Bands strictly included -V2}).
Hence, we may continue assuming that $\Jcm$ exists. Next we aim to
apply Lemma~\ref{lem: Basics for eigenavlue inequalities} for $\lo=L(\Ic(V_{2}))$
and $\mo=L(\Icmo^{1}(V_{2}))$. For the sake of simplification, we
drop the $V_{2}$ notation in the following. Let $\thc,\thcm,\theta_{[\co,m,1]}\in\{0,\pi\}$
be such that
\[
L(\Ic)\in\sigma\left(\Hc(\thc)\right),~R(\Jcm)\in\sigma\left(H^{\times1}_{[\co,m]}(\thcm)\right)~\textrm{and}~L(\Icmo^{1})\in\sigma\left(\Hcmo(\theta_{[\co,m,1]})\right).
\]
Then these spectral edges, respectively $\thc,\thcm,\theta_{[\co,m,1]}$,
are admissible by inserting the index relation (\ref{eq: Cor-unifyrecform - index relation for I_=00007Bc.m.n=00007D^1})
into the characterization of admissibility from Lemma~\ref{lem: Admissibility criterion}
for $n=1$. With this at hand, and Lemma~\ref{lem: Floquet-Bloch matrices}~(\ref{enu: lem-Floquet-Bloch matrices - 2})
applied to $\co,[\co,m,1]\in\Co$ leads to
\begin{align*}
\Nc & :=N\left(L(\Ic);~\Hc(\thc)\right)=\ind(\Ic)
\end{align*}
and using (\ref{eq: Cor-unifyrecform - index relation for I_=00007Bc.m.n=00007D^1})
\[
N_{[\co,m,1]}:=N\left(L(\Icmo^{1});~\Hcmo(\theta_{[\co,m,1]})\right)=\ind(\Icmo^{1})=\ind(\Jcm)+1+\ind(\Ic).
\]
Furthermore, $L(\Ic)<R(\Jcm)$ and Lemma \ref{lem: Floquet-Bloch matrices - n times fundamental domain}
(\ref{enu: lem-Floquet-Bloch matrices - n times fundamental domain - 4})
applied to $[\co,m]\in\Co$ for $n=1$ lead to
\begin{align*}
\Ncm & :=N\left(L(\Ic);~\Hcm^{\times1}(\thcm)\right)\leq N\left(R(\Jcm);~\Hcm^{\times1}(\thcm)\right)=\ind(\Jcm).
\end{align*}
Thus, $\Nc+\Ncm<N_{[\co,m,1]}$ follows. Since $\lo=E=L(\Ic(V_{2}))$
lies in the interior of a spectral band in $\sigcm(V_{2})$ and the
eigenvalues of $H^{\times1}_{[\co,m],V_{2}}(\thcm)$ are contained
in the spectral band edges of $\sigcm(V_{2})$ (by Lemma~\ref{lem: Floquet-Bloch matrices - n times fundamental domain}),
we conclude that $\lo$ is not an eigenvalue of $H^{\times1}_{[\co,m],V_{2}}(\thcm)$.
Thus, $\lo$ is a simple eigenvalue of $H^{\times1}_{[\co,m],V_{2}}(\thcm)\oplus H_{\co,V_{2}}(\thc)$
using Lemma~\ref{lem: Floquet-Bloch matrices - n times fundamental domain}~(\ref{enu: lem-Floquet-Bloch matrices - n times fundamental domain - 5}).
Hence, Lemma~\ref{lem: Basics for eigenavlue inequalities}~(\ref{enu: lem-Basics for eigenavlue inequalities - 1})
yields 
\[
\lo=L(\Ic(V_{2}))<L(\Icmo^{1}(V_{2}))=\mo,
\]
 contradicting that these two values are equal by the initial assumption
of the considered case.

\underline{Case 2:} Similarly as in Case~1, $R(\Icmo^{M+1}(V_{2}))=R(\Ic(V_{2}))$
yields a contradiction.
\end{proof}

We have seen that Corollary~\ref{cor: Properties A1+B2+I} is set
towards proving the forward properties \ref{enu: A1 property}, \ref{enu: B2 property},
\ref{enu: I property}. Next, we aim to prove the forward property
\ref{enu: B1 property}, (also called the \emph{tower} property)\emph{.}
\begin{cor}
\label{cor: Tower property} Let $V_{1}>0$, $m\in\N$ and $\co\in\Co$
be such that $\varphi(\co)\not\in\left\{ -1,\infty\right\} $ and
$[\co,m]\in\Co$. Consider a spectral band $I_{\co}$ in $\sigma_{\co}$
with associated spectral bands $\{\Icmi\}^{M}_{i=1}$ and $\{\Icmnj\}^{M+1}_{j=1}$
introduced in Definition~\ref{def: Notion Icmi. Icmnj and M}. If
$1\leq j\leq M+1$ and $\Icmn^{j}(V)\strict\Ic(V)$ holds for all
$V\geq V_{1}$ and all $n\in\N$, then
\[
\Icmn^{j}(V)\strict I^{j}_{[\co,m,n-1]}(V)
\]
holds for all $n\in\N$ and $V\geq V_{1}$ where $I^{j}_{[\co,m,0]}(V)=\Ic(V)$.
\end{cor}

\begin{proof}
The proof is by induction over $n\in\N$. The induction base ($n=1$)
holds trivially since $I^{j}_{[\co,m,1]}(V)\strict I_{\co}(V)$ for
all \textbf{$V\geq V_{1}$} and $\sigma_{[\co,m,n-1]}(V)=\sigma_{\co}(V)$
if $n=1$ by Proposition~\ref{Prop-traceMaps}~(\ref{enu:.Prop-traceMaps - 1}).

For the induction step, suppose $I^{j}_{[\co,m,n]}(V)\strict I^{j}_{[\co,m,n-1]}(V)$
holds for all $V\geq V_{1}$. We show that $I^{j}_{[\co,m,n+1]}(V)\strict I^{j}_{[\co,m,n]}(V)$
holds for all $V\geq V_{1}$. Due to Proposition~\ref{Prop-traceMaps}~(\ref{enu:.Prop-traceMaps - 1}),
we have $\sigma_{[\co,m,n+1]}(V)=\sigma_{[\co,m,n,1]}(V)$. Furthermore,
$I^{j}_{[\co,m,n+1]}(V)\strict I^{j}_{[\co,m,n]}(V)$ holds for $V>4$
since $\Ic(V)$ is either of type $A$ or $B$ for $V>4$ by Theorem~\ref{thm: V>4 Type A =000026 B}.
Thus, $I^{j}_{[\co,m,n+1]}(V)$ equals to the unique spectral band
$I^{1}_{[\co,m,n,1]}(V)$ of type $A$ that is strictly contained
in $I^{j}_{[\co,m,n]}(V)$ for $V>4$. Hence, it suffices to prove
$I^{1}_{[\co,m,n,1]}(V)\strict I^{j}_{[\co,m,n]}(V)$ for all $V\geq V_{1}$.

Let $V\geq V_{1}$. By induction hypothesis, we have $I^{j}_{[\co,m,n]}(V)\strict I^{j}_{[\co,m,n-1]}(V)$
for all $V\geq V_{1}$, namely $I^{j}_{[\co,m,n]}(V)$ is of backward
type $B$ for all $V\geq V_{1}$. Furthermore, $\varphi([\co,m,n])\in(0,1)$
holds as $m,n\in\N$. Thus, Lemma~\ref{lem: Backward Implies Bands strictly included}
applied to $[\co,m,n]$ implies $I^{1}_{[\co,m,n,1]}(V)\strict I^{j}_{[\co,m,n]}(V)$
for all $V\geq V_{1}$.
\end{proof}

The next lemma is the crucial ingredient to prove the forward property
\ref{enu: A2 property}.
\begin{lem}
\label{lem: task (B1)} Let $V_{1}>0$ and $\co\in\Co$ be such that
$\varphi(\co)\in(0,1)$. Consider a spectral band $V\mapsto I_{\co}(V)$
in $\sigma_{\co}(V)$ which is of backward type $B$ for all $V\geq V_{1}$
and $\Ico$ is the associated spectral band introduced in Definition~\ref{def: Notion Icmi. Icmnj and M}.
Then for all $V\geq V_{1}$, $I^{1}_{[\co,1]}(V)\strict I_{\co}(V)$
(namely $\Ico(V)$ is of backward type $A$) and $I^{1}_{[\co,1]}(V)$
is not of weak backward type $B$.
\end{lem}

\begin{proof}
Since $\Ic(V)$ is of backward type $B$ for all $V\geq V_{1}$, it
follows that $\Ic(V)$ is of type $B$ for all $V>4$ by Theorem~\ref{thm: V>4 Type A =000026 B}.
Thus, there is a unique spectral band $I^{1}_{[\co,1]}$ in $\sigma_{[\co,1]}$
such that $I^{1}_{[\co,1]}(V)\strict I_{\co}(V)$ for all $V>4$.
By Theorem~\ref{thm: V>4 Type A =000026 B}, the lemma holds for
all $V_{1}>4$, and so we can assume in the proof that $V_{1}\leq4$.

Consider the spectral bands $\Jcz$ and $\Kcz$ associated with $\Ic$
(see Definition~\ref{def: left and right proceeded band}). Since
$\Ic(V)$is of backward type $B$ for $V>4$, we have $\ind(\Kcz)=\ind(\Jcz)+1$
(i.e., there is no other spectral band between those two) and
\[
\forall V>4,\quad\quad\Jcz(V)\prec\Ic(V)\prec\Kcz(V).
\]
Lemma~\ref{lem: Backward Implies Bands strictly included} implies
$\Ico(V)\strict\Ic(V)$ for all $V\geq V_{1}$. It is left to prove
that for $V\geq V_{1}$, $I^{1}_{[\co,1]}(V)$ is not contained in
any spectral band of $\sigma_{[\co,1,-1]}(V)=\sigma_{[\co,0]}(V)$
(where the last equality follows from Proposition~\ref{Prop-traceMaps}).
Actually, it suffices to prove that for all $V\geq V_{1}$,
\begin{equation}
R(\Jcz(V))<R(\Ico(V))\quad\textrm{and\ensuremath{\quad}}L(\Ico(V))<L(\Kcz(V)).\label{eq: equivalent to Lem-Task_(B1)}
\end{equation}

Assume by contradiction that (\ref{eq: equivalent to Lem-Task_(B1)})
does not hold for some $V\geq V_{1}$. By Theorem~\ref{thm: V>4 Type A =000026 B},
(\ref{eq: equivalent to Lem-Task_(B1)}) holds for $V>4$. Thus, the
continuity of the spectral band edges in $V>0$ (Corollary~\ref{prop: Lipschitz spectral edges})
implies that the maximum
\[
V_{2}:=\max\set{V\geq V_{1}}{R(\Jcz(V))=R(\Ico(V))\;\textrm{or}\;L(\Ico(V))=L(\Kcz(V))}
\]
exists and $V_{2}\in\left[V_{1},4\right]$. We split into cases according
to the nature of failure of (\ref{eq: equivalent to Lem-Task_(B1)})
at $V=V_{2}$, and show a contradiction for each of these cases. First
note that Equation (\ref{eq: Cor-unifyrecform - index relation for I_=00007Bc.1=00007D^=00007B1=00007D})
of Lemma~\ref{lem: index identities for spectral bands} implies
\begin{equation}
\ind(I^{1}_{[\co,1]})=\ind(\Jcz)+1+\ind(\Ic)=\ind(\Kcz)+\ind(\Ic).\label{eq: Lem-Task_(B1) Index relation}
\end{equation}
Since $\varphi(\co)\in(0,1)$, there is a $k\in\N$ such that $\co=[0,c_{0},\ldots,c_{k}]$.
In the following we apply Lemma~\ref{lem: Basics for eigenavlue inequalities}
to $\cop,[\cop,m],[\cop,m,n]\in\Co$ where $\cop=[0,c_{0},\ldots,c_{k-1}]$,
$m=c_{k}$ and $n=1$. Note that $\varphi(\cop)=\varphi([\co,0])$,
$\varphi([\cop,m])=\varphi(\co)$ and $\varphi([\cop,m,n])=\varphi([\co,1])$.
Thus, in effect it is as if we apply Lemma~\ref{lem: Basics for eigenavlue inequalities}
to $[\co,0],\co,[\co,1]\in\Co$ (rather than to $\cop,[\cop,m],[\cop,m,n]\in\Co$).
We use this convention until the end of the current proof.

\underline{Case 1:} We show that $R(\Jcz(V_{2}))=R(\Ico(V_{2}))$
yields a contradiction. We aim to apply Lemma~\ref{lem: Basics for eigenavlue inequalities}
for $\lo=R(\Jcz(V_{2}))$ and $\mo=R(\Ico(V_{2}))$. For the sake
of simplification, we drop the $V_{2}$ notation in the following
unless we want to emphasize its dependence. Let $\theta_{[\co,0]},\thc,\theta_{[\co,1]}\in\{0,\pi\}$
be such that
\[
R(\Jcz)\in\sigma\left(\Hcz(\theta_{[\co,0]})\right),~R(\Ic)\in\sigma\left(H^{\times1}_{\co}(\thc)\right)~\textrm{and}~R(\Ico)\in\sigma\left(\Hco(\theta_{[\co,1]})\right).
\]
Then these spectral edges, respectively $\theta_{[\co,0]},\thc,\theta_{[\co,1]}$,
are admissible by inserting the index relation (\ref{eq: Lem-Task_(B1) Index relation})
into the characterization of admissibility from Lemma \ref{lem: Admissibility criterion}
for $n=1$. With this at hand, Lemma~\ref{lem: Floquet-Bloch matrices}~(\ref{enu: lem-Floquet-Bloch matrices - 3})
applied to $[\co,0]\in\Co$ and $[\co,1]\in\Co$ leads to
\begin{align*}
N_{[\co,0]} & :=N\left(R(\Jcz);~\Hcz(\theta_{[\co,0]})\right)=\ind(\Jcz)
\end{align*}
and
\[
N_{[\co,1]}:=N\left(R(\Ico);~\Hco(\theta_{[\co,1]})\right)=\ind(\Ico).
\]
Furthermore, $\Ico(V_{2})\strict\Ic(V_{2})$ and the assumption $R(\Jcz(V_{2}))=R(\Ico(V_{2}))$
imply $R(\Jcz(V_{2}))<R(\Ic(V_{2}))\in\sigma\left(H^{\times1}_{\co,V_{2}}(\thc)\right)$.
Thus, Lemma~\ref{lem: Floquet-Bloch matrices - n times fundamental domain}~(\ref{enu: lem-Floquet-Bloch matrices - n times fundamental domain - 4})
applied to $n=1$ and $\co\in\Co$ imply
\begin{align*}
\Nc & :=N\left(R(\Jcz);~H^{\times1}_{\co}(\thc)\right)\leq N\left(R(\Ic);~H^{\times1}_{\co}(\thc)\right)=\ind(\Ic).
\end{align*}
Thus, (\ref{eq: Lem-Task_(B1) Index relation}) implies $N_{[\co,1]}>\Nc+N_{[\co,0]}$.
If we prove that $\lo=R(\Jcz(V_{2}))$ is a simple eigenvalue of $H^{\times1}_{\co,V_{2}}(\thc)\oplus H_{[\co,0],V_{2}}(\theta_{[\co,0]})$,
then Lemma~\ref{lem: Basics for eigenavlue inequalities} yields
$\lo=R(\Jcz(V_{2}))<R(\Ico(V_{2}))=\mo$, a contradiction.

By Lemma~\ref{lem: Floquet-Bloch matrices - n times fundamental domain}~(\ref{enu: lem-Floquet-Bloch matrices - n times fundamental domain - 5}),
simplicity of the eigenvalue $\lo$ holds if it is not an eigenvalue
of $H^{\times1}_{\co,V_{2}}(\thc)=H_{\co,V_{2}}(\thc)$. Using Lemma~\ref{lem: Floquet-Bloch matrices - n times fundamental domain}~(\ref{enu: lem-Floquet-Bloch matrices - n times fundamental domain - 5}),
$R(\Ic(V_{2}))$ is the only eigenvalue of $H_{\co,V_{2}}(\thc)$
in $\Ic(V_{2})$ . Thus, our working assumption, $R(\Jcz(V_{2}))=R(\Ico(V_{2}))<R(\Ic(V_{2}))$
implies that $R(\Jcz(V_{2}))$ is not an eigenvalue of $H_{\co,V_{2}}(\thc)$.

\underline{Case 2:} We show that $L(\Ico(V_{2}))=L(\Kcz(V_{2}))$
yields a contradiction. We aim to apply Lemma \ref{lem: Basics for eigenavlue inequalities}
for $\lo=L(\Kcz(V_{2}))$ and $\mo=L(\Ico(V_{2}))$. For the sake
of simplification, we drop the $V_{2}$ notation in the following
unless we want to emphasize its dependence. Let $\theta_{[\co,0]},\thc,\theta_{[\co,1]}\in\{0,\pi\}$
be such that
\[
L(\Kcz)\in\sigma\left(\Hcz(\theta_{[\co,0]})\right),~L(\Ic)\in\sigma\left(H^{\times1}_{\co}(\thc)\right)~\textrm{and}~L(\Ico)\in\sigma\left(\Hco(\theta_{[\co,1]})\right).
\]
Then these spectral edges, respectively $\theta_{[\co,0]},\thc,\theta_{[\co,1]}$,
are admissible by inserting the index relation (\ref{eq: Lem-Task_(B1) Index relation})
into the characterization of admissibility from Lemma~\ref{lem: Admissibility criterion}
for $n=1$. With this at hand, Lemma~\ref{lem: Floquet-Bloch matrices}~(\ref{enu: lem-Floquet-Bloch matrices - 2})
applied to $[\co,0]\in\Co$ and $[\co,1]\in\Co$ leads to
\begin{align*}
N_{[\co,0]} & :=N\left(L(\Kcz);~\Hcz(\theta_{[\co,0]})\right)=\ind(\Kcz)
\end{align*}
and
\[
N_{[\co,1]}:=N\left(L(\Ico);~\Hco(\theta_{[\co,1]})\right)=\ind(\Ico).
\]
Furthermore, $\Ico(V_{2})\strict\Ic(V_{2})$ and the assumption $L(\Ico(V_{2}))=L(\Kcz(V_{2}))$
imply $\sigma\left(H^{\times1}_{\co,V_{2}}(\thc)\right)\ni L(\Ic(V_{2}))<L(\Kcz(V_{2}))$.
Thus, Lemma~\ref{lem: Floquet-Bloch matrices - n times fundamental domain}~(\ref{enu: lem-Floquet-Bloch matrices - n times fundamental domain - 3})
applied to $n=1$ and $\co\in\Co$ leads to
\begin{align*}
\Nc & :=N\left(L(\Kcz);~H^{\times1}_{\co}(\thc)\right)\geq N\left(L(\Ic);~H^{\times1}_{\co}(\thc)\right)+1=\ind(\Ic)+1.
\end{align*}
Thus, (\ref{eq: Lem-Task_(B1) Index relation}) implies $N_{[\co,1]}<\Nc+N_{[\co,0]}$.
If we prove that $L(\Kcz(V_{2}))$ is a simple eigenvalue of $H^{\times1}_{\co,V_{2}}(\thc)\oplus H_{[\co,0],V_{2}}(\theta_{[\co,0]})$,
then Lemma \ref{lem: Basics for eigenavlue inequalities} yields $\lo=L(\Kcz(V_{2}))>L(\Ico(V_{2}))=\mo$,
a contradiction.

By Lemma~\ref{lem: Floquet-Bloch matrices - n times fundamental domain}~(\ref{enu: lem-Floquet-Bloch matrices - n times fundamental domain - 5}),
simplicity of the eigenvalue $\lo$ holds if it is not an eigenvalue
of $H^{\times1}_{\co,V_{2}}(\thc)=H_{\co,V_{2}}(\thc)$. Using Lemma~\ref{lem: Floquet-Bloch matrices - n times fundamental domain}~(\ref{enu: lem-Floquet-Bloch matrices - n times fundamental domain - 5}),
$L(\Ic(V_{2}))$ is the only eigenvalue of $H_{\co,V_{2}}(\thc)$
in $\Ic(V_{2})$. Thus, our working assumption, $L(\Ic(V_{2}))<L(\Ico(V_{2}))=L(\Kcz(V_{2}))$
implies that $L(\Kcz(V_{2}))$ is not an eigenvalue of $H_{\co,V_{2}}(\thc)$.
\end{proof}

\section{Proving that backward type implies forward type\label{subsec:Backward-implies-forward}}

In this section we prove Proposition~\ref{prop: Backward implies forward},
asserting that a fixed backward type implies a fixed forward type.

We recall (Definition~\ref{def: A B types}) that a spectral band
is of $m$-type $A$ (respectively $B$) if it is both of backward
type $A$ ($B$) and of $m$-forward type $A$ ($B$). We also recall
from Definition~\ref{def:(mV)-type property} the notion of the $(m,V)$-property
(or simply $\TmV(m,V)$) for spectral bands. For $\co\in\Co$ and
$m\in\N$ such that $\varphi(\co)\not\in\left\{ -1,\infty\right\} $
and $[\co,m]\in\Co$, the associated critical value is given by
\[
\crit([\co,m]):=\inf\set{V\geq0}{\begin{array}{c}
\text{each spectral band \ensuremath{\Ic} in \ensuremath{\sigma_{\co}}}\text{ satisfies \ensuremath{\TmV(m,V)}}\end{array}}
\]

We first state two lemmata and two corollaries leading to Proposition~\ref{prop: Backward implies forward}.
Before that, we slightly relax the notion of $\TmV(m,V)$ and the
definition of $\crit$.
\begin{defn}
\label{def:quasi (mV)-type property}Let $m\in\N$, $\co\in\Co$ be
such that $\varphi(\co)\not\in\left\{ -1,\infty\right\} $ and $[\co,m]\in\Co$.
For $V>0$ a spectral band $\Ic$ in $\sigc$ satisfies $\Tquas(m,V)$
(the \emph{quasi $(m,V)$-property}) if
\end{defn}

\begin{enumerate}
\item either
\begin{itemize}
\item for all $V'\geq V$, $\Ic(V')$ is of backward type $\tA$ with $M=m-1$,
\end{itemize}
or
\begin{itemize}
\item for all $V'\geq V$, $\Ic(V')$ is of backward type $\tB$ with $M=m$,
\end{itemize}
\item for all $V'\geq V$ and all $n\in\N$, the unique spectral bands $\left\{ I^{i}_{[\co,m]}(V')\right\} ^{M}_{i=1}$
of $\sigcm$ and the unique spectral bands $\left\{ I^{j}_{[\co,m,n]}(V')\right\} ^{M+1}_{j=1}$
of $\sigcmn$ satisfy \ref{enu: A1 property}, \ref{enu: B property}
and \ref{enu: I property} of Definition~\ref{def: forward type}.
\end{enumerate}
With this notion at hand, define 
\[
\crito([\co,m]):=\inf\set{V\geq0}{\begin{array}{c}
\text{each spectral band \ensuremath{\Ic} in \ensuremath{\sigma_{\co}}}\end{array}\text{satisfies \ensuremath{\Tquas(m,V)}}}.
\]

We note that the only difference compared to Definition~\ref{def:(mV)-type property}
is that here we do not require the property \ref{enu: A2 property}
to hold for the spectral bands $\left\{ I^{i}_{[\co,m]}\right\} ^{M}_{i=1}$.

Noting that 
\[
\crito([\co,m])\leq\crit([\co,m]),
\]
the strategy for proving Proposition~\ref{prop: Backward implies forward}
is to first show $\crito([\co,m])=0$ and afterwards $\crit([\co,m])=0$.
\begin{lem}
\label{lem: keyResult} Let $m\in\N$, $\co\in\Co$ be such $\varphi(\co)\not\in\left\{ -1,\infty\right\} $
and $[\co,m]\in\Co$, and $V_{0}\geq\crito([\co,m])$. Let $\Ic$
be a spectral band in $\sigc$ such that both of the following hold:
\begin{enumerate}
\item $I_{\co}(V)$ is either of backward type~$A$ for all $V>0$ or of
backward type~$B$ for all $V>0$,
\item $I^{1}_{[\co,m,1]}(V),~I^{M+1}_{[\co,m,1]}(V)\strict I_{\co}(V)$
holds for all $V\geq V_{0}$, where 
\[
M:=\begin{cases}
m-1 & I_{\co}(V)\textrm{~is~of~backward~type}~A~\textrm{for all }V>0,\\
m & I_{\co}(V)\textrm{~is~of~backward~type}~B~\textrm{for all }V>0.
\end{cases}
\]
\end{enumerate}
If $V_{0}>0$, then there is a $\delta>0$ such that $\Ic$ satisfies
$\Tquas(m,V_{0}-\delta)$.
\end{lem}

\begin{rem*}
The $\delta$ in the statement of the lemma only depends on the values
$\big|L\big(I^{1}_{[\co,m,1]}(V_{0})\big)-L\big(I_{\co}(V_{0})\big)\big|$,
$\big|R\big(I^{M+1}_{[\co,m,1]}(V_{0})\big)-R\big(I_{\co}(V_{0})\big)\big|$
and $V_{0}$. Here we use that the Lipschitz continuity in Corollary~\ref{prop: Lipschitz spectral edges}
is independent in $\co\in\Co$.
\end{rem*}
\begin{proof}
Since for $V>4$, $I_{\co}(V)$ is either of $m$-type $A$ or of
$m$-type $B$ (by Theorem~\ref{thm: V>4 Type A =000026 B}), we
may proceed assuming that $V_{0}\leq4$. Since $I_{\co}(V)$ is either
of backward type $A$ or backward type $B$ for all $V>0$, it is
sufficient to prove the existence of a $\delta>0$ such that 
\begin{equation}
\textrm{all }\Icmi~\textrm{and}~\Icmnj~\textrm{satisfy properties \ref{enu: A1 property}, \ref{enu: B1 property}, \ref{enu: B2 property} and \ref{enu: I property} for all}~n\in\N\label{eq: keyResult: to show}
\end{equation}

for all $V>V_{0}-\delta$. Since by the assumptions of the lemma,
$V_{0}\geq\crito$, we get that (\ref{eq: keyResult: to show}) holds
for all $V>V_{0}$. In particular, we infer that 
\begin{equation}
I^{i}_{[\co,m]}(V),I^{j}_{[\co,m,n]}(V)\strict\left[L\big(I^{1}_{[\co,m,1]}(V)\big),R\big(I^{M+1}_{[\co,m,1]}(V)\big)\right],\label{Eq-keyResult0}
\end{equation}
for all $V>V_{0}$, $1\leq i\leq M$, $1\leq j\leq M+1$ and $n\in\N$.

\begin{figure}[hbt]
\includegraphics[scale=1.1]{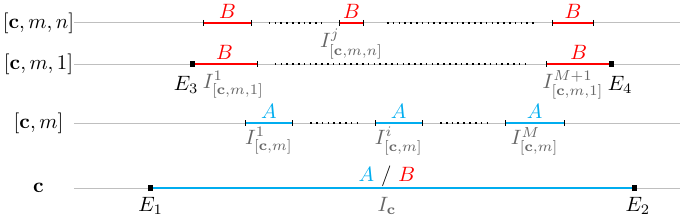} \caption{A sketch of the spectral bands considered in the proof of Lemma~\ref{lem: keyResult}.}
\label{Fig-keyResult1}
\end{figure}
Define 
\[
E_{1}(V_{0}):=L\big(I_{\co}(V_{0})\big),\qquad E_{2}(V_{0}):=R\big(I_{\co}(V_{0})\big),
\]
and 
\[
E_{3}(V_{0}):=L\big(I^{1}_{[\co,m,1]}(V_{0})\big),\qquad E_{4}(V_{0}):=R\big(I^{M+1}_{[\co,m,1]}(V_{0})\big),
\]
confer Figure~\ref{Fig-keyResult1}. Let $V\mapsto E(V)$ be a spectral
band edge of $I^{i}_{[\co,m]}(V)$ or $I^{j}_{[\co,m,n]}(V)$ for
some $1\leq i\leq M$ or $1\leq j\leq M+1$. By Corollary~\ref{prop: Lipschitz spectral edges},
the spectral band edges vary continuously in $V$, namely $V\mapsto E(V)$
is continuous, Thus, Equation~(\ref{Eq-keyResult0}) and assumption
(b) yield
\[
E_{1}(V_{0})<E_{3}(V_{0})\leq E(V_{0})\leq E_{4}(V_{0})<E_{2}(V_{0}).
\]
Hence, $\min_{k\in\{1,2\}}|E(V_{0})-E_{k}(V_{0})|\geq3\delta$ where
\[
\delta:=\frac{1}{3}\min\big\{|E_{1}(V_{0})-E_{3}(V_{0})|,~|E_{4}(V_{0})-E_{2}(V_{0})|,~V_{0}\big\}>0.
\]
Now, we use 
\[
\max\{|E(V)-E(V_{0})|,~|E_{i}(V)-E_{i}(V_{0})|\}\leq|V-V_{0}|,\qquad i\in\{1,2,3,4\},
\]
which holds by Corollary~\ref{prop: Lipschitz spectral edges}, to
conclude
\[
E_{3}(V)<E(V)<E_{4}(V),\quad V>V_{0},\qquad\Longrightarrow\qquad E_{1}(V)<E(V)<E_{2}(V),\qquad V>V_{0}-\delta.
\]

We note that $V_{0}-\delta>0$ holds, by the definition of $\delta$.

Since $E(V)$ was an arbitrary spectral edge of $I^{j}_{[\co,m]}(V)$
or $I^{j}_{[\co,m,n]}(V)$ for $n\in\N$, we deduce for all $V>V_{0}-\delta$,
$1\leq i\leq M$, $1\leq j\leq M+1$, and $n\in\N$, 
\[
I^{j}_{[\co,m]}(V)\strict I_{\co}(V)\qquad\text{and}\qquad I^{j}_{[\co,m,n]}(V)\strict I_{\co}(V).
\]
Now, we apply Corollary~\ref{cor: Properties A1+B2+I} which implies
that $I_{\co}(V)$ satisfies the forward properties \ref{enu: A1 property},
\ref{enu: B2 property} and \ref{enu: I property} for all $V>V_{0}-\delta$.
Since $I^{j}_{[\co,m,n]}(V)\strict\Ic(V)$ holds for all $n\in\N$
and $V>V_{0}-\delta$, Corollary~\ref{cor: Tower property} implies
that $\Ic(V)$ satisfies \ref{enu: B1 property} for all $V>V_{0}-\delta$.
\end{proof}

We now apply Lemma~\ref{lem: keyResult} for all spectral bands $\Ic$
in $\sigc$. Using that the number of spectral bands in $\sigc$ is
finite and taking the minimum $\delta$ among all $\Ic$ in $\sigc$
(the $\delta$ which is provided by Lemma~\ref{lem: keyResult}),
we get:
\begin{cor}
\label{Cor-keyResult} Let $V_{1}>0$, $m\in\N$ and $\co\in\Co$
be such that $\varphi(\co)\not\in\left\{ -1,\infty\right\} $ and
$[\co,m]\in\Co$. Suppose that each spectral band $\Ic$ in $\sigma_{\co}$
satisfies both of the following:
\begin{enumerate}
\item \label{enu: Cor-keyResult - 1}$I_{\co}(V)$ is either of backward
type~$A$ for all $V>0$ or of backward type~$B$ for all $V>0$,
\item \label{enu: Cor-keyResult - 2}$I^{1}_{[\co,m,1]}(V),~I^{M+1}_{[\co,m,1]}(V)\strict I_{\co}(V)$
holds for all $V\geq V_{1}$, where 
\[
M:=\begin{cases}
m-1 & I_{\co}(V)\textrm{~is~of~backward~type}~A~\textrm{for all }V>0,\\
m & I_{\co}(V)\textrm{~is~of~backward~type}~B~\textrm{for all }V>0.
\end{cases}
\]
\end{enumerate}
Then $\crito([\co,m])<V_{1}$. In particular, if (b) holds for all
$V_{1}>0$, then $\crito([\co,m])=0$.
\end{cor}

\begin{proof}
Set $V_{0}:=\crito([\co,m])$. First of all note that $V_{0}\leq4$
by Theorem~\ref{thm: V>4 Type A =000026 B}. We seek to prove $V_{0}<V_{1}$.
Assume by contradiction $V_{0}\geq V_{1}$.

Let $I_{\co}(V)$ be a spectral band in $\sigma_{\co}(V)$, which
by (\ref{enu: Cor-keyResult - 1}) is either of backward type A for
all $V>0$ or of backward type $B$ for all $V>0$. Since $V_{0}\geq V_{1}$,
(\ref{enu: Cor-keyResult - 2}) implies 
\[
I^{1}_{[\co,m,1]}(V_{0}),~I^{M+1}_{[\co,m,1]}(V_{0})\strict I_{\co}(V_{0}).
\]
Due to Lemma~\ref{lem: keyResult} and $V_{0}\geq V_{1}>0$, there
exists a $\delta:=\delta(I_{\co}(V))>0$ such that $\Ic$ satisfies
$\Tquas(m,V_{0}-\delta)$. Since there are at most finitely spectral
bands in $\sigma_{\co}(V)$, we can take the minimum of all these
$\delta(I_{\co}(V))$'s and denote it by $\delta'>0$. Then every
spectral band $\Ic$ in $\sigc$ satisfies $\Tquas(m,V_{0}-\delta')$.
Hence, by definition of $\crito([\co,m])$, we conclude 
\[
V_{0}:=\crito([\co,m])\leq V_{0}-\delta',
\]
a contradiction.
\end{proof}

By adding an additional condition to the assumption of Lemma~\ref{lem: keyResult}
we may get a stronger implication, which is done in the following
lemma.
\begin{lem}
\label{Lem-keyResult_A2} Let $m\in\N$, $\co\in\Co$ with $\varphi(\co)\not\in\left\{ -1,\infty\right\} $
and $[\co,m]\in\Co$, and $V_{0}\geq\crit([\co,m])$. Let $\Ic$ be
a spectral band in $\sigc$ such that all of the following hold:
\begin{enumerate}
\item $I_{\co}(V)$ is either of backward type~$A$ for all $V>0$ or of
backward type~$B$ for all $V>0$,
\item $I^{1}_{[\co,m,1]}(V),~I^{M+1}_{[\co,m,1]}(V)\strict I_{\co}(V)$
holds for all $V\geq V_{0}$, where 
\[
M:=\begin{cases}
m-1 & I_{\co}(V)\textrm{~is~of~backward~type}~A~\textrm{for all }V>0,\\
m & I_{\co}(V)\textrm{~is~of~backward~type}~B~\textrm{for all }V>0,
\end{cases}
\]
\item if $m=1$, then $\varphi(\co)\neq1$,\\
 if $m\geq2$, then $\crito([\co,m-1])=0$.
\end{enumerate}
If $V_{0}>0$, then there is a $\delta>0$ such that $\Ic$ satisfies
$\TmV(m,V_{0}-\delta)$.

\end{lem}

\begin{proof}
Applying Lemma~\ref{lem: keyResult}, which is justified by assumptions
(a) and (b) here, there exists $\delta>0$ such that $\Ic$ satisfies
$\Tquas(m,V_{0}-\delta)$. Thus, we only have to show that $I_{\co}(V)$
also satisfies the forward property \ref{enu: A2 property} for all
$V>V_{0}-\delta$. We consider the following two cases:

\underline{Case 1: ($m=1$)} If $m=1$ and $I_{\co}(V)$ is of backward
type $A$ then it does not contain any spectral band of $\sigma_{[\co,1]}(V)$
for all $V>4$ by Theorem~\ref{thm: V>4 Type A =000026 B}. Hence,
there are no $I^{i}_{[\co,1]}$ spectral bands (see Definition~\ref{def: Notion Icmi. Icmnj and M})
and there is nothing to prove in this case. We need only to deal with
the case $m=1$ when $I_{\co}(V)$ is of backward type $B$. Towards
doing this, notice that if $\varphi(\co)=0$ then $\co=[0,0]$ since
$[\co,m]\in\Co$ is assumed. But, $\sigma_{[0,0]}(V)=[-2,2]$ only
consists of a backward type $A$ band, see Example~\ref{exa:A-B types}.
Hence, when checking the case that $I_{\co}(V)$ is of backward type
$B$, we may further assume $\varphi(\co)\neq0$.

Combining this with condition (c) of the lemma, we may now assume
that $\varphi(\co)\in(0,1)$ and $I_{\co}(V)$ is of backward type
$B$. By Definition~\ref{def: Notion Icmi. Icmnj and M}, since $m=1$,
we have exactly one spectral band $I^{1}_{[\co,1]}(V)$ for which
we need to show that it is not of weak backward type $B$ for all
$V>V_{0}-\delta$. Indeed, Lemma~\ref{lem: task (B1)} implies that
for all $V>V_{0}-\delta$, $I^{1}_{[\co,1]}(V)\strict I_{\co}(V)$
and $I^{1}_{[\co,1]}(V)$ is not of weak backward type $B$.

\underline{Case 2: ($m\geq 2$)} We need to show that $I^{i}_{[\co,m]}(V)$
is not of weak backward type $B$ for all $1\leq i\leq M$ and $V>V_{0}-\delta$.
Let $1\leq i\leq M$. We know by Theorem~\ref{thm: V>4 Type A =000026 B}
that $\Icmi(V)$ is of backward type $A$ in $\sigcm(V)$ for $V>4$.
Denoting $m':=m-1\geq1$, Proposition~\ref{prop: duality of A-B types}
implies that $\Icmi$ equals to the spectral band $I^{i}_{[\co,m',1]}$
in $\sigma_{[\co,m',1]}$, which is of backward type $B$ for $V>4$.
Using Proposition~\ref{prop: duality of A-B types} again, it suffices
to show that $I^{i}_{[\co,m',1]}(V)$ is not of weak backward type
$A$ in $\sigma_{[\co,m',1]}(V)$ for all $V>V_{0}-\delta$.

By assumption (c) for $m\ge2$, we have $\crito([\co,m'])=\crito([\co,m-1])=0$.
Hence, for all $V>0$, $\Ic$ satisfies $\Tquas(m-1,V)$. This implies
by \ref{enu: B2 property} that $I^{i}_{[\co,m',1]}(V)$ is not of
weak backward type $A$ for all $V>0$ .
\end{proof}

We now apply Lemma~\ref{Lem-keyResult_A2} for all spectral bands
$\Ic$ in $\sigc$. Using that the number of spectral bands in $\sigc$
is finite and taking the minimum $\delta$ among all $\Ic$ in $\sigc$
(the $\delta$ which is provided by Lemma~\ref{Lem-keyResult_A2}),
we get:
\begin{cor}
\label{Cor-keyResult_A2} Let $V_{1}>0$, $m\in\N$ and $\co\in\Co$
be such that $\varphi(\co)\not\in\left\{ -1,\infty\right\} $ and
$[\co,m]\in\Co$. Suppose that each spectral band $I_{\co}$ in $\sigc$
satisfies all of the following:
\begin{enumerate}
\item \label{enu: Cor-keyResult_A2 - 1}$I_{\co}(V)$ is either of backward
type~$A$ for all $V>0$ or of backward type~$B$ for all $V>0$,
\item \label{enu: Cor-keyResult_A2 - 2}$I^{1}_{[\co,m,1]}(V),~I^{M+1}_{[\co,m,1]}(V)\strict I_{\co}(V)$
holds for all $V\geq V_{1}$, where 
\[
M:=\begin{cases}
m-1 & I_{\co}(V)\textrm{~is~of~backward~type}~A~\textrm{for all }V>0,\\
m & I_{\co}(V)\textrm{~is~of~backward~type}~B~\textrm{for all }V>0,
\end{cases}
\]
\item \label{enu: Cor-keyResult_A2 - 3}if $m=1$, then $\varphi(\co)\neq1$,\\
 if $m\geq2$, then $\crito([\co,m-1])=0$.
\end{enumerate}
Then $\crit([\co,m])<V_{1}$. In particular, if (b) holds for all
$V_{1}>0$, then $\crit([\co,m])=0$.

\end{cor}

\begin{proof}
Similarly as in Corollary~\ref{Cor-keyResult}, this follows immediately
from Lemma~\ref{Lem-keyResult_A2} and the fact that $\sigma_{\co}(V)$
consists only of finitely many spectral bands independent of $V>0$.
\end{proof}

Finally, we are ready to prove Proposition~\ref{prop: Backward implies forward}.
\begin{proof}
[Proof of Proposition \ref{prop: Backward implies forward}] Since
$\varphi(\co)\in(0,1)$ and each spectral band $I_{\co}(V)$ in $\sigma_{\co}(V)$
is either of backward type $A$ or $B$ for all $V>0$, Lemma~\ref{lem: Backward Implies Bands strictly included}
implies 
\begin{equation}
I^{1}_{[\co,m,1]}(V),I^{M+1}_{[\co,m,1]}(V)\strict I_{\co}(V),\qquad\textrm{for all }m\in\N,\;V>0,\label{Eq-Propo-v_c(=00005Bc.m=00005D)=00003D0}
\end{equation}
where $M=m-1$ if $\Ic$ is of backward type $A$ and $M=m$ if $\Ic$
is of backward type $B$. Now $\crit([\co,m])=0$ is proven by induction
over $m\in\N$.

For the induction base, let $m=1$. Since Equation~\eqref{Eq-Propo-v_c(=00005Bc.m=00005D)=00003D0}
holds for $m=1$ and $\varphi(\co)\neq1$, Corollary~\ref{Cor-keyResult_A2}
(for $m=1$) implies $\crit([\co,1])=0$.

For the induction step, let $m\in\N$ be such that $\crit([\co,m])=0$.
Thus, $\crito([\co,m])=0$ follows as $0\leq\crito([\co,m])\leq\crit([\co,m])$.
Since Equation~\eqref{Eq-Propo-v_c(=00005Bc.m=00005D)=00003D0} holds
for $m+1$ and $\crito([\co,m])=0$, Corollary~\ref{Cor-keyResult_A2}
(for $m+1\geq2$) implies $\crit([\co,m+1])=0$.
\end{proof}

\section{The induction base of the main theorem\label{sec: Pf_InductionBase}}

This section contains the proof of the induction base which is used
in the proof of Theorem~\ref{thm: Every band is A or B}. Specifically,
we show in this section that for all $V\neq0$, the spectral bands
in $\sigma_{[0,0]}(V)$ and $\sigma_{[0,0,1]}(V)$ are either of type
$A$ or $B$. For this proof we express the transfer matrices, $M_{\co}(E,V)$,
and their traces, $\tc(E,V)$, (see Section~\ref{subsec: spectra via transfer matrices})
using the dilated Chebyshev polynomials of the second kind $S_{l}:\R\to\R,\,l\in\N_{0}$.
These polynomials are defined by 
\[
S_{-1}(x):=0,\quad S_{0}(x):=1,\quad S_{l}(x):=xS_{l-1}(x)-S_{l-2}(x),
\]
see Appendix~\ref{App: Trace Maps} for more details and properties
of these polynomials.
\begin{lem}
\label{lem: powers of M_0.0} For all $m\in\N$ and $V\in\R$, we
have
\[
M^{m}_{[0,0]}(E,V)=\begin{pmatrix}S_{m}(E) & -S_{m-1}(E)\\
S_{m-1}(E) & -S_{m-2}(E)
\end{pmatrix},\qquad E\in\R.
\]
\end{lem}

\begin{proof}
We prove this by induction on $m$. The induction base ($m=1$) follows
just by definition as 
\[
M^{1}_{[0,0]}(E,V)=\begin{pmatrix}E & -1\\
1 & 0
\end{pmatrix}=\begin{pmatrix}S_{m}(E) & -S_{m-1}(E)\\
S_{m-1}(E) & -S_{m-2}(E)
\end{pmatrix}
\]
using that $S_{1}(E)=S_{0}(E)E-S_{-1}(E)=E$. For the induction step,
suppose the statement is true for $m$. Then 
\begin{align*}
M^{m+1}_{[0,0]}(E,V)=M_{[0,0]}(E,V)M^{m}_{[0,0]}(E,V)= & \begin{pmatrix}E & -1\\
1 & 0
\end{pmatrix}\begin{pmatrix}S_{m}(E) & -S_{m-1}(E)\\
S_{m-1}(E) & -S_{m-2}(E)
\end{pmatrix}\\
= & \begin{pmatrix}ES_{m}(E)-S_{m-1}(E) & -ES_{m-1}(E)+S_{m-2}(E)\\
S_{m}(E) & -S_{m-1}(E)
\end{pmatrix}\\
= & \begin{pmatrix}S_{m+1}(E) & -S_{m}(E)\\
S_{m}(E) & -S_{m-1}(E)
\end{pmatrix}
\end{align*}
proving the statement.
\end{proof}

\begin{lem}
\label{lem: basic traces t_0.0} For all $E,V\in\R$ and $m\in\N$
the following holds:
\begin{enumerate}
\item \label{enu: lem-basic traces t_0.0 - 1}$t_{[0,0,m]}(E,V)=S_{m}(E)-VS_{m-1}(E)-S_{m-2}(E)$
for all $E\in\R$.
\item \label{enu: lem-basic traces t_0.0 - 2}$t_{[0,0,1,m]}(E,V)=ES_{m}(E-V)-2S_{m-1}(E-V)$
for all $E\in\R$.
\item \label{enu: lem-basic traces t_0.0 - 3}$t_{[0,0,1,m,1]}(E,V)=ES_{m+1}(E-V)-2S_{m}(E-V)$
for all $E\in\R$.
\end{enumerate}
\end{lem}

\begin{proof}
We recall (Section~\ref{subsec: spectra via transfer matrices})
that the transfer matrices are recursively defined by 
\[
M_{[0,0,c_{1},\ldots,c_{k}]}(E,V):=M_{[0,0,c_{1}\ldots,c_{k-2}]}(E,V)M_{[0,0,c_{1},\ldots,c_{k-1}]}(E,V)^{c_{k}}.
\]

(\ref{enu: lem-basic traces t_0.0 - 1}) Using Lemma~\ref{lem: powers of M_0.0},
we get 
\begin{align*}
t_{[0,0,m]}(E,V)= & \tr\left(M_{[0,0,m]}(E,V)\right)=\tr\left(M_{[0]}(E,V)M^{m}_{[0,0]}(E,V)\right)\\
= & \tr\left(\begin{pmatrix}1 & -V\\
0 & 1
\end{pmatrix}\begin{pmatrix}S_{m}(E) & -S_{m-1}(E)\\
S_{m-1}(E) & -S_{m-2}(E)
\end{pmatrix}\right)\\
= & S_{m}(E)-VS_{m-1}(E)-S_{m-2}(E).
\end{align*}
(\ref{enu: lem-basic traces t_0.0 - 2}) We first observe that
\begin{align*}
M_{[0,0,1]}(E,V)= & M_{[0]}(E,V)M_{[0,0]}(E,V)=\begin{pmatrix}1 & -V\\
0 & 1
\end{pmatrix}\begin{pmatrix}E & -1\\
1 & 0
\end{pmatrix}\\
= & \begin{pmatrix}E-V & -1\\
1 & 0
\end{pmatrix}=M_{[0,0]}(E-V,V).
\end{align*}
Thus, Lemma~\ref{lem: powers of M_0.0} leads to
\begin{align*}
t_{[0,0,1,m]}= & \tr\left(M_{[0,0]}(E,V)M_{[0,0,1]}(E,V)^{m}\right)\\
= & \tr\left(M_{[0,0]}(E,V)M_{[0,0]}(E-V,V)^{m}\right)\\
= & \tr\left(\begin{pmatrix}E & -1\\
1 & 0
\end{pmatrix}\begin{pmatrix}S_{m}(E-V) & -S_{m-1}(E-V)\\
S_{m-1}(E-V) & -S_{m-2}(E-V)
\end{pmatrix}\right)\\
= & ES_{m}(E-V)-2S_{m-1}(E-V).
\end{align*}

(\ref{enu: lem-basic traces t_0.0 - 3}) This follows from (\ref{enu: lem-basic traces t_0.0 - 2})
and Proposition~\ref{Prop-traceMaps}~(\ref{enu:.Prop-traceMaps - 1})
asserting $t_{[0,0,1,m,1]}=t_{[0,0,1,m+1]}$.
\end{proof}

\begin{example}
\label{exa: first spectral bands} We explicitly write here the expressions
of a few traces, which will turn to be useful in the sequel. We have
\[
t_{[0,0,1,-1]}(E,V)=t_{[0]}(E,V)=t_{[0,0,0]}(E,V)=2,
\]
and
\begin{align*}
t_{[0,0]}(E,V)= & E, & t_{[0,0,-1]}(E,V)= & E+V,\\
t_{[0,0,1]}(E,V)= & E-V, & t_{[0,0,2]}(E,V)= & E^{2}-EV-2,\\
t_{[0,0,1,2]}(E,V)= & E^{3}-2E^{2}V+E(V^{2}-3)+2V, & t_{[0,0,3]}(E,V)= & E^{3}-E^{2}V-3E+V.
\end{align*}
\end{example}

Next, we prove two lemmata. The first lemma states that the spectral
band $I_{[0,0]}(V):=[-2,2]$ (in $\sigma_{[0,0]}(V)$) is of type
$A$. The second lemma states that the spectral band $I_{[0,0,1]}(V):=[-2+V,2+V]$
(in $\sigma_{[0,0]}(V)$) is of type $B$. Hence, both lemmata provide
the induction base needed to prove Theorem~\ref{thm: Every band is A or B}.
\begin{lem}
\label{lem: I_0.0 properties} Let $I_{[0,0]}(V):=[-2,2]$ be the
unique spectral band of $\sigma_{[0,0]}(V)$ for $V>0$. The following
assertions hold for all $V>0$.
\begin{enumerate}
\item \label{enu: lem-I_0.0 properties - 1}$I_{[0,0]}(V)$ is of backward
type $A$ and not of weak backward type $B$,
\item \label{enu: lem-I_0.0 properties - 2}For all $m\in\N$, $I_{[0,0]}(V)$
is of $m$-type $A$, namely $\crit([0,0,m])=0$,
\item \label{enu: lem-I_0.0 properties - 3}For all $m\in\N$, $\sigma_{[0,0,m]}(V)$
consists of $m$ spectral bands satisfying
\begin{enumerate}
\item[$\bullet$] the left-most $m-1$ spectral bands are of backward type $A$ and
not of weak backward type $B$. These spectral bands are strictly
contained in $I_{[0,0]}(V)$.
\item[$\bullet$] the right-most spectral band, which we denote $K_{[0,0,m]}(V)$,
is of backward type $B$ but not of weak backward type $A$. The spectral
bands $K_{[0,0,m]}(V)$ (one for each $m\in\N$) satisfy 
\[
I_{[0,0]}(V)\prec K_{[0,0,m]}(V)
\]
and
\[
K_{[0,0,m]}(V)\strict K_{[0,0,m-1]}(V)\strict\ldots\strict K_{[0,0,1]}(V)\strict K_{[0]}(V),
\]
with the notational convention $K_{[0]}(V):=\R=\sigma_{[0]}(V)$.
\end{enumerate}
\end{enumerate}
\end{lem}

\begin{figure}[htb]
\includegraphics[scale=1.2]{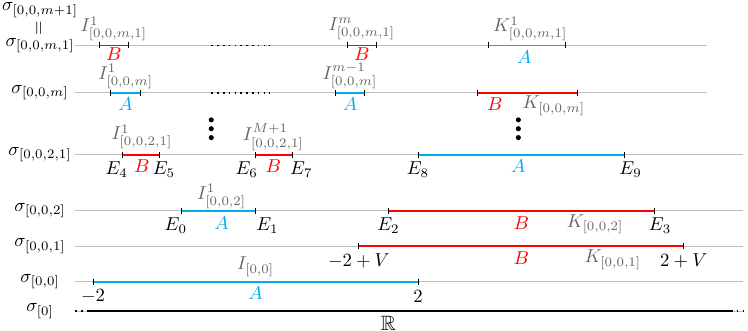} \caption{A sketch of the spectral bands considered in the proof of Lemma~\ref{lem: I_0.0 properties}.}
\label{Fig-=00005B0.0=00005D_m-forward_A}
\end{figure}
For the convenience of the reader, we sketch in Figure~\ref{Fig-=00005B0.0=00005D_m-forward_A}
the spectral bands mentioned in the proof of Lemma~\ref{lem: I_0.0 properties}
and their relations ($\prec$ and $\strict$).
\begin{proof}
Part (\ref{enu: lem-I_0.0 properties - 1}) follows from $\sigma_{[0,0,0]}(V)=\sigma_{[0]}(V)=\R$
and $\sigma_{[0,0,-1]}(V)=[-2-V,2-V]$.

Parts (\ref{enu: lem-I_0.0 properties - 2}) and (\ref{enu: lem-I_0.0 properties - 3})
of the lemma are proven together by induction over $m\in\N$. Towards
proving them we denote $M:=m-1$, since we are trying to prove that
$I_{[0,0]}$ is of $m$-forward type $A$, confer Definition~\ref{def: forward type}.

\textbf{Induction base:} The induction base consists of $m=1$ and
$m=2$. We start with proving (\ref{enu: lem-I_0.0 properties - 2})
and (\ref{enu: lem-I_0.0 properties - 3}) for $m=1$ and $m=2$.

Let $m=1$. Then $\sigma_{[0,0,1]}(V)=[-2+V,2+V]$ and its unique
spectral band $K_{[0,0,1]}:=[-2+V,2+V]$ is of backward type $B$
but not of weak backward type $A$ (as it is strictly contained in
$\sigma_{[0]}(V)=\R=K_{[0]}(V)$ and $\sigma_{[0,0]}(V)=I_{[0,0]}(V)\prec K_{[0,0,1]}(V)$).
This proves (\ref{enu: lem-I_0.0 properties - 3}) for $m=1$.

Let $m=2$. We have $t_{[0,0,2]}(E,V)=E^{2}-EV-2$ (see Example~\ref{exa: first spectral bands})
 implying 
\begin{align*}
t_{[0,0,2]}(E,V)=2 & \quad\Longleftrightarrow\quad E=\frac{V}{2}\pm\sqrt{\frac{V^{2}}{4}+4},\\
t_{[0,0,2]}(E,V)=-2 & \quad\Longleftrightarrow\quad E\in\{0,V\}.
\end{align*}
This motivates to denote
\[
E_{0}(V):=\frac{V}{2}-\sqrt{\frac{V^{2}}{4}+4},\quad E_{1}(V):=0,\quad E_{2}(V):=V,\quad E_{3}(V):=\frac{V}{2}+\sqrt{\frac{V^{2}}{4}+4},
\]
so that $E_{0}(V)<E_{1}(V)<E_{2}(V)<E_{3}(V)$ holds for all $V>0$.
Thus, $I_{[0,0,2]}(V):=[E_{0}(V),E_{1}(V)]$ and $K_{[0,0,2]}(V):=[E_{2}(V),E_{3}(V)]$
are the two spectral bands in $\sigma_{[0,0,2]}(V)$. Clearly $I_{[0,0,2]}(V)$
is of backward type $A$ for all $V>0$ since $I_{[0,0,2]}(V)\strict[-2,2]=I_{[0,0]}(V)$.
In addition $E_{0}(V)<-2+V$ for $V>0$ and so $I_{[0,0,2]}(V)$ is
not contained in $K_{[0,0,1]}(V)=[-2+V,2+V]$. Hence, $I_{[0,0,2]}(V)$
is not of weak backward type $B$ for all $V>0$. The spectral band
$K_{[0,0,2]}(V)$ is of backward type $B$ since $K_{[0,0,2]}(V)\strict[-2+V,2+V]=K_{[0,0,1]}(V)$
for all $V>0$. In addition, $E_{3}(V)\geq\frac{V}{2}+2>2$ for all
$V>0$ leading to $I_{[0,0]}(V)\prec K_{[0,0,2]}(V)$. Thus, $K_{[0,0,2]}(V)$
is not of weak backward type $A$.

Summing up, we have proven (\ref{enu: lem-I_0.0 properties - 2})
and (\ref{enu: lem-I_0.0 properties - 3}), for $m=1$ and $m=2$.
It is left to show (\ref{enu: lem-I_0.0 properties - 2}) for $m=1$
and $m=2$, namely that $\crit([0,0,1])=0$ and $\crit([0,0,2])=0$.

\underline{$\crit([0,0,1])=0$:} We first aim to apply Corollary~\ref{Cor-keyResult}
for $\co=[0,0]$ and $m=1$. Applying Corollary~\ref{Cor-keyResult}
would give that $\crito([0,0,1])=0$ (see Definition~\ref{def:quasi (mV)-type property})
and then one needs only to show property \ref{enu: A2 property} in
order to conclude $\crit([0,0,1])=0$. But, in this case property
\ref{enu: A2 property} is an empty statement since $M=m-1=0$.

In order to apply Corollary~\ref{Cor-keyResult} for $\co=[0,0]$
and $m=1$ we note the following. The spectrum $\sigc$ has only one
spectral band $I_{[0,0]}$ that is of backward type $A$ and not of
weak backward type $B$ for all $V>0$, which is assumption (\ref{enu: Cor-keyResult - 1})
of Corollary~\ref{Cor-keyResult}. To check assumption (\ref{enu: Cor-keyResult - 2})
of Corollary~\ref{Cor-keyResult} we need to prove that $I^{1}_{[0,0,1,1]}(V)\strict I_{[0,0]}(V)$
for all \textbf{$V>0$.}

By Proposition~\ref{prop: duality of A-B types}, $\sigma_{[0,0,1,1]}=\sigma_{[0,0,2]}$.
Consider the spectral band $I_{[0,0,2]}(V)=[E_{0}(V),E_{1}(V)]$,
which we calculated above. In particular, we have seen above that
$I_{[0,0,2]}(V)\strict I_{[0,0]}(V)$ holds for all $V>0$. Therefore
$I_{[0,0,2]}(V)$ equals to the unique spectral band $I^{1}_{[0,0,1,1]}(V)$
by Definition~\ref{def: Notion Icmi. Icmnj and M} . Thus, $I^{1}_{[0,0,1,1]}(V)\strict I_{[0,0]}(V)$
for all \textbf{$V>0$}, which verifies all the assumptions of Corollary~\ref{Cor-keyResult}.
As explained above, we conclude $\crit([0,0,1])=0$.

\underline{$\crit([0,0,2])=0$:} We aim to apply Corollary~\ref{Cor-keyResult_A2}
for $\co=[0,0]$ and $m=2$ in order to conclude $\crit([0,0,2])=0$.
Condition (\ref{enu: Cor-keyResult_A2 - 1}) of Corollary~\ref{Cor-keyResult_A2}
was already verified above, as the spectral band $I_{[0,0]}$ that
is of backward type $A$ and not of weak backward type $B$ for all
$V>0$. We have also proved above $\crito([0,0,1])=0$, which verifies
condition (\ref{enu: Cor-keyResult_A2 - 3}) of Corollary~\ref{Cor-keyResult_A2}.
We only have to check condition (\ref{enu: Cor-keyResult_A2 - 2})
for all $V>0$. Specifically, it is sufficient to prove $I^{1}_{[0,0,2,1]}(V),~I^{M+1}_{[0,0,2,1]}(V)\strict I_{[0,0]}(V)$
for all $V>0$. Note that $M+1=m=2$.

Using Proposition~\ref{Prop-traceMaps}~(\ref{enu:.Prop-traceMaps - 1})
and Lemma~\ref{lem: basic traces t_0.0} we conclude 
\[
t_{[0,0,2,1]}(E,V)=t_{[0,0,3]}(E,V)=S_{3}(E)-VS_{2}(E)-S_{1}(E)
\]
We use this to express all the $E$ values for which $t_{[0,0,2,1]}(E,V)\in\{-2,2\}$:
\begin{align*}
E_{4}(V):= & \frac{V-1}{2}-\frac{\sqrt{V^{2}+2V+9}}{2}, & E_{5}(V) & :=-1,\\
E_{6}(V):= & \frac{V+1}{2}-\frac{\sqrt{V^{2}-2V+9}}{2}, & E_{7}(V) & :=1,\\
E_{8}(V):= & \frac{V-1}{2}+\frac{\sqrt{V^{2}+2V+9}}{2}, & E_{9}(V) & :=\frac{V+1}{2}+\frac{\sqrt{V^{2}-2V+9}}{2},
\end{align*}
where $E_{4}(V)<E_{5}(V)<E_{6}(V)<E_{7}(V)<E_{8}(V)<E_{9}(V)$. Now,
it is straightforward to check that the three spectral bands in $\sigma_{[0,0,2,1]}$
are

\[
I^{1}_{[0,0,2,1]}(V)=[E_{4}(V),E_{5}(V)],\qquad I^{M+1}_{[0,0,2,1]}(V)=[E_{6}(V),E_{7}(V)]
\]
and
\[
K_{[0,0,3]}(V)=[E_{8}(V),E_{9}(V)].
\]
Furthermore, $I^{1}_{[0,0,2,1]}(V),I^{M+1}_{[0,0,2,1]}(V)\strict I_{[0,0]}(V)$
for all $V>0$. Thus, Corollary~\ref{Cor-keyResult_A2} implies $\crit([0,0,2])=0$,
hence statement (\ref{enu: lem-I_0.0 properties - 2}) of the current
lemma holds for $m=2$, and this finishes the proof of the induction
base.

\textbf{Induction step:} (see Figure~\ref{Fig-=00005B0.0=00005D_m-forward_A})
Let $m\geq2$ and suppose (induction hypothesis) that $\crit([0,0,m])=0$
and $\sigma_{[0,0,m]}(V)$ satisfies (\ref{enu: lem-I_0.0 properties - 3})
for all $V>0$.

We have $\varphi([0,0,m+1])=\frac{1}{m+1}$ and so $\sigma_{[0,0,m+1]}(V)$
consists of exactly $m+1$ spectral bands by Proposition~\ref{prop: Basic spectral prop periodic}
and Lemma~\ref{lem: ConnecSpectrTrace}.  Since $K_{[0,0,m]}(V)$
is of backward type $B$ (and not of weak backward type $A$) for
all $V>0$, we conclude that $K_{[0,0,m]}(V)$ is of type $B$ for
$V>4$, see Theorem~\ref{thm: V>4 Type A =000026 B}. Then property
\ref{enu: A1 property} of $K_{[0,0,m]}$ implies that for $V>4$,
there is a spectral band $K_{[0,0,m,1]}(V)$ in $\sigma_{[0,0,m,1]}(V)$
of backward type $A$ such that $K_{[0,0,m,1]}(V)\strict K_{[0,0,m]}(V)$.

Furthermore (referring again to Theorem~\ref{thm: V>4 Type A =000026 B}),
$I_{[0,0]}(V)=[-2,2]$ is of type $A$ for $V>4$ and so it strictly
contains $m$ spectral bands of type $B$ in $\sigma_{[0,0,m,1]}(V)$
for $V>4$. Since $\sigma_{[0,0,m,1]}(V)$ has $m+1$ spectral bands,
the spectral band $K_{[0,0,m,1]}(V)$ mentioned above satisfies the
following: for $V>4$, it is the unique spectral band in $\sigma_{[0,0,m,1]}(V)$
of backward type $A$. In addition $K_{[0,0,m,1]}(V)\strict K_{[0,0,m]}(V)$
(as seen above) and $I_{[0,0]}(V)\prec K_{[0,0,m,1]}(V)$ for $V>4$
(since $I_{[0,0]}(V)\prec K_{[0,0,m]}(V)$ by the induction hypothesis).
Furthermore, the left-most $m$ spectral bands in $\sigma_{[0,0,m,1]}(V)$
are strictly contained in $\sigma_{[0,0]}(V)$ for $V>4$. By Proposition~\ref{prop: duality of A-B types},
$\sigma_{[0,0,m,1]}(V)=\sigma_{[0,0,m+1]}(V)$, and in particular,
we can identify $K_{[0,0,m,1]}(V)$ with a spectral band $K_{[0,0,m+1]}(V)$
in $\sigma_{[0,0,m+1]}(V)$ and $K_{[0,0,m+1]}(V)$ is of backward
type $B$ for $V>4$.

We will show that
\begin{enumerate}
\item[(d)] $K_{[0,0,m+1]}(V)\strict K_{[0,0,m]}(V)$ and $K_{[0,0,m+1]}(V)\not\subseteq I_{[0,0]}(V)$
for all $V>0$,
\item[(e)] $I_{[0,0]}(V)\prec K_{[0,0,m+1]}(V)$ for all $V>0$,
\item[(f)] $\crit([0,0,m+1])=0$.
\end{enumerate}
Observe that these statements imply that parts (\ref{enu: lem-I_0.0 properties - 2})
and (\ref{enu: lem-I_0.0 properties - 3}) of the lemma hold for $m+1$.
These implications are rather straightforward, and one just needs
to notice that to get the first bullet of (\ref{enu: lem-I_0.0 properties - 3})
for $m+1$, one needs also to employ $\crit([0,0,m+1])=0$, which
provides the \ref{enu: A1 property} and \ref{enu: A2 property} properties
of $I_{[0,0]}(V)$ for all $V>0$.

\uline{Proof of (d):} Since $m\geq2$, we have $\varphi([0,0,m])\in(0,1)$
and by induction hypothesis $K_{[0,0,m]}(V)$ is of backward type
$B$ for all $V>0$. Proposition~\ref{prop: duality of A-B types}
implies $\sigma_{[0,0,m,1]}=\sigma_{[0,0,m+1]}$. Thus, Lemma~\ref{lem: task (B1)}
applied to the spectral band $K_{[0,0,m]}(V)$ implies $K_{[0,0,m+1]}(V)=K_{[0,0,m,1]}(V)\strict K_{[0,0,m]}(V)$
for all $V>0$. Moreover, Lemma~\ref{lem: task (B1)} asserts that
$K_{[0,0,m,1]}(V)$ is not of weak backward type $B$ for all $V>0$,
namely $K_{[0,0,m,1]}(V)$ is not contained in a spectral band of
$\sigma_{[0,0,m,1,-1]}=\sigma_{[0,0]}=[-2,2]$. Since $K_{[0,0,m+1]}(V)=K_{[0,0,m,1]}(V)$
holds for all $V>0$, we conclude $K_{[0,0,m+1]}(V)\not\subseteq[-2,2]$
for $V>0$.

\uline{Proof of (e):} For $V>0$, (d) and $[-2,2]=I_{[0,0]}(V)\prec K_{[0,0,m]}(V)$
imply 
\[
-2=L(I_{[0,0]}(V))<L(K_{[0,0,m]}(V))<L(K_{[0,0,m+1]}(V)).
\]
Furthermore, (d) asserts $K_{[0,0,m+1]}(V)\not\subseteq[-2,2]$ for
all $V>0$ implying $R(I_{[0,0]}(V))<R(K_{[0,0,m+1]}(V))$ for all
$V>0$. Thus, $I_{[0,0]}(V)\prec K_{[0,0,m+1]}(V)$ follows for $V>0$.

\uline{Proof of (f):} Since $\sigma_{[0,0]}(V)$ consists only of
the spectral band $I_{[0,0]}(V)$ we need to show that $I_{[0,0]}(V)$
is of $(m+1)$-type $A$ for all $V>0$. Since $\crit([0,0,m])=0$
and $I_{[0,0]}(V)$ if of backward type $A$ but not of weak backward
type $B$ for all $V>0$, Corollary~\ref{Cor-keyResult_A2} (for
$m\geq2$) asserts that we only have to prove
\[
I^{1}_{[0,0,m+1,1]}(V),I^{M+1}_{[0,0,m+1,1]}(V)\strict I_{[0,0]}(V)
\]
for all $V>0$ where $M=(m+1)-1=m$. Lemma~\ref{Lem-Chebyshev-Traces}~(\ref{enu: Chebychev - |x|=00003D2})
asserts $S_{l}(\pm2)=(\pm1)^{l}(l+1)$ for $l\in\N$. Thus, Lemma~\ref{lem: basic traces t_0.0}~(\ref{enu: lem-basic traces t_0.0 - 1})
leads to 
\begin{align}
\big|t_{[0,0,l+1]}(\pm2,V)\big|= & \big|S_{l+1}(\pm2)-VS_{l}(\pm2)-S_{l-1}(\pm2)\big|\label{eq:=00005B0.0=00005D m-forward A - induction step}\\
= & \big|(\pm1)^{l+1}(l+2)-(\pm1)^{l}(l+1)V-(\pm1)^{l-1}l\big|\nonumber \\
= & \big|(\pm1)^{l+1}2\mp(\pm1)^{l+1}(l+1)V\big|\nonumber \\
= & \big|2\mp(l+1)V\big|.\nonumber 
\end{align}
Hence, we conclude from Proposition~\ref{Prop-traceMaps} (\ref{enu:.Prop-traceMaps - 1})
that 
\[
\big|t_{[0,0,m+1,1]}(-2,V)\big|=\big|t_{[0,0,m+2]}(-2,V)\big|=2+(m+2)V>2,\qquad V>0.
\]
This means that for all $V>0$, $E=-2$ is not a spectral edge of
any spectral band in $\sigma_{[0,0,m+1,1]}(V)$, see Lemma~\ref{lem: Traces and spectral edges}~(\ref{enu: prop-Traces and spectral edges - 1}).
Since $E=-2$ is a spectral band edge of $I_{[0,0]}(V)$, and since
$I^{j}_{[0,0,m+1,1]}(V)\strict I_{[0,0]}(V)$ for $1\leq j\leq M+1$
and $V>4$ (by Theorem~\ref{thm: V>4 Type A =000026 B}), we conclude
that $I^{j}_{[0,0,m+1,1]}(V)\strict I_{[0,0]}(V)$ may be violated
if and only if $R\left(I^{j}_{[0,0,m+1,1]}(V)\right)\geq2$. Since
$I^{j}_{[0,0,m+1,1]}\prec I^{M+1}_{[0,0,m+1,1]}$ holds by definition
for all $1\leq j<M+1$, it suffices to show that for all $V>0$,
\begin{equation}
R\left(I^{M+1}_{[0,0,m+1,1]}(V)\right)<2.\label{eq:=00005B0.0=00005D m-forward A. induction step - sufficient}
\end{equation}
For $V\geq\frac{4}{m+1}$, Equation (\ref{eq:=00005B0.0=00005D m-forward A - induction step})
and Proposition~\ref{Prop-traceMaps} (\ref{enu:.Prop-traceMaps - 1})
lead to
\[
\big|t_{[0,0,m+1,1]}(2,V)\big|=\big|t_{[0,0,m+2]}(2,V)\big|=\big|(m+2)V-2\big|\geq4\frac{m+2}{m+1}-2>2.
\]
Hence, (\ref{eq:=00005B0.0=00005D m-forward A. induction step - sufficient})
holds for all $V\geq\frac{4}{m+1}$, proving $I^{1}_{[0,0,m+1,1]}(V),I^{M+1}_{[0,0,m+1,1]}(V)\strict I_{[0,0]}(V)$
for all $V\geq\frac{4}{m+1}$. Recalling also the induction hypothesis,
$\crit([0,0,m])=0$, we apply Lemma~\ref{Lem-keyResult_A2} for $m+1\geq2$
and $\co=[0,0]$ and conclude that there exists $\delta>0$ such that
$I_{[0,0]}(V)$ is of $(m+1)$-type $A$ for $V>\tfrac{4}{m+1}-\delta$.
Since, $I_{[0,0]}$ is the only spectral band in $\sigma_{[0,0]}$
this implies $\crit([0,0,m+1])\leq\tfrac{4}{m+1}-\delta$. Thus, it
is left to prove (\ref{eq:=00005B0.0=00005D m-forward A. induction step - sufficient})
for $0<V\leq\frac{4}{m+1}-\delta$. Equation (\ref{eq:=00005B0.0=00005D m-forward A - induction step})
together with $0<V\leq\frac{4}{m+1}-\delta$ implies
\begin{align*}
\big|t_{[0,0,m+1]}(2,V)\big|=\big|2-(m+1)V\big|<2,
\end{align*}
so that there is a spectral band of $\sigma_{[0,0,m+1]}(V)$ which
contains $E=2$, for $0<V\leq\frac{4}{m+1}-\delta$. But, since $\sigma_{[0,0,m,1]}=\sigma_{[0,0,m+1]}$
(Proposition~\ref{prop: duality of A-B types}) and by the induction
hypothesis, $I^{m}_{[0,0,m,1]}(V)\strict I_{[0,0]}(V)$ for $V>0$,
the only spectral band which can contain $E=2$ is $K^{1}_{[0,0,m,1]}(V)=K_{[0,0,m+1]}(V)$,
and so
\[
L\left(K_{[0,0,m+1]}(V)\right)<R\left(I_{[0,0]}(V)\right)=2\qquad\textrm{for }0<V\leq\frac{4}{m+1}-\delta.
\]
In order to conclude (\ref{eq:=00005B0.0=00005D m-forward A. induction step - sufficient})
for $0<V\leq\frac{4}{m+1}-\delta$, we will apply Lemma~\ref{lem: Basics for eigenavlue inequalities}
for $[0,0],[0,0,m+1],[0,0,m+1,1]\in\Co$ and $\lo=R(I_{[0,0]}(V))$
and $\mo=R(I^{M+1}_{[0,0,m+1,1]}(V))$. A direct computation invoking
Lemma~\ref{lem: Floquet-Bloch matrices} and Lemma~\ref{lem: Floquet-Bloch matrices - n times fundamental domain}
yields
\begin{align*}
R(I_{[0,0]}(V)) & \in\sigma\left(H_{[0,0],V}(0)\right),~L(K_{[0,0,m+1]}(V))\in\sigma\left(H^{\times1}_{[0,0,m+1],V}(\pi)\right)
\end{align*}
and
\[
R(I^{M+1}_{[0,0,m+1,1]}(V))\in\sigma\left(H_{[0,0,m+1,1],V}(\pi)\right).
\]
Thus, $\theta_{[0,0]}=0$ , $\theta_{[0,0,m+1]}=\pi$, $\theta_{[0,0,m+1,1]}=\pi$
and these spectral edges are admissible (Definition~\ref{def: admissibility}).
Moreover, we can directly compute the values of the counting function
\begin{align*}
N_{[0,0]} & :=N(R(I_{[0,0]}(V));~H_{[0,0],V}(0))=0,\\
N_{[0,0,m+1,1]} & :=N(R(I^{M+1}_{[0,0,m+1,1]}(V);~H_{[0,0,m+1,1],V}(\pi))=m,
\end{align*}
and using $L\left(K_{[0,0,m+1]}(V)\right)<R\left(I_{[0,0]}(V)\right)$for
$0<V\leq\frac{4}{m+1}-\delta$, we get
\[
N_{[0,0,m+1]}:=N(R(I_{[0,0]}(V));~H^{\times1}_{[0,0,m+1],V}(\pi))=m+1.
\]
Hence, $N_{[0,0]}+N_{[0,0,m+1]}>N_{[0,0,m+1,1]}$ follows. Moreover,
$L\left(K_{[0,0,m+1]}(V)\right)<R\left(I_{[0,0]}(V)\right)=\lo$ for
$0<V\leq\frac{4}{m+1}-\delta$ implies that $\lo$ is a simple eigenvalue
in $H^{\times1}_{[0,0,m+1],V}(\pi)\oplus H_{[0,0],V}(0)$. Thus, Lemma
\ref{lem: Basics for eigenavlue inequalities}~(\ref{enu: lem-Basics for eigenavlue inequalities - 1})
yields 
\[
2=R(I_{[0,0]}(V))=\lo>R(I^{M+1}_{[0,0,m+1,1]}(V))=\mo.
\]
Hence, we have proven (\ref{eq:=00005B0.0=00005D m-forward A. induction step - sufficient})
for all $0<V\leq\frac{4}{m+1}-\delta$. Thus, Corollary~\ref{Cor-keyResult_A2}
(for $m\geq2$) implies $\crit([0,0,m+1])=0$ proving (f).
\end{proof}

Next we prove that the unique spectral band $I_{[0,0,1]}(V)=[-2+V,2+V]$
of $\sigma_{[0,0,1]}(V)$ is of type $B$ for all $V>0$. This is
demonstrated in Figure~\ref{fig: =00005B0.0.1=00005D_m-forward_B}.
\begin{lem}
\label{Lem-=00005B0.0.1=00005D_m-forward_B} Let $I_{[0,0,1]}(V)=[-2+V,2+V]$
be the unique spectral band of $\sigma_{[0,0,1]}(V)$ for $V\in\R$.
The following holds for all $V>0$:
\begin{enumerate}
\item \label{enu: Lem-=00005B0.0.1=00005D_m-forward_B - 1}$I_{[0,0,1]}(V)$
is of backward type $B$ but not of weak backward type $A$,
\item \label{enu: Lem-=00005B0.0.1=00005D_m-forward_B - 2}For all $m\in\N$,
$I_{[0,0,1]}(V)$ is of $m$-type $B$, i.e. $\crit([0,0,1,m])=0$.
\end{enumerate}
\end{lem}

\begin{proof}
Statement (\ref{enu: Lem-=00005B0.0.1=00005D_m-forward_B - 1}) follows
immediately from the equalities $\sigma_{[0,0,1,0]}(V)=\sigma_{[0,0]}(V)=[-2,2]$
and $\sigma_{[0,0,1,-1]}(V)=\sigma_{[0]}(V)=\R$ for all $V>0$. 
\begin{figure}[hbt]
\includegraphics[scale=1.2]{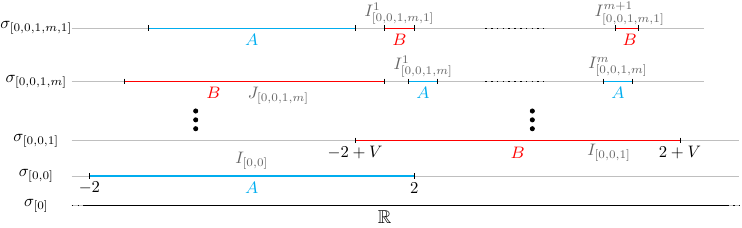} \caption{A sketch of the spectral bands considered in the proof of Lemma~\ref{Lem-=00005B0.0.1=00005D_m-forward_B}.}
\label{fig: =00005B0.0.1=00005D_m-forward_B}
\end{figure}

In order to prove (\ref{enu: Lem-=00005B0.0.1=00005D_m-forward_B - 2}),
let $m\in\N$ and $M=m$ since we aim to show that $I_{[0,0,1]}$
is of $m$-type $B$, see Definition~\ref{def: forward type}. This
will be done in two steps: first by applying Lemma~\ref{lem: keyResult}
to show $\crito([0,0,1,m])=0$ for all $m\in\N$; then applying Corollary~\ref{Cor-keyResult_A2}
to show $\crit([0,0,1,m])=0$ for all $m\in\N$. In order to apply
both corollaries we need to show that 
\begin{equation}
I^{1}_{[0,0,1,m,1]}(V),I^{M+1}_{[0,0,1,m,1]}(V)\strict I_{[0,0,1]}(V),\qquad V>0.\label{eq: =00005B0.0.1=00005D m-forward B: to show}
\end{equation}

\underline{Step 1:} We prove $\crito([0,0,1,m])=0$. Let $V\geq\frac{4}{m+1}$.
Lemma~\ref{Lem-Chebyshev-Traces}~(\ref{enu: Chebychev - |x|=00003D2})
asserts $S_{l}(\pm2)=(\pm1)^{l}(l+1)$. Thus, Lemma~\ref{lem: basic traces t_0.0}~(\ref{enu: lem-basic traces t_0.0 - 3})
leads to 
\begin{align*}
\big|t_{[0,0,1,m,1]}(\pm2+V,V)\big|= & \big|(\pm2+V)S_{m+1}(\pm2)-2S_{m}(\pm2)\big|\\
= & \big|(\pm2+V)(m+2)\mp2(m+1)\big|\\
= & \big|V(m+2)\pm2\big|.
\end{align*}
Thus, $\big|t_{[0,0,1,m,1]}(\pm2+V,V)\big|>2$ follows whenever $V\geq\frac{4}{m+1}$.
The values $\pm2+V$, at which the trace is evaluated, are the spectral
band edges of $I_{[0,0,1]}(V)$. By Theorem~\ref{thm: V>4 Type A =000026 B},
the strict inclusions (\ref{eq: =00005B0.0.1=00005D m-forward B: to show})
hold for $V>4.$ Thus, by continuity of the spectral edges (Corollary~\ref{prop: Lipschitz spectral edges})
and $\big|t_{[0,0,1,m,1]}(\pm2+V,V)\big|>2$, (\ref{eq: =00005B0.0.1=00005D m-forward B: to show})
holds for all $V\geq\frac{4}{m+1}$. Hence, Lemma~\ref{lem: keyResult}
implies that there is a $\delta>0$ such that $I_{[0,0,1]}$ satisfies
$\Tquas\left(m,\tfrac{4}{m+1}-\delta\right)$, namely $\crito([0,0,1,m])\leq\tfrac{4}{m+1}-\delta$.

Now let us consider the range $0<V\leq\tfrac{4}{m+1}-\delta$. First,
for all $V>0$, the trace computation above gives 
\[
\big|t_{[0,0,1,m,1]}(2+V,V)\big|=V(m+2)+2>2.
\]
Hence, $R(I^{M+1}_{[0,0,1,m,1]}(V))<2+V$ follows for all $V>0$,
since it holds for $V>4$ by Theorem~\ref{thm: V>4 Type A =000026 B},
and using continuity (Corollary~\ref{prop: Lipschitz spectral edges}).
Since $I^{1}_{[0,0,1,m,1]}(V)\prec I^{M+1}_{[0,0,1,m,1]}(V)$ holds
for all $V>0$, it thus suffices to prove
\begin{equation}
-2+V<L(I^{1}_{[0,0,1,m,1]}(V))\qquad\textrm{for all }0<V\leq\tfrac{4}{m+1}-\delta.\label{eq: =00005B0.0.1=00005D m-forward B: sufficient to show}
\end{equation}
Let $J_{[0,0,1,m]}$ be the spectral band associated with $I_{[0,0,1]}$
of Definition~\ref{def: left and right proceeded band}, see also
Figure~\ref{fig: =00005B0.0.1=00005D_m-forward_B}. Hence, it is
the spectral band satisfying $\ind(J_{[0,0,1,m]})=\ind(I^{1}_{[0,0,1,m]})-1$.
Lemma~\ref{lem: basic traces t_0.0}~(\ref{enu: lem-basic traces t_0.0 - 2}),
Lemma~\ref{Lem-Chebyshev-Traces}~(\ref{enu: Chebychev - |x|=00003D2})
and $0<V\leq\tfrac{4}{m+1}-\delta$ lead to 
\begin{align*}
\big|t_{[0,0,1,m]}(-2+V,V)\big|= & \big|(-2+V)S_{m}(-2)-2S_{m-1}(-2)\big|\\
= & \big|(-2+V)(-1)^{m}(m+1)-2(-1)^{m-1}m\big|\\
= & \big|(-2+V)(m+1)+2m\big|\\
= & \big|V(m+1)-2\big|<2.
\end{align*}
By Lemma~\ref{lem: ConnecSpectrTrace}, we conclude that $-2+V$
is in the interior of a spectral band of $\sigma_{[0,0,1,m]}(V)$.
On the other hand, $-2+V<L(I^{1}_{[0,0,1,m]}(V))$ holds for $V>\tfrac{4}{m+1}-\delta$.
Thus, the continuity of the spectral edges in $V$ implies $-2+V<R(J_{[0,0,1,m]})$
for $0\leq V\leq\tfrac{4}{m+1}-\delta$.

With this at hand, we prove (\ref{eq: =00005B0.0.1=00005D m-forward B: sufficient to show})
by applying Lemma~\ref{lem: Basics for eigenavlue inequalities}
with $\co:=[0,0,1]$, $[\co,m]=[0,0,1,m]$, $[\co,m,n]=[0,0,1,m,1]$,
$\lo=L(I_{[0,0,1]}(V))=-2+V$ and $\mo=L(I^{1}_{[0,0,1,m,1]}(V))$.
Let $\thc,\thcm,\theta_{[\co,m,1]}\in\{0,\pi\}$ be such that
\begin{align*}
L(I_{[0,0,1]}(V)) & \in\sigma\left(H_{[0,0,1],V}(\thc)\right),~R(J_{[0,0,1,m]}(V))\in\sigma\left(H^{\times1}_{[0,0,1,m],V}(\thcm)\right)
\end{align*}
and
\[
L(I^{1}_{[0,0,1,m,1]}(V))\in\sigma\left(H_{[0,0,1,m,1],V}(\theta_{[\co,m,1]})\right).
\]
A direct computation yields $\ind(I_{[0,0,1]})=0$, $\ind(J_{[0,0,1,m]})=0$
and $\ind(I^{1}_{[0,0,1,m,1]})=1$, see Figure~\ref{fig: =00005B0.0.1=00005D_m-forward_B}.
Inserting these indices into the characterization of admissibility
in Lemma~\ref{lem: Admissibility criterion}, we conclude that these
spectral edges are admissible. Furthermore, we can directly compute
the values of the counting function
\begin{align*}
\Nc & :=N(L(I_{[0,0,1]}(V));~H_{[0,0,1],V}(\thc))=0,\\
N_{[\co,m,1]} & :=N(L(I^{1}_{[0,0,1,m,1]}(V);~H_{[0,0,1,m,1],V}(\theta_{[0,0,1,m,1]}))=1,
\end{align*}
and using $L(I_{[0,0,1]}(V))=-2+V<R(J_{[0,0,1,m]})$ for $0<V\leq\frac{4}{m+1}-\delta$,
we get
\[
\Ncm:=N(L(I_{[0,0,1]}(V));~H^{\times1}_{[0,0,1,m],V}(\thcm))=0.
\]
Hence, $\Nc+\Ncm<N_{[\co,m,1]}$ follows. Moreover, $\lo=L(I_{[0,0,1]}(V))<R(J_{[0,0,1,m]}(V))$
for $0<V\leq\frac{4}{m+1}-\delta$ implies that $\lo$ is a simple
eigenvalue in $H^{\times1}_{[0,0,1,m],V}(\thcm)\oplus H_{[0,0,1],V}(\thc)$.
Thus, Lemma~\ref{lem: Basics for eigenavlue inequalities} yields
$-2+V=\lo<\mo=L(I^{1}_{[0,0,1,m,1]}(V))$ proving (\ref{eq: =00005B0.0.1=00005D m-forward B: sufficient to show})
for all $0<V\leq\frac{4}{m+1}-\delta$. Hence, $\crito([0,0,1,m])=0$
follows.

\underline{Step 2:} We prove $\crit([0,0,1,m])=0$. Let $m\geq2$.
We aim to apply Corollary~\ref{Cor-keyResult_A2} and need to check
its assumptions (\ref{enu: Cor-keyResult_A2 - 1}), (\ref{enu: Cor-keyResult_A2 - 2})
and (\ref{enu: Cor-keyResult_A2 - 3}). We have seen above that $\sigma_{[0,0,1]}(V)$
has exactly one spectral band $I_{[0,0,1]}$ which is of backward
type $B$ but not of weak backward type $A$, so that assumption (\ref{enu: Cor-keyResult_A2 - 1})
of Corollary~\ref{Cor-keyResult_A2} holds. By step 1 of the proof,
we have $\crito([0,0,1,m])=0$ and $\crito([0,0,1,m-1])=0$. The first
implies that assumption (\ref{enu: Cor-keyResult_A2 - 2}) in Corollary~\ref{Cor-keyResult_A2}
holds and the second implies that assumption (\ref{enu: Cor-keyResult_A2 - 3})
in Corollary~\ref{Cor-keyResult_A2} holds using $m\geq2$. Hence,
$\crit([0,0,1,m])=0$ follows for all $m\geq2$.

Let $m=1$. Since $\varphi([0,0,1])=1$, assumption (\ref{enu: Cor-keyResult_A2 - 3})
in Corollary~\ref{Cor-keyResult_A2} does not hold for $m=1$, and
we cannot apply that corollary. Instead, we directly verify that $\crit([0,0,1,1])=0$.
Recall $t_{[0,0,1,1]}=t_{[0,0,2]}=E^{2}-EV-2$, see Example~\ref{exa: first spectral bands}.
Thus, $\sigma_{[0,0,1,1]}(V)=[E_{0}(V),E_{1}(V)]\cup[E_{2}(V),E_{3}(V)]$
with 
\[
E_{0}(V):=\frac{V}{2}-\sqrt{\frac{V^{2}}{4}+4},\quad E_{1}(V):=0,\quad E_{2}(V):=V,\quad E_{3}(V):=\frac{V}{2}+\sqrt{\frac{V^{2}}{4}+4}.
\]
Thus, $I_{[0,0,1,1]}(V)=[E_{2}(V),E_{3}(V)]\strict I_{[0,0,1]}(V)$
and $I_{[0,0,1,1]}(V)$ is the unique spectral band in $\sigma_{[0,0,1,1]}(V)$
of backward type $A$. Since $2<E_{3}(V)$ for all $V>0$, we conclude
that $I_{[0,0,1,1]}(V)$ is not included in $I_{[0,0]}(V)$ and hence
it is not of weak backward type $B$ for all $V>0$. Thus, $I_{[0,0,1]}(V)$
satisfies \ref{enu: A2 property} for all $V>0$. Since $\crito([0,0,1,1])=0$
by step 1, we conclude $\crit([0,0,1,1])=0$. 
\end{proof}

\subsection*{Acknowledgments}

 We are grateful to David Damanik, Jake Fillman and Anton Gorodetski
for inspiring discussions on Sturmian systems. We thank Barak Biber
and Yannik Thomas for extensive discussions on the work of Laurent
Raymond helping us to get a deeper insights. We are thankful to Michael
Baake for organizing and hosting a joint meeting with Laurent Raymond
in Bielefeld on December 2023. We thank Laurent Raymond for his support
and our interesting discussions.

We thank the Israel Institute of Technology and the University of
Potsdam for providing excellent working conditions during our mutual
visits.  SB was partially supported by the Deutsche Forschungsgemeinschaft
{[}BE 6789/1-1 to SB{]} and the Maria-Weber Grant 2022 offered by
the Hans Böckler Stiftung. RB was supported by the Israel Science
Foundation (ISF Grants No. 844/19 and 2362/25). This article was finalized
at the Israel Institute for Advanced Studies, as part of the Research
Group Analysis, Geometry, and Spectral Theory of Graphs (2025). RB
and SB are grateful to the IIAS for the excellent working conditions.

\bibliographystyle{amsalpha}
\bibliography{references}

\appendix
\appendix
\renewcommand{\thesection}{\Roman{section}}

\section{Sturmian dynamical systems\label{App: Sturmian dynamical systems}}

This appendix contains a very short description of Sturmian dynamical
systems. A thorough background may be found in the books \cite{Fogg_book02,Loth02,DaFi24-book_2}.
In addition, we state a lemma summarizing some basic properties of
Sturmian sequences and mechanical words, which are applied in this
paper.

To define the Sturmian Hamiltonian we have defined the sequences
\[
\omega_{\alpha}(n):=\chi_{[1-\alpha,1[}(n\alpha\mod 1),\quad n\in\N,~\alpha\in[0,1],
\]
which are called \emph{mechanical words} \cite[Sec.~2.1.2]{Loth02}
. If $\alpha\not\in\Q$, $\omega_{\alpha}$ is also called a \emph{Sturmian
sequence}. These sequences naturally define a dynamical system as
follows. Let $\Aa:=\{0,1\}$ be equipped with the discrete topology
and $\Aa^{\Z}:=\{\omega:\Z\to\Aa\}$ be the compact metrizable space
equipped with the product topology. Consider the shift $T:\Aa^{\Z}\to\Aa^{\Z},~(T\omega)(n):=\omega(n-1),~n\in\Z,$
being a homeomorphism. This induces a continuous group action $\Z\curvearrowright\Aa^{\Z}$
via $(n,\omega)\mapsto T^{n}\omega$. For $\alpha\in[0,1]$, we have
$\omega_{\alpha}\in\{0,1\}^{\Z}$ and its associated orbit closure
(in the product topology)
\[
\Omega_{\alpha}:=\overline{\textrm{Orb}(\omega_{\alpha})}:=\overline{\set{T^{n}\omega_{\alpha}}{n\in\Z}}
\]
defines a dynamical system $\Z\curvearrowright\Omega_{\alpha}$. If
$\alpha\text{\ensuremath{\in\Q}}$, then $\omega_{\alpha}$ is periodic,
i.e. there is a period $q\in\N$ such that $T^{q}\omega_{\alpha}=\omega_{\alpha}$.
Note that in this case $\textrm{Orb}(\omega_{\alpha})=\Omega_{\alpha}$.
There are various different representations of this dynamical system.
For instance, the authors in \cite[Lem.~1]{BIST89} proved that 
\[
\omega_{\alpha}(n)=\fl{\alpha(n+1)}-\fl{\alpha\thinspace n},\qquad n\in\N,~\alpha\in[0,1]\setminus\Q.
\]
A different approach to describe these words is via a recursive rule
using the continued fraction expansion $(0,c_{1},c_{2},\ldots)$ of
$\alpha\in[0,1]\setminus\Q$ is described in \cite[Eq.~.9]{Loth02}.
The reader is also referred to \cite{BaBeBiTh22} for a more detailed
discussion.  The following lemma provides the properties of mechanical
words which are useful in our paper.
\begin{lem}
\label{lem: periods of three approximants} Let $\co=\left[0,0,c_{1},c_{2},\ldots,c_{k}\right]\in\Co$
for $k\in\N$ and $\frac{p_{k}}{q_{k}}:=\varphi(\co)$, with co-prime
$p_{k},q_{k}$. Then the following holds.
\end{lem}

\begin{enumerate}
\item \label{enu: lem-periods of three approximants-1}The sequence $\ohn$
is periodic with period length $q_{k}$. Let its period $W_{k}\in\left\{ 0,1\right\} ^{q_{k}}$
be defined by
\[
W_{k}(i):=\omega_{\alpha_{k}}(i),\qquad0\leq i\leq q_{k}-1
\]
\item \label{enu: lem-periods of three approximants-2}For $k\in\N$, we
have $q_{k}=c_{k}\cdot q_{k-1}+q_{k-2}$ with $q_{-1}=0$ and $q_{0}=1$.
\item \label{enu: lem-periods of three approximants-3}The period of $\ohn$
satisfy the following $W_{0}=0,W_{1}=0^{c_{1}-1}1$ and if $k\geq2$,
then
\[
W_{k}=\begin{cases}
W_{k-2}W^{c_{k}}_{k-1},\quad & k\equiv0\mod 2,\\
W^{c_{k}}_{k-1}W_{k-2},\quad & k\equiv1\mod 2.
\end{cases}
\]
\item \label{enu: lem-periods of three approximants-4}For $k\geq1$,
\begin{itemize}
\item If $k\equiv0\mod 2$ then $\omega_{\alpha}(i)=W_{k}(i)$ for all $0\leq i\leq q_{k}-1$.
\item If $k\equiv1\mod 2$ then $\omega_{\alpha}(i)=W_{k}(i)$ for all $0\leq i\leq q_{k}-2$.
\end{itemize}
\end{enumerate}
\begin{proof}
The first two parts of the lemma are basic. The third and fourth parts
appears e.g. in \cite[Lem.~2.4]{BaBeBiTh22}.
\end{proof}

\section{Chebyshev polynomials and trace identities\label{App: Trace Maps}}

In this section, we provide several known identities of traces and
their connection to Chebyshev polynomials, see e.g. \cite{Casdagli1986,Raym95,BIST89,BaBeBiTh22,Raym-AperiodicOrder,DaFi24-book_2}.
Moreover, we prove Lemma~\ref{lem: Trace estimates =00005Bc.m.n=00005D}.

\subsection{Chebyshev polynomials}

A crucial tool for studying the spectral theory of Sturmian Hamiltonians
are the dilated Chebyshev polynomials of the second kind (see \cite[\href{https://dlmf.nist.gov/18.1.E3}{(18.1.3)}]{DigLibMathFunc})
$S_{n}:\R\to\R,\,n\in\N_{-1},$ defined by 
\[
S_{-1}(x):=0,\quad S_{0}(x):=1,\quad S_{n}(x):=xS_{n-1}(x)-S_{n-2}(x)\quad\textrm{for }x\in\R.
\]
For $x\in\mathbb{R}\setminus\left\{ 0\right\} $, denote by $\sgn(x)\in\{+1,-1\}$
the sign of $x$.
\begin{lem}
\label{Lem-Chebyshev-Traces} Let $x\in\mathbb{R}$ and $n\in\N_{0}$.
Then the following holds.
\begin{enumerate}
\item \label{enu: Chebychev - invariant}We have $S_{n+1}S_{n-1}-S^{2}_{n}=-1$.
\item \label{enu: Chebychev - |x|=00003D2}If $|x|=2$, then $\sgn(x)^{n-1}S_{n-1}(x)=n$.
\item \label{enu: Chebychev - |x|=00005Cgeq 2}If $|x|\geq2$, then $\sgn(x)^{n}S_{n}(x)=|S_{n}(x)|$
and
\[
\sgn(x)^{n}xS_{n-1}(x)\geq2\big|S_{n-1}(x)\big|.
\]
\item \label{enu: Chebychev - |x|=00005Cgeq 2 - term =00005Cgeq 1}If $|x|\geq2$,
then 
\[
\sgn(x)^{n}\big(S_{n}(x)-\frac{x}{2}S_{n-1}(x)\big)\geq1.
\]
\item \label{enu: Chebychev - |x|> 2 - term > 1}If $|x|>2$ and $n\geq1$,
then 
\[
\sgn(x)^{n}\big(S_{n}(x)-\frac{x}{2}S_{n-1}(x)\big)>1.
\]
\end{enumerate}
\end{lem}

\begin{proof}
The proof follows by induction using the recursive relation, the details
can be found for instance in \cite[Lem.~III.2.]{BaBeBiTh22}.
\end{proof}

\subsection{Trace identities \label{Sec-TraceMaps}}

This section is devoted to various trace identities and the proof
of Lemma~\ref{lem: Trace estimates =00005Bc.m.n=00005D}.

The following proposition is a collection of well-known identities
of the traces, see e.g. \cite{Casdagli1986,Raym95,BIST89,DaFi22-book_1,BaBeBiTh22,Raym-AperiodicOrder,DaFi24-book_2}.
Recall that for $\co\in\Co$, $\tc$ is a function of $E,V\in\R$,
but we abbreviate notation and suppress this dependencies in the following.
\begin{prop}
[trace maps] \label{Prop-traceMaps} Let $m\in\N$, $n\in\N_{0}$
and $\co\in\Co$ be such that $[\co,m]\in\Co$. Then the following
holds.
\begin{enumerate}
\item \label{enu:.Prop-traceMaps - 1}We have $t_{[\co,m,0]}=t_{\co}$,
$t_{[\co,m,1]}=t_{[\co,m+1]}$ and $t_{[\co,m,-1]}=t_{[\co,m-1]}$.
\item \label{enu:.Prop-traceMaps - 2}We have for all $V\in\R$, \textup{(the
Fricke--Vogt invariant)}
\[
V^{2}+4=t^{2}_{[\co,n+1]}+t^{2}_{[\co,n]}+t^{2}_{\co}-t_{[\co,n+1]}t_{[\co,n]}t_{\co}
\]
\item \label{enu:.Prop-traceMaps - 3}For $-1\leq l\leq n$, we have
\end{enumerate}
\[
t_{[\co,n+1]}=S_{l+1}(t_{\co})t_{[\co,n-l]}-S_{l}(t_{\co})t_{[\co,n-l-1]}.
\]
In particular, we have $t_{[\co,n+1]}=t_{\co}t_{[\co,n]}-t_{[\co,n-1]}$
(for $l=0$).
\end{prop}

We will continue with two auxiliary lemmata which are needed to prove
Lemma~\ref{lem: Trace estimates =00005Bc.m.n=00005D}. In order to
treat certain cases of the backward type $A$ ($\ell=0$) bands or
backward type $B$ ($\ell=-1$) bands, we need the following identity.
\begin{lem}
\label{Lem-TraceMap_Identity} Let $m\in\N$ and $\co\in\Co$ be such
that $[\co,m]\in\Co$. Let $V\in\mathbb{R}$ and $E\in\mathbb{R}$
be such that $|t_{\co}(E,V)|=2$. Then for all $n\in\N$ and $\ell\in\left\{ -1,0\right\} $,
\begin{align*}
t_{[\co,\ell]}(E,V)S_{n}\big(t_{[\co,m]}(E,V)\big)= & z^{m-\ell}\Big(2(m-\ell)S_{n+1}\big(t_{[\co,m]}(E,V)\big)-z(m-\ell)t_{[\co,m,n]}(E,V)\\
 & -(m-1-\ell)t_{[\co,m]}(E,V)S_{n}\big(t_{[\co,m]}(E,V)\big)\Big)
\end{align*}
holds where $z:=\sgn(t_{\co}(E,V))$.
\end{lem}

\begin{rem*}
This lemma is closely related to \cite[Prop.~2]{BIST89}. 
\end{rem*}
\begin{proof}
For the sake of simplicity, we abbreviate the notation in the following
and write $t_{\cop}=t_{\cop}(E,V)$ for $\cop\in\Co$. As a direct
consequence of Lemma~\ref{Lem-Chebyshev-Traces}~(\ref{enu: Chebychev - |x|=00003D2}),
we conclude that $S_{l}(t_{\co})\neq0$ for all $l\geq0$ since $|t_{\co}|=2$.
Proposition~\ref{Prop-traceMaps}~(\ref{enu:.Prop-traceMaps - 3})
(applied for $n=m-1$ and $l=m-2-\ell\geq-1$) leads to 
\[
t_{[\co,m]}=S_{m-1-\ell}\big(t_{\co}\big)t_{[\co,1+\ell]}-S_{m-2-\ell}\big(t_{\co}\big)t_{[\co,\ell]}.
\]
Since $S_{m-1-\ell}(t_{\co})\neq0$, we derive 
\begin{equation}
t_{[\co,1+\ell]}=\frac{1}{S_{m-1-\ell}\big(t_{\co}\big)}\Big(t_{[\co,m]}+S_{m-2-\ell}\big(t_{\co}\big)t_{[\co,\ell]}\Big).\label{Eq-TraceMap_Identity_1}
\end{equation}
Let $n\in\N$. Using again Proposition~\ref{Prop-traceMaps}~(\ref{enu:.Prop-traceMaps - 1})
and (\ref{enu:.Prop-traceMaps - 3}) (with $l=n$), we conclude 
\begin{align}
t_{[\co,m,n]}=S_{n+1}\big(t_{[\co,m]}\big)t_{[\co,m,0]}-S_{n}\big(t_{[\co,m]}\big)t_{[\co,m,-1]}=S_{n+1}\big(t_{[\co,m]}\big)t_{\co}-S_{n}\big(t_{[\co,m]}\big)t_{[\co,m-1]}.\label{Eq-TraceMap_Identity_2}
\end{align}
The case $\ell=0$ and $m=1$ need to be treated separately. We treat
this case later and first assume that if $\ell=0$ then $m\geq2$.
Then Proposition~\ref{Prop-traceMaps}~(\ref{enu:.Prop-traceMaps - 3})
(applied for $n=m-2$ and $l=m-3-\ell\geq-1$) and Equation~\eqref{Eq-TraceMap_Identity_1}
leads to 
\begin{align*}
t_{[\co,m-1]}= & S_{m-2-\ell}\big(t_{\co}\big)t_{[\co,1+\ell]}-S_{m-3-\ell}\big(t_{\co}\big)t_{[\co,\ell]}\\
= & \frac{S_{m-2-\ell}\big(t_{\co}\big)}{S_{m-1-\ell}\big(t_{\co}\big)}t_{[\co,m]}+t_{[\co,\ell]}\left(\frac{S_{m-2-\ell}\big(t_{\co}\big)^{2}}{S_{m-1-\ell}\big(t_{\co}\big)}-S_{m-3-\ell}\big(t_{\co}\big)\right)\\
= & \frac{S_{m-2-\ell}\big(t_{\co}\big)}{S_{m-1-\ell}\big(t_{\co}\big)}t_{[\co,m]}+t_{[\co,\ell]}\left(\frac{S_{m-2-\ell}\big(t_{\co}\big)^{2}-S_{m-3-\ell}\big(t_{\co}\big)S_{m-1-\ell}\big(t_{\co}\big)}{S_{m-1-\ell}\big(t_{\co}\big)}\right).
\end{align*}
Since $S_{k}S_{k-2}-S^{2}_{k-1}=-1$ for $k=m-1-\ell$ by Lemma~\ref{Lem-Chebyshev-Traces}~(\ref{enu: Chebychev - invariant}),
we conclude 
\begin{align*}
t_{[\co,m-1)}= & \frac{S_{m-2-\ell}\big(t_{\co}\big)}{S_{m-1-\ell}\big(t_{\co}\big)}t_{[\co,m]}+t_{[\co,\ell]}\frac{1}{S_{m-1-\ell}\big(t_{\co}\big)}.
\end{align*}
Inserting the latter into Equation~\eqref{Eq-TraceMap_Identity_2},
we get 
\begin{equation}
t_{[\co,m,n]}=S_{n+1}\big(t_{[\co,m]}\big)t_{\co}-S_{n}\big(t_{[\co,m]}\big)\frac{S_{m-2-\ell}\big(t_{\co}\big)}{S_{m-1-\ell}\big(t_{\co}\big)}t_{[\co,m]}-t_{[\co,\ell]}\frac{S_{n}\big(t_{[\co,m]}\big)}{S_{m-1-\ell}\big(t_{\co}\big)}.\label{Eq-TraceMap_Identity_3}
\end{equation}
We claim that the latter identity holds also if $\ell=0$ and $m=1$.
Indeed, this follows immediately from Equation~\eqref{Eq-TraceMap_Identity_2},
$t_{[\co,m-1)}=t_{[\co,\ell]}$, $S_{m-1-\ell}\big(t_{\co}\big)=1$
and $S_{m-2-\ell}\big(t_{\co}\big)=0$.

Now we can proceed with arbitrary $\ell\in\left\{ -1,0\right\} $
and $m\in\N$. Reorganizing Equation~\eqref{Eq-TraceMap_Identity_3}
leads to 
\begin{align*}
t_{[\co,\ell]}S_{n}\big(t_{[\co,m]}\big)= & S_{m-1-\ell}\big(t_{\co}\big)\Big(S_{n+1}\big(t_{[\co,m]}\big)t_{\co}-t_{[\co,m,n]}\Big)-S_{m-2-\ell}\big(t_{\co}\big)S_{n}\big(t_{[\co,m]}\big)t_{[\co,m]}.
\end{align*}
Since we assumed that $t_{\co}=t_{\co}(E,V)=2z$, Lemma~\ref{Lem-Chebyshev-Traces}~(\ref{enu: Chebychev - |x|=00003D2})
implies $S_{n}(t_{\co})=z^{n}(n+1)$ for $n\geq0$. Thus, 
\begin{align*}
t_{[\co,\ell]}S_{n}\big(t_{[\co,m]}\big)= & z^{m-1-\ell}\big(m-\ell\big)\Big(S_{n+1}\big(t_{[\co,m]}\big)2z-t_{[\co,m,n]}\Big)\\
 & -z^{m-2-\ell}\big(m-1-\ell\big)S_{n}\big(t_{[\co,m]}\big)t_{[\co,m]}
\end{align*}
follows proving the desired identity.
\end{proof}

\begin{lem}
\label{Lem-TraceMap_Est} Let $m,n\in\N$ and $\co\in\Co$ be such
that $[\co,m]\in\Co$. Let $V\in\R$ and $E\in\R$ be such that 
\[
|t_{\co}(E,V)|=2\qquad\text{and}\qquad|t_{[\co,m]}(E,V)|\geq2.
\]
Then for all $n\in\N$ and $\ell\in\left\{ -1,0\right\} $, 
\[
\left|t_{[\co,\ell]}(E,V)S_{n}\big(t_{[\co,m]}(E,V)\big)\right|\geq(m-\ell)\big(2-\big|t_{[\co,m,n]}(E,V)\big|\big)+2\big|S_{n}\big(t_{[\co,m]}(E,V)\big)\big|
\]
holds and the estimate is strict if additionally $|t_{[\co,m]}(E,V)|>2$.
\end{lem}

\begin{rem*}
The latter estimate is the general formula that we need to treat backward
type $A$ bands ($\ell=0$) or backward type $B$ bands ($\ell=-1$).
\end{rem*}
\begin{proof}
In order to simplify notation, set $t:=t_{[\co,m]}(E,V)$, $z_{1}:=\sgn(t)$
and $z_{0}:=\sgn(t_{\co}(E,V))$. Furthermore, we abbreviate the notation
and write $t_{\cop}=t_{\cop}(E,V)$ for $\cop\in\Co$. Due to Lemma~\ref{Lem-TraceMap_Identity}
and $z^{2n}_{1}=1$, we have 
\begin{align*}
&t_{[\co,\ell]}S_{n}\big(t\big)\\
	= &z^{m-\ell}_{0}z^{n+1}_{1}\Big(2(m-\ell)z^{n+1}_{1}S_{n+1}\big(t\big)-(m-1-\ell)z^{n+1}_{1}tS_{n}\big(t\big)-z_{0}z^{n+1}_{1}(m-\ell)t_{[\co,m,n]}\Big).
\end{align*}
Hence, 
\begin{align*}
\big|t_{[\co,\ell]}S_{n}\big(t\big)\big|\geq & \Big|2(m-\ell)z^{n+1}_{1}S_{n+1}\big(t\big)-(m-1-\ell)z^{n+1}_{1}tS_{n}\big(t\big)\Big|-\big|(m-\ell)t_{[\co,m,n]}\big|\\
= & 2(m-\ell)z^{n+1}_{1}\Big(S_{n+1}\big(t\big)-\frac{t}{2}S_{n}\big(t\big)\Big)+z^{n+1}_{1}tS_{n}\big(t\big)-(m-\ell)\big|t_{[\co,m,n]}\big|\\
\geq & 2(m-\ell)+2\big|S_{n}\big(t\big)\big|-(m-\ell)\big|t_{[\co,m,n]}\big|
\end{align*}
follows by first using the triangle inequality, secondly Lemma~\ref{Lem-Chebyshev-Traces}~(\ref{enu: Chebychev - |x|=00005Cgeq 2})
and (\ref{enu: Chebychev - |x|=00005Cgeq 2 - term =00005Cgeq 1})
since $|t|\geq2$ and finally Lemma~\ref{Lem-Chebyshev-Traces}~(\ref{enu: Chebychev - |x|=00005Cgeq 2})
and (\ref{enu: Chebychev - |x|=00005Cgeq 2 - term =00005Cgeq 1}).
Note that the last estimate is strict by Lemma~\ref{Lem-Chebyshev-Traces}~(\ref{enu: Chebychev - |x|> 2 - term > 1})
if $|t|>2$. This leads to the desired estimate.
\end{proof}

Now we can prove Lemma~\ref{lem: Trace estimates =00005Bc.m.n=00005D}.
\begin{proof}[Proof of Lemma~\ref{lem: Trace estimates =00005Bc.m.n=00005D}]
 Recall the assumptions of the proposition. Let $V\in\R$, $m,n\in\N$,
$\co\in\Co$. Let $I(V)$ be a spectral band of $\sigc(V)$ of backward
type $A$ or $B$. Set 
\[
\ell:=\begin{cases}
0,\qquad & I(V)\text{ is of backward type }A,\\
-1,\qquad & I(V)\text{ is of backward type }B.
\end{cases}
\]
Let $E\in\left\{ L(I(V)),~R(I(V))\right\} $. Then $|t_{[\co,\ell]}(E,V)|\leq2$
follows from Lemma~\ref{lem: Traces and spectral edges} and the
estimate is strict if $\varphi(\co)\in(0,1)$.

(\ref{enu: prop-Trace estimates =00005Bc.m.n=00005D - 1}) If $|\tcm(E,V)|\geq2$,
then Lemma~\ref{Lem-TraceMap_Est} and $m-\ell\geq1$ imply 
\begin{align*}
2\left|S_{n}\big(t_{[\co,m]}\big)\right|\geq & \left|t_{[\co,\ell]}S_{n}\big(t_{[\co,m]}\big)\right|\geq(m-\ell)\big(2-\big|t_{[\co,m,n]}\big|\big)+2\big|S_{n}\big(t_{[\co,m]}\big)\big|
\end{align*}
and so $|t_{[\co,m,n]}|\geq2$ must hold.

(\ref{enu: prop-Trace estimates =00005Bc.m.n=00005D - 2}) If $|\tcm(E,V)|>2$,
then Lemma~\ref{Lem-TraceMap_Est} and $m-\ell\geq1$ imply 
\begin{align*}
2\left|S_{n}\big(t_{[\co,m]}\big)\right|\geq & \left|t_{[\co,\ell]}S_{n}\big(t_{[\co,m]}\big)\right|>(m-\ell)\big(2-\big|t_{[\co,m,n]}\big|\big)+2\big|S_{n}\big(t_{[\co,m]}\big)\big|
\end{align*}
and so $|t_{[\co,m,n]}|>2$ must hold.

(\ref{enu: prop-Trace estimates =00005Bc.m.n=00005D - 3}) Since $\varphi(\co)\in(0,1)$,
Lemma~\ref{lem: Traces and spectral edges} asserts $|t_{[\co,\ell]}(E,V)|<2$.
If $|\tcm(E,V)|\geq2$, then Lemma~\ref{Lem-TraceMap_Est} and $m-\ell\geq1$
imply 
\begin{align*}
2\left|S_{n}\big(t_{[\co,m]}\big)\right|> & \left|t_{[\co,\ell]}S_{n}\big(t_{[\co,m]}\big)\right|\geq(m-\ell)\big(2-\big|t_{[\co,m,n]}\big|\big)+2\big|S_{n}\big(t_{[\co,m]}\big)\big|.
\end{align*}
Thus, $|t_{[\co,m,n]}|>2$ must hold.
\end{proof}

\section{A perturbation argument for eigenvalue interlacing\label{App: Perturbation argument}}

This section is devoted to the proof of Theorem~\ref{thm: perturbation thm for our matrices}.
Given an $n\times n$ hermitian matrix $X$, we denote its eigenvalues
in non-decreasing order by
\[
\lambda_{0}(X)\leq\lambda_{1}(X)\cdots\leq\lambda_{n-2}(X)\leq\lambda_{n-1}(X).
\]
We first recall a well-known result on interlacing of eigenvalues
using classical Weyl inequalities, see e.g. \cite[Cor.~4.3.3, Thm.~4.3.6]{HJ13}.

\begin{prop}
\label{prop: Rank1_Pert} Let $X$ and $Q$ be $n\times n$ hermitian
matrices, and suppose that $Q$ is a positive semidefinite, rank one
matrix.
\begin{enumerate}
\item \label{enu: Rank1_pert - (a)}For $j=1,\ldots n-1$, we have
\[
\lambda_{j-1}(X+Q)\leq\lambda_{j}(X)\leq\lambda_{j}(X+Q).
\]
\item \label{enu: Rank1_pert - (b)}For $j=0,\ldots n-2$, we have 
\[
\lambda_{j}(X-Q)\leq\lambda_{j}(X)\leq\lambda_{j+1}(X-Q).
\]
\end{enumerate}
\end{prop}

From these inequalities we can directly derive the following estimates
for traceless rank two perturbations.
\begin{cor}
\label{cor: Rank2_Pert} Let $X$ and $Y$ be $n\times n$ hermitian
matrices, and let $Q=Y-X$. If $Q$ has rank two and trace zero, then
\[
\lambda_{j-1}(Y)\leq\lambda_{j}(X)\leq\lambda_{j+1}(Y),\qquad j=1,2,\ldots n-2.
\]
\end{cor}

\begin{proof}
Using matrix diagonalization one may verify that there exist $n\times n$
hermitian, positive semidefinite matrices $Q_{1}$, $Q_{2}$ of rank
one such that $Q=Q_{1}-Q_{2}$. Applying first Proposition~\ref{prop: Rank1_Pert}~(\ref{enu: Rank1_pert - (a)})
to $X$ and $X+Q_{1}$, and then Proposition~\ref{prop: Rank1_Pert}~(\ref{enu: Rank1_pert - (b)})
to $X+Q_{1}$ and $Y=(X+Q_{1})-Q_{2}$ yields the desired inequalities.
\end{proof}

Let $\co=[0,c_{0},c_{1},\ldots,c_{k}]\in\Co$ be such that $k\in\N_{0}$
and $c_{k}\in\N$ if $k\geq1$. Note that this implies $\varphi(\co)\not\in\left\{ -1,\infty\right\} $.
Recall the notations of the matrices $\HcV(\theta)$ and $\nHcV(\theta)$
for $\theta\in[0,\pi]$ as introduced in Section~\ref{subsec: Perturbation argument}.
We aim to apply Corollary~\ref{cor: Rank2_Pert} to the matrices
\[
H_{[\co,m,n],V}\left(\theta_{[\co,m,n]}\right)\textrm{ and }\nHcmV\left(\theta_{[\co,m]}\right)\oplus H_{\co,V}(\theta_{\co})
\]
with appropriate choice of $\thc,\thcm,\thcmn\in\left\{ 0,\pi\right\} $.
It turns out that if $\thc,\thcm,\thcmn\in\left\{ 0,\pi\right\} $
are admissible (see Definition~\ref{def: admissibility}), then these
matrices are a rank two perturbation with trace zero of each other.
To formalize this statement, we define the matrix
\[
\Hcmn^{\oplus}(\theta_{[\co,m]},\thc):=\begin{cases}
\begin{pmatrix}H_{\co,V}(\theta_{\co}) & 0\\
0 & \nHcmV\left(\theta_{[\co,m]}\right)
\end{pmatrix} & \textrm{if \ensuremath{k\equiv0\mod 2},}\\
 & \textrm{}\\
\begin{pmatrix}\nHcmV\left(\theta_{[\co,m]}\right) & 0\\
0 & H_{\co,V}(\theta_{\co})
\end{pmatrix} & \textrm{if \ensuremath{k\equiv1\mod 2}.}
\end{cases}
\]

In the following statements we refer to vectors $x\in\R^{q}$ as column
vectors and use the notation $x^{t}$ to indicate the transpose of
a vector (which is then a row vector). We also use the notation $\langle x,y\rangle$
to denote the (Euclidean) inner product between vectors.
\begin{lem}
\label{lem: Pert_Hamil} Let $V\in\R$, $m,n\in\N$ and $\co=[0,c_{0},c_{1},\ldots,c_{k}]\in\Co$
be such that $k\in\N_{0}$ and $c_{k}\in\N$ if $k\geq1$ be such
that $[\co,m]\in\Co$. Let $\frac{p_{1}}{q_{1}}=\varphi(\co)$, $\frac{p_{2}}{q_{2}}=\varphi([\co,m])$
and $\frac{p_{3}}{q_{3}}=\varphi([\co,m,n])$ be such that $p_{i},q_{i}$
are coprime. If \textup{$\thc,\thcm,\thcmn\in\left\{ 0,\pi\right\} $
are admissible}, then there are $x:=x(\thc,\thcm,\thcmn),y:=y(\thc,\thcm,\thcmn)\in\R^{q_{3}}$
such that
\[
H_{[\co,m,n],V}\left(\theta_{[\co,m,n]}\right)-\Hcmn^{\oplus}(\theta_{[\co,m]},\thc)=xx^{t}-yy^{t}
\]
is a symmetric rank two perturbation with trace zero. Furthermore,
set 
\[
d_{1}:=\begin{cases}
q_{1} & \textrm{if }\ensuremath{k\equiv0\mod 2},\\
nq_{2} & \textrm{if }\ensuremath{k\equiv1\mod 2},
\end{cases}\qquad\textrm{and}\qquad d_{2}:=\begin{cases}
nq_{2} & \textrm{if }\ensuremath{k\equiv0\mod 2},\\
q_{1} & \textrm{if \ensuremath{k\equiv1\mod 2}}.
\end{cases}
\]
\begin{enumerate}
\item \label{enu: Pert_Hamil - n even - first components}If $w=(w_{1},\ldots w_{d_{1}},0,\ldots,0)^{t}\in\R^{q_{3}}$
is orthogonal to $x$ and to $y$ then $w_{1}=w_{d_{1}}=0$.
\item \label{enu: Pert_Hamil - n even - second components}If $w=(0,\ldots,0,w_{1},\ldots w_{d_{2}})^{t}\in\R^{q_{3}}$
is orthogonal to $x$ and to $y$ then $w_{1}=w_{d_{2}}=0$.
\end{enumerate}
\end{lem}

\begin{proof}
Let $e_{1},\ldots,e_{q_{3}}$ be the standard orthonormal basis of
$\R^{q_{3}}$, namely $e_{i}$ is the $i$-th unit vector in $\R^{q_{3}}$.
Recall that for the continued fraction expansion, we have the identity
$q_{3}=nq_{2}+q_{1}=d_{1}+d_{2}$, see e.g. Lemma~\ref{lem: periods of three approximants}.
The statement of the lemma follows by straightforward calculations
invoking Lemma~\ref{lem: periods of three approximants}, so we just
explicitly write here the expressions for the $x,y\in\R^{q_{3}}$
in the statement of the lemma. Let $k\in\N$ be even. If $(\thc,\thcm,\thcmn)=(0,0,0)$,
then
\begin{align*}
x=\frac{1}{\sqrt{2}}\big(e_{1}-e_{nq_{2}}-e_{nq_{2}+1}+e_{q_{3}}\big),\quad & y=\frac{1}{\sqrt{2}}\big(-e_{1}-e_{nq_{2}}+e_{nq_{2}+1}+e_{q_{3}}\big).
\end{align*}
If $(\thc,\thcm,\thcmn)=(\pi,\pi,0)$, then 
\begin{align*}
x=\frac{1}{\sqrt{2}}\big(e_{1}+e_{nq_{2}}+e_{nq_{2}+1}+e_{q_{3}}\big),\quad & y=\frac{1}{\sqrt{2}}\big(-e_{1}+e_{nq_{2}}-e_{nq_{2}+1}+e_{q_{3}}\big).
\end{align*}
If $(\thc,\thcm,\thcmn)=(\pi,0,\pi)$, then 
\begin{align*}
x=\frac{1}{\sqrt{2}}\big(-e_{1}+e_{nq_{2}}+e_{nq_{2}+1}+e_{q_{3}}\big),\quad & y=\frac{1}{\sqrt{2}}\big(e_{1}+e_{nq_{2}}-e_{nq_{2}+1}+e_{q_{3}}\big).
\end{align*}
If $(\thc,\thcm,\thcmn)=(0,\pi,\pi)$, then 
\begin{align*}
x=\frac{1}{\sqrt{2}}\big(e_{1}+e_{nq_{2}}+e_{nq_{2}+1}-e_{q_{3}}\big),\quad & y=\frac{1}{\sqrt{2}}\big(-e_{1}+e_{nq_{2}}-e_{nq_{2}+1}-e_{q_{3}}\big).
\end{align*}
The case when $n\in\N$ is odd is treated similarly.
\end{proof}

\begin{rem}
Statements (\ref{enu: Pert_Hamil - n even - first components}) and
(\ref{enu: Pert_Hamil - n even - second components}) in Lemma~\ref{lem: Pert_Hamil}
are used to conclude strict inequalities in Theorem~\ref{thm: perturbation thm for our matrices}.
Towards this, we use that if two consecutive (up to cyclic permutation)
entries of a solution (see (\ref{enu: Pert_Hamil - n even - first components})
and (\ref{enu: Pert_Hamil - n even - second components})) vanish,
then the whole solution vanishes, since we have a nearest neighbor
interaction.
\end{rem}

Now we have all tools at hand to prove Theorem~\ref{thm: perturbation thm for our matrices}.
\begin{proof}[Proof of Theorem~\ref{thm: perturbation thm for our matrices}]
 Recall the statement of the theorem. Let $V>0$, $m,n\in\N$ and
$\co=[0,c_{0},\ldots,c_{k}]\in\Co$ be such that $\varphi(\co)\not\in\left\{ -1,\infty\right\} $
and $[\co,m]\in\Co$. Thus, $c_{k}\in\N$ if $k\geq1$. Furthermore,
$\thc,\thcm,\thcmn\in\left\{ 0,\pi\right\} $ are admissible, i.e.,
they satisfy $\thc+\thcm+\thcmn\in\left\{ 0,2\pi\right\} $. Consider
$Y=\HcmnV(\thcmn)$ and $X=\nHcmV(\thcm)\oplus\HcV(\thc)$. We need
to prove that 
\[
\lambda_{j-1}(Y)\leq\lambda_{j}(X)\leq\lambda_{j+1}(Y)
\]
and that the inequalities are strict whenever $\lambda_{j}(X)$ is
a simple eigenvalue of $X$. First, observe that by construction $X$
and $Z:=\Hcmn^{\oplus}(\theta_{[\co,m]},\thc)$ share the same eigenvalues
(with same multiplicities). Thus, the claimed inequalities follow
directly from Corollary~\ref{cor: Rank2_Pert} and Lemma~\ref{lem: Pert_Hamil}.
It is left to prove that those inequalities are strict if $\lambda_{j}(X)=\lambda_{j}(Z)$
is simple.

Let $n\in\N$ and borrow the notation of Lemma~\ref{lem: Pert_Hamil}
for $q_{1},q_{2},q_{3}\in\N$ and $d_{1},d_{2}\in\N$. Following Lemma~\ref{lem: Pert_Hamil},
there are $x:=x(\thc,\thcm,\thcmn),~y:=y(\thc,\thcm,\thcmn)\in\R^{q_{3}}$
such that $Y-Z=xx^{t}-yy^{t}$. Moreover, $x$ and $y$ satisfy the
assertions (\ref{enu: Pert_Hamil - n even - first components}) and
(\ref{enu: Pert_Hamil - n even - second components}) in Lemma~\ref{lem: Pert_Hamil}.
Set 
\[
Z(x):=Z+xx^{t}\qquad\text{and}\qquad Z(y):=Z-yy^{t}.
\]
Then we have 
\[
Y=Z+xx^{t}-yy^{t}=Z(x)-yy^{t}=Z(y)+xx^{t}.
\]
Recall that $q_{3}=nq_{2}+q_{1}=d_{1}+d_{2}$ (Lemma~\ref{lem: periods of three approximants})
and $\nHcmV(\thcm)$ is an $nq_{2}\times nq_{2}$ matrix while $\HcV(\thc)$
is an $q_{1}\times q_{1}$.

(a) We prove $\lambda_{j}(Z)<\lambda_{j+1}(Y)$. Assume by contradiction
that $\lambda_{j}(Z)=\lambda_{j+1}(Y)$ holds and $\lambda_{j}(Z)$
is a simple eigenvalue of $Z$. Due to Proposition~\ref{prop: Rank1_Pert}
(using that $xx^{t}$ and $yy^{t}$ are positive semidefinite), the
previous identities lead to 
\[
\lambda_{j}(Z)\leq\lambda_{j+1}(Z(y))\leq\lambda_{j+1}(Y)\quad\text{and}\quad\lambda_{j}(Z)\leq\lambda_{j}(Z(x))\leq\lambda_{j+1}(Y).
\]
Thus, 
\[
\lambda:=\lambda_{j}(Z)=\lambda_{j}(Z(x))=\lambda_{j+1}(Z(y))=\lambda_{j+1}(Y).
\]
follows by our assumption.

Let $w\in\R^{q_{3}}\setminus\left\{ 0\right\} $ be an eigenvector
of $Z$ corresponding to the eigenvalue $\lambda$. Since $\lambda$
is a simple eigenvalue of $Z$, then either (1) $\lambda$ is an eigenvalue
of $\nHcmV(\thcm)$ or (2) $\lambda$ is an eigenvalue of $\HcV(\thc)$,
but not both. These two cases can be treated similarly using Lemma~\ref{lem: Pert_Hamil}.
We only prove here case (1).

Since $\lambda$ is an eigenvalue of $\nHcmV(\thcm)$ but not of $\HcV(\thc)$,
we conclude that the corresponding eigenvector $w$ of $Z$ is of
the form $w=(w_{1},\ldots,w_{nq_{2}},0,\ldots,0)^{t}\in\R^{q_{3}}$
if $k\in\N$ is odd (where $k$ is determined by the length of the
tuple $\co$) and $w=(0,\ldots,0,w_{1},\ldots,w_{nq_{2}})^{t}\in\R^{q_{3}}$
if $k\in\N$ is even. Set $u:=(w_{1},\ldots,w_{nq_{2}})^{t}$.

We claim that $\langle x,w\rangle=0=\langle w,y\rangle$ holds. Before
proving this identity, let us show how these equalities finish our
proof. If $\langle x,w\rangle=0=\langle w,y\rangle$, then $w_{1}=w_{nq_{2}}=0$
follow from Lemma~\ref{lem: Pert_Hamil}~(\ref{enu: Pert_Hamil - n even - first components})
if $k$ is odd and from Lemma~\ref{lem: Pert_Hamil}~(\ref{enu: Pert_Hamil - n even - second components})
if $k$ is even. Since $\nHcmV(\thcm)u=\lambda u$ and each equation
in the system involves three consecutive (going cyclically) of the
entries of $u\in\R^{nq_{2}}$, we derive $u=0$ and so $w=0$. This
is a contradiction as $w\neq0$ is an eigenvector of $Z$ for the
eigenvalue $\lambda$.

Now let us prove the claim $\langle x,w\rangle=0=\langle w,y\rangle$.
Since $\lambda=\lambda_{j}(Z)=\lambda_{j}(Z(x))$, there is an eigenvector
$v$ of $Z+xx^{t}$ with eigenvalue $\lambda$. Using that $Z+xx^{t}$
is hermitian and $x^{t}w=\langle x,w\rangle$, we conclude 
\[
\lambda\langle v,w\rangle=\langle v,(Z+xx^{t})w\rangle=\langle v,Zw\rangle+\langle x,w\rangle\langle v,x\rangle=\lambda\langle v,w\rangle+\langle x,w\rangle\langle v,x\rangle
\]
implying $\langle x,w\rangle\langle v,x\rangle=0$. If $\langle v,x\rangle\neq0$,
we immediately derive $\langle x,w\rangle=0$ as desired. If $\langle v,x\rangle=0$,
then 
\[
\lambda v=(Z+xx^{t})v=Zv+\langle x,v\rangle x=Zv
\]
follows. Thus, $v=Cw$ holds for some $C\in\mathbb{R}\setminus\{0\}$
as $\lambda$ is a simple eigenvalue of $Z$ with eigenvector $w$.
Hence, $\langle v,x\rangle=0$ leads to $\langle w,x\rangle=0$ as
claimed.

Similarly, we conclude $\langle w,y\rangle=0$ using that $\lambda$
is an eigenvalue of $Z(y)$. 

(b) Similarly to case (a), we can prove $\lambda_{j-1}(Y)<\lambda_{j}(Z)$.
Assume by contradiction that $\lambda_{j-1}(Y)=\lambda_{j}(Z)$ holds
and $\lambda_{j}(Z)$ is a simple eigenvalue of $Z$. Then Proposition~\ref{prop: Rank1_Pert}
leads to 
\[
\lambda_{j-1}(Y)\leq\lambda_{j-1}(Z(x))\leq\lambda_{j}(Z)\quad\text{and}\quad\lambda_{j-1}(Y)\leq\lambda_{j}(Z(y))\leq\lambda_{j}(Z).
\]
Thus, our assumption yields 
\[
\lambda:=\lambda_{j}(Z)=\lambda_{j-1}(Z(x))=\lambda_{j}(Z(y))=\lambda_{j-1}(Y).
\]
Next, let $w\in\R^{q_{3}}\setminus\left\{ 0\right\} $ be an eigenvector
of $Z$ for the eigenvalue $\lambda$. By simplicity of the eigenvalue
$\lambda=\lambda_{j}(Z)$, $\lambda$ is either an eigenvalue of $\nHcmV(\thcm)$
or of $\HcV(\thc)$. Thus, $w$ has either the form $w=(w_{1},\ldots,w_{d_{1}},0,\ldots,0)^{t}\in\R^{q_{3}}$
or $w=(0,\ldots,0,w_{1},\ldots,w_{d_{2}})^{t}\in\R^{q_{3}}$. As before
one can show that in both cases $\langle x,w\rangle=0=\langle w,y\rangle$
holds. Then Lemma~\ref{lem: Pert_Hamil} yields $w_{1}=w_{d_{1}}=0$
(respectively $w_{1}=w_{d_{2}}=0$) and so $w=0$ follows, a contradiction.
\end{proof}

\end{document}